%% file: paper.tex
\documentclass[acmsmall,screen]{acmart}

\usepackage{booktabs}   
\usepackage{subcaption} 

\makeatletter\if@ACM@journal\makeatother
\acmJournal{PACMPL}
\startPage{1}
\else\makeatother
\acmYear{2018}
\startPage{1}
\fi

\setcopyright{none}             

\usepackage{booktabs}   
\usepackage{subcaption} 
\usepackage{listing}

\setcitestyle{round}

\usepackage{amsthm,mathtools,amssymb,proof,stmaryrd}
\usepackage[linesnumbered,ruled]{algorithm2e}
\usepackage{listings}
\usepackage{galois}
\usepackage{isomath}
\usepackage{multirow}
\usepackage{prooftree}
\usepackage{textcomp}
\usepackage{amscd}
\usepackage{amsfonts}
\usepackage{varwidth}
\usepackage{graphicx}
\usepackage{listings}
\usepackage{float}
\usepackage{multirow}
\usepackage[noend]{algorithmic}
\usepackage{mathrsfs}
\usepackage{mathpartir}
\usepackage{stmaryrd}
\usepackage{url}
\usepackage{hyperref}
\usepackage{xspace}
\usepackage{verbatim}
\usepackage{lipsum}
\usepackage{wrapfig}

\usepackage{enumitem}
\usepackage{url}
\usepackage{fancyvrb}
\usepackage[figurename=Figure]{caption}
\usepackage{xcolor}

\SetAlFnt{\small}
\SetAlCapFnt{\small}
\SetAlCapNameFnt{\small}
\SetAlCapHSkip{0pt}
\IncMargin{-\parindent}

\floatstyle{plaintop}
\restylefloat{table}

\input{defs}

\setlist[itemize]{leftmargin=*}
\setlist[enumerate]{leftmargin=*}

\begin{document}

\hypersetup{colorlinks,
  linkcolor=ACMDarkBlue,
  citecolor=ACMPurple,
  urlcolor=ACMDarkBlue,
  filecolor=ACMDarkBlue}

\setlength{\pdfpageheight}{\paperheight}
\setlength{\pdfpagewidth}{\paperwidth}

\title[A True Positives Theorem for a Static Race Detector]
{A True Positives Theorem for a Static Race Detector\\Extended Version}

\author{Nikos Gorogiannis}
\affiliation{%
  \institution{Facebook}
  \country{UK}
}
\affiliation{%
  \institution{Middlesex University London}
  \country{UK}
}
\email{nikosgorogiannis@fb.com}

 \author{Peter W. O'Hearn}
\affiliation{%
  \institution{Facebook}
  \country{UK}
}
\affiliation{%
  \institution{University College London}
  \country{UK}
}
\email{peteroh@fb.com}

\author{Ilya Sergey} \authornote{Work done while employed as a
  part-time contractor at Facebook.}
\affiliation{%
  \institution{Yale-NUS College}
    \country{Singapore}
}
\affiliation{%
  \institution{National University of Singapore}
    \country{Singapore}
}
\email{ilya.sergey@yale-nus.edu.sg}

\begin{abstract}
\input{abstract}
\end{abstract}

%
%

%


\maketitle


\input{intro}
\input{overview}
\input{language}
\input{analysis}
\input{completeness}
\input{reconstruction}
\input{evaluation}
\input{related}
\input{conclusion}
\bibliography{references,proceedings}

\newpage
\appendix
\input{appendix-comp}
\input{appendix-reconstr}

\end{document}

%% file: defs.tex
\definecolor{shadecolor}{gray}{1.00}
\definecolor{ddarkgray}{gray}{0.75}
\definecolor{darkgray}{gray}{0.30}
\definecolor{light-gray}{gray}{0.87}

\newtheorem*{conjecture*}{Conjecture}

\newcommand{\etc}{\emph{etc}}
\newcommand{\ie}{\emph{i.e.}\xspace}
\newcommand{\Ie}{\emph{I.e.}\xspace}
\newcommand{\eg}{\emph{e.g.}\xspace}

\newcommand{\Wlog}{\emph{w.l.o.g.}\xspace}

\newcommand{\cf}{\textit{cf.}\xspace}
\newcommand{\wrt}{\emph{wrt.}\xspace}
\newcommand{\Iff}{\emph{iff}\xspace}

\newcommand{\yada}{Detailed proof is in the supplementary
  material.\xspace}

{
\theoremstyle{plain}
{\unskip\nobreak\hskip 1em plus 1fil\nobreak$\square$
\parfillskip=0pt%
\endtrivlist}
}

\newcommand{\racerd}{\tname{RacerD}}
\newcommand{\racerdx}{\tname{RacerDX}}

\newcommand{\chord}{\tname{Chord}}

\makeatletter
\newcommand{\dotminus}{\mathbin{\text{\@dotminus}}}

\newcommand{\@dotminus}{%
  \ooalign{\hidewidth\raise1ex\hbox{.}\hidewidth\cr$\m@th-$\cr}%
}
\makeatother

\newcommand{\pow}[1]{\wp({#1})}
\newcommand{\set}[1]{\left\{{#1}\right\}}
\newcommand{\defeq}{=_{\textrm{\scriptsize{def}}}}

\newcommand{\partialfn}{\rightharpoonup}
\newcommand{\finpartialfn}{\partialfn_{\textrm{\tiny fin}}}

\newcommand{\wb}{\mathsf{WB}(\e_{1 \dots n})}
\newcommand{\wcup}{\cup_1}

\newcommand{\nil}{\mathsf{nil}}

\newcommand{\dom}[1]{\mathrm{dom}\left(#1\right)}

\newcommand{\sem}[1]{\llbracket#1\rrbracket}

\newcommand{\loc}{{\sf Loc}}
\newcommand{\stacks}{{\sf Stack}}
\newcommand{\heaps}{{\sf Heap}}
\newcommand{\progstates}{\mathsf{State}}
\newcommand{\locks}{{\sf Locks}}

\newcommand{\Addr}{\mathsf{Addr}}

\makeatletter
\newcommand{\removelatexerror}{\let\@latex@error\@gobble}
\makeatother

\definecolor{pblue}{rgb}{0.13,0.13,1}
\definecolor{pgreen}{rgb}{0,0.5,0}
\definecolor{pred}{rgb}{0.9,0,0}
\definecolor{pgrey}{rgb}{0.46,0.45,0.48}

\definecolor{ckeyword}{HTML}{7F0055}
\definecolor{ccomment}{HTML}{3F7F5F}
\definecolor{cnumber}{HTML}{2A0099}

\lstdefinelanguage{Java}{
  keywords={new, return, public, this, if, else, then, private,
    push, pop, lock, unlock, while,
    synchronized, final, class, new, protected, extends, implements, import,
    volatile, null, static},
  ndkeywords={bool, int, void},
  showspaces=false,
  showtabs=false,
  breaklines=true,
  showstringspaces=false,
  breakatwhitespace=true,
  lineskip=-0.6pt,
  morecomment=[l]{//}, 
  morecomment=[s]{/*}{*/}, 
  basewidth={0.54em, 0.4em},%
  basicstyle=\footnotesize\ttfamily,
  keywordstyle={\color{ckeyword}\ttfamily\bfseries},
  ndkeywordstyle={\color{pblue}\ttfamily\bfseries},
  commentstyle={\color{ccomment}\itshape},
  stringstyle=\color{green},
  moredelim=[il][\textcolor{pgrey}]{$$},
  numberstyle=\scriptsize\sf,
  xleftmargin=1em,
  moredelim=[is][\textcolor{pgrey}]{\%\%}{\%\%}
}

\lstdefinestyle{number}{
  numbers=left,
  numberstyle=\scriptsize\em,
  xleftmargin=1em
}

\newcommand{\jcode}[1]{\lstinline[language=Java,basicstyle=\small\ttfamily]{#1}}
\newcommand{\code}[1]{\jcode{#1}}

\newcommand{\many}[1]{\overline{#1}}

\newcommand{\Method}{\mathsf{Method}}
\newcommand{\Exp}{\mathsf{Exp}}

\newcommand{\Var}{\mathsf{Var}}

\newcommand{\Stmt}{\mathsf{Stmt}}
\newcommand{\Stmtr}{\mathsf{Stmt}^{\!*}}
\newcommand{\CStmt}{\mathsf{CStmt}}
\newcommand{\Program}{\mathsf{Program}}

\newcommand{\FieldName}{\mathsf{Field}}

\newcommand{\MethodName}{\mathsf{Method}}

\newcommand{\angled}[1]{\left\langle{#1}\right\rangle}

\newcommand{\astate}{\varsigma}

\newcommand{\mbody}{\mathsf{mbody}}

\newcommand{\Nat}{\mathbb{N}}

\newcommand{\eqdef}{\triangleq}

\newcommand{\e}{\mathsf{e}}
\newcommand{\f}{\mathsf{f}}
\newcommand{\g}{\mathsf{g}}
\newcommand{\fs}{\bar{\f}}
\newcommand{\gs}{\bar{\g}}

\newcommand{\m}{\mathsf{m}}
\newcommand{\body}[1]{\mbody(#1)}
\newcommand{\argm}{\mathtt{arg}}

\newcommand{\arity}[1]{\#(#1)}

\newcommand{\lval}[2]{\lfloor #1 \rfloor_{#2}}
\newcommand{\eval}[2]{\llbracket #1 \rrbracket_{#2}}
\newcommand{\locn}[1]{\mathsf{locn}({#1})}
\newcommand{\prefix}{\preccurlyeq}
\newcommand{\properprefix}{\prec}
\newcommand{\rootvar}[1]{\mathsf{root}(#1)}
\newcommand{\trunc}[1]{\mathsf{trunc}(#1)}
\newcommand{\pathimg}[3]{{#1}\!\!\downarrow^{#2}_{#3}}
\newcommand{\fix}[1]{\mathsf{fix}~{#1}}
\newcommand{\WF}{\mathcal{F}}

\newcommand{\Path}{\mathsf{Path}}

\newcommand{\tname}[1]{\textsc{#1}\xspace}
\newcommand{\pname}[1]{\textsf{#1}\xspace}

\newcommand{\cmd}[1]{\mathtt{#1}}
\newcommand{\cmdskip}{\cmd{skip}}
\newcommand{\cmdalloc}{{\cmd{new}}()}
\newcommand{\cmdif}[3]{\cmd{{if}}\,{#1}\;{{\cmd{then}}}\;{#2}\; {{\cmd{else}}}\;{#3}}
\newcommand{\cmdwhile}[2]{{{\cmd{while}}}\,{#1}\; {{\cmd{do}}}\;{#2}}
\newcommand{\cmdpush}[1]{{\cmd{push}}(#1)}
\newcommand{\cmdpop}[1]{{\cmd{pop}}(#1)}
\newcommand{\cmdlock}{{\cmd{lock}}()}
\newcommand{\cmdunlock}{{\cmd{unlock}}()}

\newcommand{\cons}{\dblcolon}
\newcommand{\snoc}{{\dblcolon}}

\newcommand{\incrlock}[1]{\mathsf{incr}(#1)}
\newcommand{\decrlock}[1]{\mathsf{decr}(#1)}
\newcommand{\addlock}[2]{\mathsf{add}(#1,#2)}

\newcommand{\accessessymb}{\mathsf{accs}}
\newcommand{\wobblysymb}{\mathsf{wobbly}}
\newcommand{\lockingsymb}{\mathsf{locking}}

\newcommand{\accesses}[2]{\accessessymb\; #1 \; #2}
\newcommand{\wobbly}[2]{\wobblysymb\; #1\; #2}
\newcommand{\locking}[2]{\lockingsymb\; #1\; #2}

\newcommand{\asem}[1]{\sem{#1}^{\#}}
\newcommand{\exec}{\mathsf{exec}{\text{\xspace}}}
\newcommand{\pre}{\sqsubseteq}
\newcommand{\tr}{\tau}
\newcommand{\Traces}{\mathcal{T}}
\newcommand{\TTraces}{\pow{\Traces}}
\newcommand{\Abs}{\mathcal{D}}
\newcommand{\Acc}{\mathcal{A}}

\newcommand{\Wob}{\mathcal{W}}
\newcommand{\Lock}{\mathcal{L}}
\newcommand{\ab}[1]{\widehat{#1}}
\newcommand{\lub}{\sqcup}

\newcommand{\Lub}{\bigsqcup}

\newcommand{\bota}{\bot_{\Abs}}

\newcommand{\last}[2]{\mathsf{last}(#1)}
\newcommand{\llast}[1]{\mathsf{last}(#1)}

\newcommand{\interl}[3]{\mathsf{interl}~#1~#2~\angled{#3}}
\newcommand{\append}{+\!\!\!+~}

\newtheorem{inneralemma}{Lemma}
\newenvironment{alemma}[1]
  {\renewcommand\theinneralemma{#1}\inneralemma}
  {\endinneralemma}

\newtheorem{inneratheorem}{Theorem}
\newenvironment{atheorem}[1]
  {\renewcommand\theinneratheorem{#1}\inneratheorem}
  {\endinneratheorem}

\newtheoremstyle{proofstyle} 
    {\topsep}                    
    {\topsep}                    
    {}                                
    {}                                
    {\itshape}                  
    {.}                              
    {.5em}                       
    {}  

\theoremstyle{proofstyle}

\newtheorem{innerproof}{}
\newenvironment{refproof}[1]
  {\renewcommand\theinnerproof{#1}\innerproof}
  {\endinnerproof}
\newcommand{\refp}[1]{the \lowercase{\emph{\ref{#1}}}}
\newcommand*{\QED}{\hfill\ensuremath{\Box}}%

\newcommand{\subst}[2]{{#1}^{#2}}
\newcommand{\fml}{\mathsf{fml}}
\newcommand{\fst}{\mathsf{fst}}
\newcommand{\ifext}[2]{#1}

%% file: abstract.tex

\racerd  is a static race detector that has been proven to be effective in engineering practice: it has seen thousands of data races fixed by developers before reaching production, and has supported the migration of Facebook's Android app rendering infrastructure from a single-threaded to a multi-threaded architecture.
We prove a True Positives Theorem stating that, under certain assumptions, 
an idealized theoretical version of the analysis {\em never reports a false positive\/}. We also provide an empirical evaluation of an implementation of this analysis, versus the original \racerd.

The theorem was motivated in the first case by the desire to understand the observation from production that \racerd was providing remarkably accurate signal to developers,
and then the theorem guided further analyzer design decisions.
Technically, our result can be seen as saying that the analysis computes an under-approximation of an over-approximation, which is the reverse of the more usual (over of under) situation in static analysis.
Until now, static analyzers that are effective in practice but unsound have
often been regarded as ad hoc; in contrast, we suggest
that, in the future, theorems of this variety might
be generally useful in understanding, justifying and designing effective
static analyses for bug catching.

%

%% file: intro.tex
\section{Context for the True Positives Theorem}
\label{sec:intro}

The purpose of this paper is to state and prove a theorem that
has come about by
reacting to surprising properties we observed of a   static program analysis that has been in production at Facebook for over a year.

The \racerd program analyzer searches for data races in Java programs, and it has had significantly more reported industrial impact than any other concurrency analysis that we are aware of.
It was released as open source in October of 2017, and the OOPSLA'18
paper by \citet{Blackshear-al:OOPSLA18} describes its design, and gives more details about its deployment. They report, for example, that over 2,500 concurrent data races found by \racerd have been fixed by Facebook developers, and that it has been used to support the conversion of Facebook's Android app rendering infrastructure from a single-threaded to a multi-threaded architecture.

\racerd's designers did not establish the formal properties of the analyzer, but it has been shown to be effective in practice. We wanted to understand this point from a theoretical point of view.  \racerd is not sound in the sense of computing an over-approximation of some abstraction of executions. Over-approximations support a theorem: if the analyzer says there are no bugs, then there are none. \Ie, there are no false negatives (when the program has no bugs), which is often considered as a ``soundness'' theorem.
\racerd favours reducing false positives over false negatives.  A design goal was to ``detect actionable races that developers find useful and respond to'' but ``no need to (provably) find them all''. In fact, it is very easy to generate artificial false negatives, but \citet{Blackshear-al:OOPSLA18} say that few have been reported in over a year of   \racerd in production.

One can react to this by saying that \racerd is simply an {\em ad hoc}\/, if effective, tool.  But, the tool does not heuristically filter out bug reports in order to reduce false positives. It first does   computations with an abstract domain, and then issues data race reports if any potential races are found according to the abstract domain. Its architecture is  like that of a principled analyzer based on abstract interpretation, \emph{even though it does not satisfy the usual soundness theorem}.
This suggests that saying \racerd is ad hoc because it does not satisfy a standard soundness theorem is somehow missing something: It would be better if a demonstrably-effective analyzer with a principled design came with a theoretical explanation, even if partial, for its effectiveness.  That is the research problem we set ourselves to solve in this work.

A natural question is if it is possible to actually modify \racerd so as to make it sound, without losing its effectiveness in generating \emph{signal}---actionable and useful data race reports that developers are keen to fix.
\racerd's initial design elided standard analysis techniques such as alias and escape analysis, on the grounds that this was consistent with the goal of reducing false positives.
To try to get closer to soundness our colleague Sam Blackshear, one \racerd's authors,  implemented an escape analysis to find race bugs due to locally declared references escaping their defining scope,  \ie, to reduce the false negatives.
The escape analysis led to too many false positives; it contradicted the goal of high signal, and was not put into production.
One of the current authors, Gorogiannis,  tried another tack to reduce false negatives: a simple alias analysis, to find races between distinct syntactic expressions which may denote the same lvalue.  Again the attempt caused too many false positives to make it to production.

Next we wondered: might there be a {\em different}\/ theorem getting to the heart of why \racerd works?
Because it is attempting to reduce false positives, a natural thing to try would be an \emph{under-approximation} theorem. This would imply that every report is a true positive. This theorem is false for the analyzer, because it is possible to artificially generate false positives.
One of the main reasons for this is that conditionals and loops are treated in a path-insensitive manner (\ie, join corresponds to taking the union of potential racy accesses across the different branches).

It seems plausible to modify \racerd to be under-approximate in two ways, but each of these has practical problems.  In the first way, one considers sets (disjunctions) of abstract states, and uses conjunction for interpreting \code{if} statements and disjunction at join points: this is like in symbolic execution for testing \cite{CadarS13}.  This would cause scaling challenges because of the path explosion problem. \racerd runs quickly, in minutes, on (modifications to) millions of lines of code, and its speed is important for delivering timely signal.
To make a much slower analysis would not be in the spirit of \racerd, and would not explain why the existing analyzer is effective.
The other way would be to use a meet (like intersection) while sticking with one abstract state per program point.  The problem here is that this prunes very many reports, so many that it would miss a great many bugs and (we reasoned)
would not be worth deploying in Continuous Integration.

These considerations led us to the following hypothesis, which we would like to validate:

\begin{conjecture*}[True Positives (TP) Theorem]
Under certain assumptions, the analyzer reports \emph{no false positives}.
\end{conjecture*}
An initial version of the assumptions was described as follows.
Consider an idealized language, IL, in which programs have only
\emph{non-deterministic} choice in conditionals and loops, and where
there no is recursion. Then the analyzer should only report true
positives for that language.

The absence of booleans in conditionals in IL reflects an analyzer
assumption that the code it will apply (well) to will use
coarse-grained locking and not (much) fine-grained synchronization,
where one would (say) race or not dependent on boolean conditions. In
Facebook's Android code we do find fine-grained concurrency. The
\pname{Litho} concurrent UI library\footnote{Available at
  \url{https://fblitho.com}.} has implementations of ownership
transfer and double-checked locking, both of which rely on boolean
conditions for concurrency control.
This kind of code can lead to false positives for \racerd, but we did not regard that as a mistake in \racerd. For, the Android engineers advised us to concentrate on coarse grained locking, which is used in the vast majority of product code at Facebook.
For instance, we rarely observed code calling into \pname{Litho} which selects a lock conditionally based on the value of a mutable field.  Thus, the TP theorem based on assumptions reflected in IL could, if true, be a way of explaining why the analyzer reports few false positives, even though it uses join for \code{if}-statements.

The no-recursion condition is there because we want to say that the analyzer gets the races right except for when divergence makes a piece of code impossible to reach. If a data race detector reports a bug on a memory access that comes after a divergent statement (\code{diverge; acc}) then this would be a false positive, but we would not blame   data race reporting, but rather failure to recognize the divergence which makes the memory access unreachable.  The no-recursion requirement is just a form of {\em separation of concerns}\/ to help theory focus on explaining the data race reporting aspect. Note that if non-deterministic choices are the only booleans in \code{while}-loops, then such loops do not necessarily loop forever; that is why we forbid recursion and not loops here. We give a full description of the assumptions for the TP theorem in Section \ref{sec:language}.

One can see our True Positives Theorem as establishing an
\emph{under-approximation} of an \emph{over-approximation} in a
suitably positioned \emph{context}. Start with a program without
recursion (context). Replace boolean values in conditionals and loop
statements by non-determinism (over-approximation). Finally, report
only true positives in that over-approximation (\ie, under-approximate
the potential bugs in that over-approximation).
The more usual position in program analysis is to go for the reverse
decomposition, an over-approximation of an under-approximation, to
account for soundness relative to assumptions. The ``under'' comes
about from ignored behaviors -- \eg, if an analyzer does not deal well
with a particular feature, such as reflection, prune those behaviours
involving reflection for stating the soundness property-- and the
``over'' comes from a usual sound-for-bug-prevention analysis, but
only \wrt this under-approximate model. In contrast, our notion of
\emph{under-of-over} seems like a good way to think about a static
analysis for bug catching.

We explained above how we arrived at the statement of the TP Theorem by considering properties of a particular analyzer, paired with considerations of its effectiveness. We don't claim that to get an effective race detector you {\em must}\/ end up in the position of the TP theorem.
It might be possible to obtain a demonstrably useful-in-practice data race detector that satisfies a standard soundness (over-approximation) theorem.
It might also be possible to find one satisfying a standard (unconditional) under-approximation theorem. In fact, there exist a number of dynamic data race detectors that
are very promising, and at least one~\cite{TSAN} that is widely deployed in practice. Both of these directions are deserving of further research.
The work here is not in conflict with these valuable research directions, but would simply complement them by providing new insights on designing \emph{static} race analyzers.

Now, the True Positives Theorem was not actually true of \racerd when we formulated it, and apparently it is still not. But, subsequent to our discovery of the theorem, the authors of \racerd took it as a guiding principle.
%
For example, they implemented a ``deep ownership'' assumption~\cite{Clarke-Drossopoulou:OOPSLA02}: if an access path, say \code{x.f}, is owned (\ie, accessible through just one object instance), then so are all extensions, \eg~\code{x.f.g}.
The analyzer would never report a race for an owned access.
Somewhat surprisingly, even though the deep ownership assumption goes against soundness for bug prevention (over-approximation), it is compatible with the goal of reducing false positives and commonly holds in practice.




We have proven a version of the TP Theorem for \racerdx, a modified version of \racerd analyzer, which is based on IL and is not too far removed from the original.
%
%
In this paper we formulate and prove the TP theorem for an idealized theoretical language, and we carry out an empirical evaluation of an implementation of \racerdx in relation to the in-production analyzer (\racerd), \wrt change in the the number of reports produced.
The distance between \racerd and \racerdx is such that the latter, which has a precise theorem attached to it, makes on the order of 10-57\% fewer data race reports on our evaluation suite.
To switch \racerd for \racerdx in production would require confidence  concerning amount of true bugs in their difference,
or other factors (\eg, simplicity of maintenance) which require engineering judgement; in general, replacing well-performing in-production software requires strong reasons.   But, our experiments show how the basic design of \racerd is not far from an analyzer satisfying a precise theorem and, as we now explain, the extent of their difference is not so important for the broader significance of our results.
%




\paragraph{Broader Significance}
Our starting point was a desire to understand theoretically why the specific analyzer \racerd is effective in practice, and
we believe that our results go some way towards achieving this, but it appears that they have broader significance.
They exemplify a way of studying program analysis tools designed for catching bugs rather than ensuring their absence.

The first point we emphasize is the following:
\begin{itemize}
\item [] {\em Unsound (and incomplete) static analyses can be principled, satisfying meaningful theorems that help to understand their behaviour and guide their design.}\/
\end{itemize}
The concept `sound' seems to have been taken as almost a synonym for `principled' in some branches of the research community, and colleagues we have presented our results to have often reacted with surprise that such a theorem is  possible.
Representative are the remarks of one of the POPL referees: `It is remarkable that the authors manage to prove any theorem at all about an unsound analysis' and the referees collectively who said `proving a theorem for an unsound static analysis is a first'.
Note that our our analysis is neither sound for showing the absence of
bugs nor for bug finding (every report is a true positive: the static
analysis often refers to this as `complete', where in symbolic as well
as concrete testing it is sometimes what is meant by `sound').
The potential for theoretical insights on unsound and incomplete analyses is perhaps less widely appreciated than it could be. Such analyses are not necessarily \emph{ad hoc} (although they can be).

Just to avoid misunderstanding, note that we are not talking about the
common situation of where an analyzer is unsound generally but
designed to be sound under assumptions. As described above, this
situation can be understood as soundness (over-approximation) \wrt a
different model than the usual concrete semantics, typically an
under-approximation of the concrete semantics. By striving for
minimizing false positives instead of false negatives, our analyzer is
purposely designed to be unsound, and this difference is reflected in
the {\em shape}\/ of our theorem being under-of-over rather than the
converse.

Much more significant than saying that {\em there exists}\/ an unsound (and incomplete) analysis with a theorem, would be if we could make the above claim for an analyzer that was useful in the sense of helping people.

Here, the fact that \racerd (the analyzer deployed in production, without a theorem) and \racerdx (its close cousin, with a theorem) are not the same  at first glance seems problematic. However, our
experiments suggest strongly that, if \racerdx rather than \racerd had been
put into production originally, it would have found thousands of bugs,
far outstripping the (publicly reported) impact of all previous static race detectors.
We infer:
\begin{itemize}
\item [] {\em One can have an unsound (and incomplete) but effective static analysis, which has significant industrial impact, and
which is supported by
a meaningful theorem; in our case the TP theorem.}
\end{itemize}
The discussion of this second point is admittedly based on counter-factual reasoning (what if we `had' deployed \racerdx instead of \racerd?), but the possibility of it being a false argument seems to be vanishingly small. The point is so powerful, and accepted by the Infer static analysis team members even outside the authors,\footnote{\racerd and \racerdx are implemented using the Infer.AI abstract interpretation framework; see \url{https://fbinfer.com}.} that the TP theorem is guiding the construction of new analysers at Facebook.

Thus, it seems that our results could be more broadly significant than the initial but perhaps worthwhile goal of `understand
why \racerd is effective'. The fact that \racerd and \racerdx don't
coincide is not so important for the larger point, though it would
have made for a neat story if they had. Such
neat stories have so far been rather rare when developing
analysers under industrial rather than scientific constraints.
Future  analysers, including ones being developed at Facebook, will possibly adhere more to the `neat' story.

\paragraph{Paper Outline}
In the remainder of the paper we give an overview of the intuition and
reasoning principles for identifying and reporting concurrent data
races in \racerd (Section~\ref{sec:overview}). We then describe an
idealized language (IL), whose set of features matches the common
conventions followed in production Java code, as well as IL's
over-approximating semantics (Section~\ref{sec:language}).
We then provide a formal definition of \racerdx, a modified \racerd analysis,
tailored for reporting no false positives, in the framework of abstract
interpretation (Section~\ref{sec:analysis}).
Our formal development culminates with a proof of the True Positives
Theorem, coming in two parts: in Section~\ref{sec:nfp-comp} we prove
completeness of \racerdx \wrt to its abstraction and in
Section~\ref{sec:nfp-reconstr} we show how to reconstruct provably
racy executions from the analysis results.
In Section~\ref{sec:evaluation}, we discuss the implementation of the
theoretical analyzer described in Section~\ref{sec:analysis}.
We then present an evaluation of \racerdx which compares it with
\racerd on a set of real-world Java projects.
We discuss related work in Section~\ref{sec:related} and elaborate on the
common formal guarantees considered for static concurrency analyses,
positioning \racerdx amongst existing tools and approaches.


%% file: overview.tex
\section{Overview}
\label{sec:overview}


A textbook definition of a data race is somewhat low-level: a race is
caused by potentially concurrent operations on a \emph{shared memory
  location}, of which at least one is a
write~\cite{Herlihy-Shavit:08}.

Data races in object-based languages with a language-provided
synchronization mechanism (\eg, Objective C or Java) can be described
more conveniently, for the sake of being understood and fixed by the
programmer, in terms of the program's \emph{syntax} (rather than
memory), via \emph{access paths}, which would serve as runtime race
``witnesses'', by referring to the dynamic semantics of concurrent
executions.\footnote{We will provide a more rigorous definition in
  Section~\ref{sec:language}.}
Given two programs (\eg, calls to methods of the same instance of a
Java class), $C_1$ and $C_2$, there is a \emph{data race} between
$C_1$ and $C_2$ if one can construct their concurrent execution trace
$\tr$, and identify two \emph{access paths}, $\pi_1$ and $\pi_2$
(represented as field-dereferencing chains $x.\f_1.\ldots,\f_n$, where
$x$ is a variable or \code{this}), in $C_1$ and $C_2$,
respectively, such that:
\begin{itemize}[leftmargin=1.5em]
\item[(a)] at some point of $\tr$ both $\pi_1$ and $\pi_2$ are involved
  into two concurrent operations, at least one of which is a write,
  while both $\pi_1$ and $\pi_2$ point to the same shared memory
  location;

\item[(b)] the sets of \emph{locks} held by $C_1$ and $C_2$ at
  that execution point respectively, are disjoint.
\end{itemize}

This ``definition'' of a data race provides a lot of freedom for
substantiating its components. Of paramount importance are ($i$)
the considered set of pairs of programs that can race, ($ii$) the
assumptions about the \emph{initial state} of $C_1$ and $C_2$,
and ($iii$) the notion of the \emph{dynamic semantics}, employed
for constructing concurrent execution traces.
A choice of those determines what is considered to be a race and
reflects some assumptions about program executions. For instance,
accounting for the ``worst possible configuration'' (\eg, arbitrary
initial state and uncontrolled aliasing) would make the problem of
sound (\ie, over-approximating) reasoning about data races
non-tractable in practice, or render its results too imprecise to be
useful.

\subsection{Race Detection in \racerd}

\begin{wrapfigure}[12]{r}{0.38\textwidth}
\begin{center}
\vspace{-14pt}
{\footnotesize{
\begin{lstlisting}[basicstyle=\scriptsize\ttfamily,numbers=left,xleftmargin=3em]
@ThreadSafe
public class Dodo {
  private Dodo dee;

  public void zap(Dodo d) {
    synchronized (this) {
      System.out.println(d.dee);
    }
  }
  public void zup(Dodo d) {
    d.dee = new Dodo();
  }
}
\end{lstlisting}
}}
\end{center}
\vspace{-10pt}
\caption{A racy Java class.}
\label{fig:racy}
\end{wrapfigure}
As a specific set of design choices \wrt identifying data races in
terms of access paths, let us consider
\racerd~\cite{Blackshear-al:OOPSLA18}---a tool for static
compositional race detection by Facebook---and its take on a toy example in
Figure~\ref{fig:racy}.
The Java class \code{Dodo} has one \code{private} field \code{dee} and
two \code{public} methods, both manipulating \code{dee}'s state.
To instantiate the definition of a data race above, \racerd considers
all pairs of public method calls of the \emph{same} class instance as
concurrently running programs $C_1$ and $C_2$ from the definition
above~($i$).
To instantiate the initial state, it assumes that method parameters
of the same type (\eg, \code{Dodo}) may alias, thus, maximizing a
possibility of a race ($ii$). Finally, it assumes there is only one
reentrant lock in the entire program (\eg, \code{Dodo}'s
\code{this}), and over-approximates the dynamic semantics by taking
\emph{all} branches of the conditionals and entering every loop~($iii$).

With this setup, \racerd reports the class \code{Dodo} (annotated with
\code{javax.annotation.concurrent}'s \code{@ThreadSafe}) as racy, due
to concurrent unsynchronized accesses to the field \code{dee} in two
paths: a read from \code{d.dee} in the method \code{zap()} and the
write to \code{d.dee} in \code{zup()}---both detected based on the
assumption that the two \code{d}s can be run-time aliases, since
they share the same type. This assumption makes it easy to reconstruct
a concurrent execution of the two culprit methods, invoked with the
same object, which exhibit a race, thus making this report a
\emph{true positive} of the analysis.

\subsection{Fighting False Positives}
By its nature, \racerd is a bug detector, hence it sacrifices the
traditional concept of \emph{soundness} (the property customary for static
analyses stating that \emph{all} behaviors of interest are
detected~\cite{Cousot-Cousot:POPL79}) and, thus, may suffer \emph{false
  negatives} (\ie, miss races).\footnote{Although there were very few reports of races missed
  in Facebook's production code~\cite{Blackshear-al:OOPSLA18}.}
Instead, for the purposes of correctly detecting bugs, and minimizing
the time programmers spend investigating the reports, \racerd
emphasizes \emph{completeness}~\cite{Ranzato:VMCAI13} over soundness,
targeting what we are going to refer to as the \emph{True Positives
  Conjecture}---that is, if it reports a data race, one should be able
to find a witness execution and access paths exhibiting the
concurrency bug, given the assumptions made.

\input{burble}
As a next example, consider the Java class \code{Burble} in
Figure~\ref{fig:burble}. For the same reason as with \code{Dodo}, the
methods \code{meps} and \code{reps} race with each other when run with
aliased arguments. It is less obvious, however, whether \code{meps}
races with \code{beps}---and in fact they do not! The reason for that
is that \code{beps} reassigns a freshly allocated instance of
\code{Bloop} to the formal \code{b} before assigning to the field of
the latter, thus, effectively \emph{avoiding} a race with a concurrent
access to \code{b.f} in \code{meps}.

This phenomenon of ``destabilizing'' an access path in a potentially
racy program can be manifested both intra-procedurally (as in
\code{Burble}) and inter-procedurally. To wit, in another example in
Figure~\ref{fig:wurble}, the class \code{Wurble} demonstrates a
similar instance of a \emph{false race}, with the \code{private}
method \code{zwup()} ``destabilizing'' the path \code{w.g} by
assigning a newly allocated \code{Wurble} instance to \code{w}, thus,
ensuring that \code{qwop()} and \code{gwap()} avoid a race
with each other.

A sound static analyzer would typically be expected to report races
in both of these examples, corresponding to a loss of
precision. However, having a non-negligible number of \emph{false
  positives} is not something a practical bug detector can afford.

To avoid this loss of signal effectiveness, \racerd employes an
\emph{ownership tracking domain}~\cite{Naik-al:PLDI06,
  Flanagan-Freund:PLDI09}, used to record the variables and paths that
have been assigned a newly allocated object, thus remedying the
situation shown above.

Upon closer examination, we found that \racerd's abstract domain,
including that of ownership, was not enough to allow us to prove that
an access path resolves to the same address, before and after
execution. To wit, knowing that an access path \code{x.f.g} is
\emph{not} owned, does not guarantee that the lvalue it corresponds to
stays the same during execution. The reason we wanted this latter
property is that it is one of the simplest ways to exhibit a race:
once we have set up an initial state where a path resolves to a
certain address, and have shown that execution of $C_1~||~C_2$ does not
modify that address (\ie, the path to address is \emph{stable}), we
are in the position to uncoditionally say that if both programs access
that address, they will race. The answer we came up with is that of
\emph{stability}; its negation, instability (or \emph{wobbliness}),
over-approximates ownership.
%

Thus we pose the question: can we state the reasonable (\ie,
non-trivial) conditions under which we can in confidence state (\ie,
formally prove) that \emph{all} of \racerd's reported races are
\emph{true positives}?

We refer to this desirable result as the \textbf{True Positives
  Theorem} (TPT) for a static race detector, and in this paper we
deliver such a theorem for a version of \racerd (called \racerdx, using stability),
formulating a set of assumptions under which it holds, and assessing
their practical implications and impact on signal.

\subsection{A True Positives Theorem for \racerdx}
\label{sub:tpt}

Our main result enabling the TPT proof is defining an
\emph{over-approximating} concrete semantics and stating the
sufficient conditions, both reflecting the behavior of production
code, under which \racerdx, a modified version of \racerd, reports
\emph{no false data races}.
The \racerdx True Positives Theorem, which substantiates this
statement, builds on the following three pieces of the formal
development that together form the central theoretical contribution of
this work.

\subsubsection{An Over-Approximating Concrete Semantics and Practical
  Assumptions}
\racerdx is a flow- and path-insensitive static analysis, which goes in all
branches of conditional expressions and loops.
To account for this design choice while stating a formal completeness
result, we adopt a novel non-deterministic (single-threaded)
trace-collecting semantics that treats branch and loop guards
as non-deterministically valued variables and explores all
execution sequences (Section~\ref{sec:language}). We use this
semantics to give meaning to single-threaded program executions, thus,
building a \emph{concrete domain} for the main \racerdx analysis
procedure.
Amongst other things, we assume just one global lock and a language
with a single class (although we do not restrict the number of its
fields and methods, as well as their signatures). We also restrict the
reasoning to programs with \emph{well-balanced locking} (\eg, in
Java terminology it would mean that only \code{synchronized}-enabled
locking is allowed), and forbid recursion.

\subsubsection{\emph{\racerdx} Abstraction and Sequential Completeness of its
  Analysis}
We formulate the abstract domain for \racerdx analysis
(Section~\ref{sec:analysis}) along with the abstraction function to it
from the concrete domain of the previously defined multi-threaded
trace-collecting semantics, \`{a} la~\citet{Brookes:TCS07}, which we
show to form a Galois connection~\cite{Cousot-Cousot:POPL79}.
We then prove, in Section~\ref{sec:nfp-comp}, a tower of lemmas,
establishing the ``sequential'' \emph{completeness} of \racerdx's
static analysis, which analyzes each sequential sub-program
compositionally \emph{in isolation}, without considering concurrent
interleavings, with respect to this abstraction
(Theorem~\ref{thm:completeness}).
The novelty of our formal proof is in employing standard Abstract
Interpretation, while taking a different perspective by establishing
the completeness rather than soundness of an analysis in the spirit of
the work by~\citet{Ranzato:VMCAI13}.

\subsubsection{Syntactic Criteria for Ensuring True Positives}
We introduce the notions of path \emph{stability} and its
counterpart, \emph{wobbliness} (\ie, instability) ---simple syntactic properties, which
are at the heart of stating the sufficient conditions for \racerdx's
TPT--- and connect it to the previously developed abstract domain and
single-threaded \racerdx abstraction.
A similar tower of lemmas is built, showing that a stable path resolves to
the same address before and after execution, thus providing us with the
tools for validating the reported race.
Our formal development culminates in Section~\ref{sec:nfp-reconstr},
with leveraging the sequential completeness result of \racerdx for
\emph{reconstructing} the concrete concurrent execution traces
\emph{exhibiting} provably true data races, thus delivering
the final statement and the proof of TPT (Theorem~\ref{thm:no-fp}).
%

\subsection{Measuring the Impact of True Positives Theorem on
  Signal Effectiveness}

The practical contribution of our work, described in
Section~\ref{sec:evaluation}, is an implementation of \racerdx, a
revised version of \racerd, incorporating the analysis machinery
enabling the result of TPT, and its evaluation.
We aimed to measure the impact of employing stability as a sufficient
condition for detecting true races with respect to the reduction in
overall number of reported bugs.\footnote{It is easy to have a vacuous
  TPT by reporting no bugs whatsoever!}
We ran experiments contrasting \racerd and \racerdx on a number of
open-source Java projects, ranging from 25k to 273k LOC.  We evaluated
the runtime of each analyser, and looked at the reports the two
analyzers produced in several ways, in order to assess the loss of
signal induced by \racerdx.






%% file: burble.tex
\begin{figure}[t]
\setlength{\belowcaptionskip}{-10pt}
\centering
\begin{tabular}{cc}
\begin{minipage}{0.45\linewidth}
\centering
{\footnotesize{
\begin{lstlisting}[basicstyle=\scriptsize\ttfamily,numbers=left,xleftmargin=2em]
class Bloop {
  public int f = 1;
}

class Burble {

  public void meps(Bloop b) {
    synchronized (this) {
      System.out.println(b.f);
    }
  }

  public void reps(Bloop b) {
    b.f = 42;
  }

  public void beps(Bloop b) {
    b = new Bloop();
    b.f = 239;
  }
}
\end{lstlisting}
}}
\vspace{8.5pt}
\caption{A Java class with a false race.}
\label{fig:burble}
\end{minipage}
&
\begin{minipage}{0.50\linewidth}
\centering
{\footnotesize{
\begin{lstlisting}[basicstyle=\scriptsize\ttfamily,numbers=left,firstnumber=22,   xleftmargin=2.5em,]
class Wurble {
  Wurble x = new Wurble();
  Bloop g = new Bloop();

  public void qwop(Wurble w) {
    zwup(w.x);
  }

  public void gwap(Wurble w) {
    synchronized (this) {
      System.out.println(w.x.g);
    }
  }

  private void zwup(Wurble w) {
    synchronized (this) {
      System.out.println(w.x.g);
    }
    w = new Wurble();
    w.g.f = 21;
  }
}
\end{lstlisting}
}}
\caption{A class with a false interprocedural race.}
\label{fig:wurble}
\end{minipage}
\end{tabular}
\end{figure}

%% file: language.tex
\section{Concrete Execution Model}
\label{sec:language}

To formally define our execution model and the notion of a race, we
first describe an idealized programming language (IL) that accurately
captures the essence of \racerdx's intermediate representation and
faithfully represents the race-relevant aspects of Java and similar
languages.

\subsection{Language and Assumptions about Programs}

We start by defining a programming language with a semantics suitable
for \racerdx's goals. Our language is simplified compared to Java,
with assumptions made to help our study of the question of whether the
data race reports are effective.

\begin{enumerate}[label=\textbf{A\arabic*}]
\item \label{asm:class} The language has only one class, that of
  record-like objects with an arbitrary, but finite set of fields
  (with names from a fixed set $\FieldName$) which are themselves
  pointers.

\item \label{asm:lock} The concurrency model is restricted to exactly
  two threads, and use a single, reentrant lock. The commands
  \code{lock()} and \code{unlock()} are only allowed to appear in
  \emph{balanced pairs} within block statements and method bodies.

\item \label{asm:assign} Local variables need no declaration, but must
  be assigned to before first use. Formal parameters are always taken
  from a fixed set. There are no global variables.

\item \label{asm:alloc} There is no destruction of objects, only
  allocation.

\item \label{asm:recur} All methods are non-recursive, have
  call-by-value parameters and no return value.  The assumption on non-recursion is
  standard, as we don't want to reason here about termination.


\item \label{asm:control} Control for conditional statements and while
  loops is \emph{non-deterministic}; there are no booleans.

\end{enumerate}

Assumptions \ref{asm:class} and \ref{asm:recur} are about ignoring potential sources of false positives which
have nothing in particular to do with races. For example, if you enter an infinite loop before a potential data race then you would have a false positive if you reported the race, but we wouldn't expect a (static) race detector to detect infinite loops. Similarly,
if you have a class that can't be inhabited you can't get races on it, but we view this as separate from the question of the effectiveness of the race reports themselves.

Assumption \ref{asm:recur} about recursion is perhaps not as practically restrictive as might first appear. Its impact is less than, for example, bounded symbolic model checking, which can be seen as performing a finite unwinding of a program to produce a non-recursive underapproximation of it  (and without loops as well) before doing analysis. As we explained in Section \ref{sec:intro}, our main reason for making the assumption is a conceptual rather than a practical one, to do with separation of concerns, but we state the point about bounded model checking  for additional context.

The reasons for the simplifications~\ref{asm:class} and~\ref{asm:lock}
are both related to \racerdx's focus on detection of races in one
class at a time, with races manifested by parallel execution of methods on a
single instance.
Assumption~\ref{asm:lock} is a potential source of real false negatives,
as methods of the same class that
use distinct locks can race.
Furthermore, the well-balanced assumption in~\ref{asm:lock}
corresponds directly to scoped synchronisation mechanisms like Java's
\code{synchronized(m)\{ \}} construction~\cite{jcip}, where \code{m}
is a static global mutex object.
\ref{asm:assign} corresponds to \racerdx's intermediate language
representation of parameters and variables.
The lack of global variables is a genuine restriction, but which is easy to address; we discuss
this from a practical point of view in Section~\ref{sec:evaluation}.
The assumption~\ref{asm:alloc} is due to the fact that Java is a
garbage-collected language.
%
%
%

Finally, the assumption \ref{asm:control} is the most significant one,
both from the perspective of formalising a language semantics, and
from the point of specifying completeness. By assuming
that \emph{every} execution branch may be taken, we do not
have to reason statically about branching and looping conditions. This
is effectively an \emph{over-approximation} of the actual semantics of
the concurrent programs.
This assumption is motivated by two considerations:
(i) \racerd  purposely avoids tackling fine-grained concurrency, and is applied at Facebook to code where
developers do not often avoid races by choosing which branch to take;
(ii) \racerd uses join for \code{if} statements in order not to have too many false negatives.


\begin{figure}[t]
\[
\begin{array}{@{}rcl@{}}
\f \in \FieldName & & \multicolumn{1}{l}{\text{field names}}
\\[2pt]
x,\argm_i \in \Var & & \multicolumn{1}{l}{\text{variables}}
\\[2pt]
\pi \in \Path & ::= & x.\f \mid \pi.\f
\\[2pt]
\e \in \Exp & \eqdef & \Var \cup \Path
\\[2pt]
c \in \Stmt &::= &
\cmdskip \mid x := x \mid x := \pi \mid \pi := x \mid x := \cmdalloc \mid \cmdlock \mid \cmdunlock
\mid \cmdpop{}
\\[2pt]
C \in \CStmt &::= & c \mid C;c \mid C; \cmdif {*} {C}{C} \mid C; \cmdwhile{*}{C} \mid C; \m(\e,\ldots,\e)
\\[2pt]
M \in \Method &::= & \m(\argm_1,\ldots, \argm_n) \set{~C~}
\\[2pt]
p \in \Program &::=&  C \parallel C
\end{array}
\]
\caption{Syntax grammar of the concurrent programming language.}
\label{fig:grammar}
\end{figure}

The grammar defining the language of interest is given in Figure~\ref{fig:grammar}.
We represent expressions $\e\in\Exp$ as either program variables and
method formals $x \in \Var$, and access paths
$\pi \in \Path = \Var\times\FieldName^+$.
The language contains no constants such as, \eg, \code{null}.
We partition statements into \emph{simple} and \emph{composite}.
Simple statements include assignments to variables, paths, reading from
paths, allocating a new object, (blocking) lock acquisition via $\cmdlock{}$,
releasing a lock via $\cmdunlock{}$, and popping an execution stack.
The command $\cmdpop{}$ is not to be used directly in programs; it
only occurs in our semantics of method calls and is needed for
defining the trace-collecting semantics described below.
Composite statements provide support for sequential composition
($;$), as well as conditionals
($\cmdif{*}{\cdot}{\cdot}$),  loops ($\cmdwhile{*}{\cdot}$) and method
calls, with the latter resulting in pushing a new stack frame on the call stack, as
well as emitting a $\cmdpop{}$ command to remove it at runtime after
the method body is fully executed.

Following the tradition of Featherweight Calculi for object-oriented
languages~\cite{Igarashi-al:TOPLAS01}, we introduce the function
$\body{\cdot}: \MethodName \to \CStmt$, which maps method names to bodies and
$\arity{\cdot} : \MethodName \to \Nat$, which gives the number of formal
parameters.
We impose the convention that the formal parameters are always
$\argm_1,\ldots\argm_n$ for all methods, where $n=\arity{\m}$ and that
no other variables of the form $\argm_j$, where $j>n$, appear in
$\body{\m}$. We also allow local variables (\ie variables $x$ which are not of the form $\argm_i$)
but forbid their use without definite prior initialisation.

\subsection{Concrete Semantics}
\label{sec:semantics}

In this work, we do not tackle fine-grained
concurrency~\cite{Turon-al:POPL13}, and we only consider lock-based
synchronisation, as per assumption~\ref{asm:lock}. Therefore, our
concurrent semantics adopts the model of sequential
consistency~\cite{Lamport:TC79}.
To define it, we start by giving semantics to sequential program runs,
which we will later combine to construct concurrent executions.

Our definition of runtime executions works over the following semantic
categories.
$\loc$ denotes a countably infinite set of object \emph{locations}. A
stack\footnote{We overload the term \emph{stack} to mean store here, borrowing from Separation Logic, the style of which informs most of our development.  We disambiguate the term by explicitly using \emph{call-stack} for a list of stacks.}
$s\in\stacks$ is a mapping $s : \Var \to \loc$.
Addresses $\Addr$ are defined in a field-splitting style, as
$\Addr = \loc \times \FieldName$. A heap $h\in\heaps$ is a partial,
finite map $h : \Addr \to_{\text{fin}} \loc$. The constant $\nil$ is such that
no heap is ever defined on an address of the form $(\nil,\f)$.
We write $s_\nil$ for the stack such that $s(x)=\nil$ for all
$x \in \Var$.\footnote{Here, $\nil$ is simply an unallocated address
  (corresponding to Java's \texttt{null}). We introduce it in order to
  avoid deviating too much from standard developments of shape
  analyses.}
The projection of an address to its first component is $\locn{\cdot}$,
\ie, $\locn{\ell,\f} = \ell$. We lift this definition to sets of
addresses, $\locn{A} = \{ \locn{\alpha} \mid \alpha\in A\}$, and to
heaps, $\locn{h} = \locn{\dom{h}}$. For example,
$\locn{\{(\ell_1,\f)\mapsto \ell_2 ; (\ell_2,\g)\mapsto \ell_3\}} =
\{\ell_1,\ell_2\}$.
A lock context $L\in\locks=\Nat$ is a natural number.

A (single-thread) \emph{program state}
$\astate\in\progstates = \Stmtr
\times\stacks^+\times\heaps\times\locks$ is a tuple
$\angled{c', S, h, L}$, where $c' \in \Stmtr$ is a either simple
statement $c \in \Stmt$ (Figure~\ref{fig:grammar}) or a \emph{runtime}
call-statement $\cmdpush{e_1, \ldots, e_n}$;
$S\in\stacks^+$ models the call-stack,
$h$ is a heap, and $L\in\locks$.
The stack at head position in the list $S$ is the current stack frame.
The $L$ component is the lock context of the thread,
signifying how many times a $\cmdlock$ instruction has been executed
without a corresponding $\cmdunlock$. We remark that the stack, heap
and lock state components are those \emph{produced by executing $c$
  starting at some previous program state}.

A \emph{two-threaded} program state is a tuple
$\angled{c_{\parallel}, (S_1,S_2), h, (L_1,L_2)}$, where
$c_{\parallel}$ is either $c\parallel\epsilon$ or
$\epsilon\parallel c$, denoting one of the two threads executing
command $c$. The pairs $(S_1,S_2)$ and $(L_1,L_2)$ are the
thread-local call-stacks and lock contexts for each thread, and $h$ is
the shared heap.

For the next series of definitions, we overload the term \emph{state}
to describe stack-heap pairs for both single-threaded and
multi-threaded executions when the context is unambiguous.

\begin{definition}
\label{def:address}
  The \emph{address}, $\lval{\pi}{s,h}$, of a path $\pi$ in a state
  $(s, h)$ is recursively defined as follows:
\[
\lval{x.\f}{s,h} \eqdef (s(x), \f) \qquad\qquad
\lval{\pi'.\f}{s,h} \eqdef (h(\lval{\pi'}{s,h}), \f)
\]
The \emph{value}, $\eval{\e}{s,h}$, of an expression $\e$ in a state $s,h$
is defined as follows:
\[
\eval{x}{s,h} \eqdef s(x) \qquad\qquad
\eval{\pi}{s,h} \eqdef h(\lval{\pi}{s,h})
\]
\end{definition}
The address of a path is the address read or written when a load or a
store accesses that path.
%

\begin{definition}[Execution trace]
\label{def:trace}
A (single-threaded) \emph{execution} trace is a possibly empty
list $\tr=[\astate_0, \ldots, \astate_n]$ of program states.
The set of all traces is $\Traces=\progstates^*$.
\end{definition}

We are now equipped with all the necessary formal components to give
meaning to both single- and multi-threaded executions.

\begin{figure}[t]
\setlength{\belowcaptionskip}{-10pt}
\centering
\begin{align*}
  \sem{\cmdskip}~\angled{S, h, L} & \eqdef \{\epsilon\}
  \\
  \sem{x:=\pi}~\angled{s \cons S, h, L} & \eqdef
  \begin{cases}
    \emptyset & \text{ if } \lval{\pi}{s,h} \notin\dom{h} \\
    \set{\left[\angled{x:=\pi, s[x \mapsto h(\lval{\pi}{s,h})] \cons S, h, L}\right]}
    & \text{ otherwise}
  \end{cases}
  \\
  \sem{\pi := x}~\angled{s \cons S, h, L} & \eqdef
  \begin{cases}
    \emptyset & \text{ if } \lval{\pi}{s,h} \notin\dom{h} \\
    \set{\left[\angled{\pi := x, s \cons S, h[\lval{\pi}{s,h} \mapsto s(x)], L}\right]}
    & \text{ otherwise}
  \end{cases}
  \\
  \sem{x:=\cmdalloc}~\angled{s \cons S, h, L} & \eqdef
  \left\{\left[\angled{x:=\cmdalloc, s' \cons S, h', L}\right] \;\middle|\;
    \begin{array}{l}
      \ell \notin \locn{h}, s' = s[x\mapsto\ell], \\
      h' = h \cup \bigcup_{\f\in\FieldName} \set{ (\ell,\f) \mapsto \ell}
    \end{array}\right\}
  \\
  \sem{x:=y}~\angled{s \cons S, h, L} & \eqdef
  \set{\left[\angled{x := y, s[x \mapsto s(y)] \cons S, h, L}\right]}
  \\
  \sem{\cmdlock}~\angled{S, h, L} & \eqdef
  \set{\left[\angled{\cmdlock , S, h, \addlock{L}{1}}\right]}
  \\
  \sem{\cmdunlock}~\angled{S, h, L} & \eqdef
  \begin{cases}
  \emptyset & \text{ if } L \leq 0 \\
    \set{\left[\angled{\cmdunlock{}, S, h , L - 1}\right]}
    & \text{ otherwise}
  \end{cases}
  \\
  \sem{\cmdpop{}}~\angled{s \cons S, h, L} & \eqdef
  \set{\left[\angled{\cmdpop{}, S, h, L}\right]}
\end{align*}
\hrule
\begin{align*}
  \sem{C; \m(\e_1,\ldots,\e_n)}~\angled{S, h, L} & \eqdef
  \begin{cases}
    \emptyset \qquad
      \begin{array}{l}
        \text{ if for some } i\le n, \eval{\e_i}{s, h'}
        \text{ is undefined for } \\
        \angled{s \cons S', h', L'} = \llast{\sem{C}~\angled{S, h, L}}
      \end{array}
    \\
    \left\{
    \tr \append  \tr_{\text{push}} \append \tr'
    ~\middle|
      \begin{array}{l@{}}
        \tr_{\text{push}} = [\angled{\cmdpush{\e_1,\ldots,\e_n}, \hat s
        \cons s \cons S', h', L'}], \\[1ex]
        \tr \in \sem{C}~\angled{S, h, L}, 
        \angled{s \cons S', h', L'} = \last{\tr}, \\[1ex]
        \tr' \in \sem{\body{\m}; \cmdpop{}}\angled{\hat s \cons s \cons S', h', L'} \\[1ex]
        \hat s = s_{\nil}[\argm_1 \mapsto \eval{\e_1}{s,h}]\cdots[\argm_n\mapsto\eval{\e_n}{s,h}]
      \end{array}
    \right\}
  \end{cases}
\end{align*}
\begin{align*}
  \sem{C; c}~\angled{S, h, L} & \eqdef
  \left\{
  \tr \append \tr'
  ~\middle|
    \begin{array}{l@{}}
      \tr \in \sem{C}~\angled{S, h, L},
      \angled{S', h', L'} = \last{\tr}{\angled{S, h, L}}, \\
      \tr' \in \sem{c}\angled{S', h', L'}
    \end{array}
  \right\}
  \\
  \sem{C; (\cmdif{*}{C_1}{C_2})}~\angled{S, h, L} & \eqdef
  \left\{
  \tr \append \tr'
  ~\middle|
    \begin{array}{l@{}}
      \tr \in \sem{C}~\angled{S, h, L},
      \angled{S', h', L'} = \last{\tr}{\angled{S, h, L}}, \\
      \tr' \in \sem{C_1}\angled{S', h', L'}  \cup \sem{C_2}\angled{S', h', L'}
    \end{array}
  \right\}
  \\
  \sem{C; (\cmdwhile{*}{C'})}~\angled{S, h, L} & \eqdef \fix{\WF}~(\sem{C}~\angled{S, h, L}), \text{ where}
  \\
  \WF & \eqdef \lambda T.~
  T \cup
  \left\{ \tr \append \tr' ~\middle|~
      \begin{array}{l}
        \tr \in T, \angled{S', h', L'} = \last{\tr}{}, \\
        \tr' \in \sem{C'}~\angled{S', h', L'}
      \end{array}
  \right\}
\end{align*}

\caption{Single-threaded trace-collecting semantics of simple (top) and
  compound statements (bottom).}
\label{fig:traces}
\end{figure}

\begin{definition}[Sequential trace-collecting semantics]
\label{def:tcs}
A trace-collecting semantics of a single-threaded (possibly compound)
program $C$, denoted
$\sem{C} : \stacks^+ \times \heaps \times \locks \to \TTraces$ is a map to a
set of traces starting from an initial configuration
$\angled{S, h, L}$. The trace-collecting semantics are defined in
Figure~\ref{fig:traces}.\footnote{For the sake of uniformity, we
  assume that every program has the form $\cmdskip;C$.}
The auxiliary function $\last{\tr}{}$ (defined only on non-empty traces) returns the
triple $\angled{S', h', L'}$ if the last element of the non-empty trace $\tr$
is $\angled{c, S', h', L'}$ for some command $c$.
\end{definition}

We give the semantics for concurrent programs in the style
of~\citet{Brookes:TCS07}.

\begin{definition}[Concurrent execution trace]
\label{def:concurrent-trace}
A (two-threaded) \emph{concurrent execution} trace is a possibly empty
list $\tr^{||}=[\astate^{||}_0, \ldots, \astate^{||}_n]$ of two-threaded program states.
The set of all concurrent traces is $\Traces^{||}$.
\end{definition}

\begin{definition}[Concurrent trace-collecting semantics]
\label{def:ctraces}
The trace-collecting semantics of a parallel composition of programs
$C_1$ and $C_2$ is defined in Figure~\ref{fig:ctraces}.
It is a map from two-threaded program states to sets of concurrent traces:
$\sem{C_1\parallel C_2}: (\stacks^+ \times \stacks^+) \times \heaps \times (\locks \times \locks) \to \pow{\Traces^{||}}$.
Intuitively, it interleaves all single-threaded executions of $C_1$ with those of $C_2$ taking care to guarantee that only one thread can hold the lock at a given step of the trace'.
\end{definition}

\begin{figure}[t]
  \begin{gather*}
    \begin{aligned}
      \sem{C_1\parallel C_2}~\angled{(S_1,S_2),h,(L_1,L_2)} & \eqdef
      \bigcup_{\substack{\tau_1 \in \sem{C_1}~\angled{S_1,h,L_1} \\ \tau_2 \in \sem{C_2}~\angled{S_2,h,L_2}}}
      \interl{\tau_1}{\tau_2}{(S_1,S_2),h,(L_1,L_2)}, ~\text{where} \\
      \interl{\epsilon}{\epsilon}{(S_1,S_2),h,(L_1,L_2)} &\eqdef \{\epsilon\} \\
      \interl{\tau_1}{\tau_2}{(S_1,S_2),h,(L_1,L_2)} & \eqdef \set{\epsilon} \cup \\
    \end{aligned} \\
    \left(\;\begin{split}
        \set{\angled{C_1\parallel\epsilon,(S_1,S_2),h,(L_1,L_2)} \mid \tau_1 = \angled{C_1,\_,\_,\_}\cons \_} \cup \\
        \set{\angled{\epsilon\parallel C_2,(S_1,S_2),h,(L_1,L_2)} \mid
        \tau_2 = \angled{C_2,\_,\_,\_}\cons \_ } \cup \\
        \left\{
          \angled{C_1\parallel\epsilon,(S_1,S_2),h,(L_1,L_2)}\cons \tau
        \left|
           \begin{array}{l}
           \tau_1 = \angled{C_1,\_,\_,\_}\cons \hat{\tau_1}\; \land \\[3pt]
           \tau \in \interl{\hat{\tau_1}}{\tau_2}{(S'_1,S_2),h',(L'_1,L_2)}\; \land \\[3pt]
           \angled{C_1,S'_1,h',L'_1} \in \sem{C_1}~\angled{S_1,h,L_1} \\[3pt]
           L_2=0 \lor L_1=L'_1=0
         \end{array}
        \right. \right\} \cup \\
        \left\{
         \angled{\epsilon\parallel C_2,(S_1,S_2),h,(L_1,L_2)}\cons \tau
           \left|
           \begin{array}{l}
           \tau_2 = \angled{C_2,\_,\_,\_}\cons \hat{\tau_2}\; \land \\[3pt]
           \tau \in \interl{\tau_1}{\hat{\tau_2}}{(S_1,S'_2),h',(L_1,L'_2)}\; \land \\[3pt]
           \angled{C_2,S'_2,h',L'_2} \in \sem{C_2}~\angled{S_2,h,L_2} \\[3pt]
           L_1=0 \lor L_2=L'_2=0
         \end{array}
        \right. \right\} \hphantom{\cup}
    \end{split}\;\right)
  \end{gather*}

\caption{Concurrent trace-collecting semantics.}
\label{fig:ctraces}
\end{figure}

We conclude this section of definitions by formally specifying
concurrent data races.

\begin{definition}[Data Race]
  \label{def:race}
  The program $C_1 \parallel C_2$ \emph{races}
  if there exists a state $\astate_0$ and a
  non-empty concurrent trace $\tr\in\sem{C_1\parallel C_2}~\astate_0$ such that
  $\llast{\tr} = \angled{\_, (s_1\cons \_, s_2\cons \_), h, \_)}$\footnote{
  Throughout the paper we use the notation $\_$ for we don't care about, effectively existentially quantifying them.
  }
   and,
  \begin{itemize}
    \item there exist paths $\pi_1,\pi_2$ such that $\lval{\pi_1}{s_1,h}=\lval{\pi_2}{s_2,h}$;
    \item there exist states $\astate_1=\angled{c_1\parallel\epsilon, \_,\_,\_}$ and $\astate_2=\angled{\epsilon\parallel c_2, \_,\_,\_}$ such that $\tr\snoc\astate_1,\tr\snoc\astate_2\in\sem{C_1\parallel C_2}~\astate_0$;
    \item
      $c_1 = (\pi_1:=\_) \land c_2 = (\pi_2:=\_)$, or,
      $c_1 = (\pi_1:=\_) \land c_2 = (\_:=\pi_2)$, or,
      $c_1 = (\_:=\pi_1) \land c_2 = (\pi_2:=\_)$.
  \end{itemize}
\end{definition}

How can this definition capture races without mentioning locks?
For a thread to be blocked on a lock acquisition, the \emph{successor}
instruction to the current state \emph{must} be a $\cmdlock$
statement, which is excluded by the syntactic condition on the
successor instructions (\cf restrictions \wrt locking contexts
$L_i = L'_i = 0$ in Figure~\ref{fig:ctraces}).
Other sources of getting stuck are excluded by our
assumptions~\ref{asm:lock}--\ref{asm:control}: there no deadlocks due
to a single reentrant lock, no deterministic infinite loops and no
recursion.

%% file: analysis.tex
\section{\racerdx Analysis and its Abstract Domain}
\label{sec:analysis}

The core \racerdx analysis statically collects information about
path accesses occurring during an abstract program execution, as well as their
locking contexts. In addition to those fairly standard bits, it tracks
an additional program property, which makes it possible to filter out
potential false negatives at the reporting phases---\emph{wobbly}
paths.

\begin{definition}
\label{def:wobbly}
We call a path $\pi$ \emph{non-stable} (or \emph{unstable}, or \emph{wobbly}) in
a program $c$ if it appears as either LHS or RHS in a read or an
assignment command during some execution of $C$. We elaborate this concept
in the presence of method calls below.
\end{definition}

Formally, the analysis operates on an abstract domain, which is a
product of the three components:
domain of \emph{wobbly paths}, an \emph{access path} domain and a \emph{lock} domain.

\begin{definition}[\racerdx abstract states]
\label{def:astate}
The abstract domain is $\Abs \eqdef \pow{\Wob} \times \Lock \times \Acc$, where
\begin{itemize}
\item $\Wob \eqdef \Exp$ is the set of wobbly accesses;
\item $\Lock \eqdef \locks = \Nat$ is the current lock context;
\item
  $\Acc \eqdef \pow{\set{ \angled{c, L} ~|~ c \in
      \set{\_:=\pi, \pi := \_}, L \in \Lock}}$ is a domain of sets of
  recorded read/write accesses from/to paths which occur under lock
  state $L$.
\end{itemize}
We will use $W,L,A$ as identifiers of elements of the corresponding domains above.
\end{definition}

\noindent
The lattice structure on $\Abs$ (which we will ofter refer to as
\emph{summaries}) is ordered (via $\pre$) by
pointwise lifting of the order relations
$\angled{\subseteq, \leq, \subseteq}$ on the three components of the
domain~\cite{Cousot-Cousot:POPL79}.
For an element $D \in \Abs$, we refer to the corresponding three
projections as $D_W$, $D_L$ and~$D_A$.
For computing the join of branching control-flow and loops, the
analysis employs a standard monotone version of a lub-like operator.

\begin{definition}[Least-upper bound on $\Abs$]
\label{def:lub}
For $D_1 = \angled{W_1, L_1, A_1}$ and $D_2 = \angled{W_2, L_2, A_2}$,
\[
D_1 \lub D_2 \eqdef \angled{W_1 \cup W_2, \max(L_1, L_2), A_1 \cup A_2}.
\]
\end{definition}

\input{fig-analysis}

The intuition of what might go wrong with $\lub$ defined this way,
when aiming for a ``no-information-loss'' analysis, is easier to see
on an example.
Consider the program $\cmdif{*}{C_1}{C_2}; C$, where in $C_1$ the
lock (which is reentrant, as agreed above) is taken strictly more
times than in $C_2$; then by taking an over-approximation of the total
number of times the lock is taken for both branches (\ie, $\max$), we have
a chance to miss a race in the remainder program $C$, as some access
path can be recorded by the analyzer as having a \emph{larger} lock
context than it would have in a concrete execution.
What comes to the rescue is the assumption~\ref{asm:lock}, which we
will exploit in our proofs of TP Theorem.
Turns out, in practice, implementations of \emph{all} classes we run
the analysis on (\cf Section~\ref{sec:evaluation}) have well-scoped
locking, mostly relying on Java's \code{synchronized} primitive.

The definition of \racerdx analysis for arbitrary compound programs is
given in Figure~\ref{fig:analysis}.
The abstract transition function $\asem{\cdot}$ relies on the three
primitives, $\wobblysymb$, $\lockingsymb$, and $\accessessymb$, that
account for the three corresponding components of the abstract domain.
In the case of method calls, ($\asem{\m(\e_1,\ldots,\e_n)}$), the
analysis also takes advantage of its own compositionality, adapting a
summary of a method $\m$ to its caller context ($\angled{W, L, A}$)
and the actual arguments $\e_1, \ldots, \e_n$.


\paragraph{What does the analysis achieve?}
\label{sec:whence}
First, the results of the analysis from Figure~\ref{fig:analysis} are
used to construct a set of race candidates in \racerdx. Following the
original \racerd algorithm described
by~\citet{Blackshear-al:OOPSLA18}, if for a pair of method calls
(possibly of the same method) of the same class instance, the same
syntactic access path $\pi$ appears in the both methods summaries,
depending on the nature of the access (read or write) and the locking
context $L$, at which the access has been captured, it might be deemed
a race.

Second, the gathered wobbliness information comes into play. If a
path~$\pi'$, which is a proper prefix of~$\pi$, is identified as
\emph{wobbly} by the analysis, a race on~$\pi$, \emph{might} be a
false positive.
That is, wobbly paths ``destabilise'' the future results of the
analysis, allowing the real path-underlying values ``escape'' the
race, thus rendering it a false positive---precisely what we are
aiming to avoid. Therefore, in order to report only true races,
\racerdx removes from the final reports all paths that were affected
by wobbly prefixes.

In Sections~\ref{sec:nfp-comp}--\ref{sec:nfp-reconstr} we establish
this completeness guarantee formally, showing that races reported on
non-destabilised accesses are indeed \emph{true positives}, in the
sence that there \emph{exists} a pair of execution traces for a pair of
method calls of the same class instances that exhibit the behaviour
described by Definition~\ref{def:race}.
Furthermore, the evaluation of \racerdx in
Section~\ref{sec:evaluation} provides practical evidence that the
notion of wobbliness \emph{does not remove too} many reports, that is
that the analysis remains efficient while being precise, in the
assumptions~\ref{asm:class}--\ref{asm:control}.

%% file: fig-analysis.tex
\begin{figure}[t]

\begin{align*}
\wobbly{W}{C} & =
\begin{cases}
  W\cup\set{x,\e}^\fml & \text{if $C\equiv(x:=\e)$ or $C\equiv (\e:=x)$} \\
  W\cup\set{x}^\fml & \text{if $C\equiv(x:=\cmdalloc)$} \\
  W & \text{for any other command $C$}
\end{cases} \\
\locking{L}{C} & =
\begin{cases}
  \incrlock{L} & \text{if $C\equiv(\cmdlock)$}  \\
  \decrlock{L} & \text{if $C\equiv(\cmdunlock)$} \\
  L & \text{for any other command $C$}
\end{cases} \\
\accesses{(A,L)}{C} & =
\begin{cases}
A \cup \set{\angled{x:=\pi, L}}^\fml & \text{if $C\equiv(x:=\pi)$} \\
A \cup \set{\angled{\pi:=x, L}}^\fml & \text{if $C\equiv(\pi:=x)$} \\
A & \text{for any other command $C$}
\end{cases}
\end{align*}

\begin{align*}
\asem{c} \angled{W, L, A} &=
\angled{\wobbly{W}{c}, \locking{L}{c}, \accesses{(A,L)}{c}}
\\
\asem{C; c} \angled{W, L, A} &= \asem{c} (\asem{C} \angled{W, L, A})
\\
\asem{C; \m(\e_1,\ldots,\e_n)} \angled{W, L, A} &=
\asem{\m(\e_1,\ldots,\e_n)} (\asem{C} \angled{W, L, A})
\\
\asem{C; (\cmdif{*}{C_1}{C_2})} \angled{W, L, A} &=
\asem{C_1}(\asem{C} \angled{W, L, A}) \lub \asem{C_2}(\asem{C} \angled{W, L, A})
\\
\asem{C; (\cmdwhile{*}{C'})} \angled{W, L, A} &=
\asem{C} \angled{W, L, A} \lub±\asem{C'}(\asem{C} \angled{W, L, A})
\\
\text{where }\angled{W_1, L_1, A_1} \lub \angled{W_2, L_2, A_2}
& = \angled{W_1 \cup W_2, \max(L_1, L_2), A_1 \cup A_2}
\end{align*}

\begin{align*}
\asem{\m(\e_1,\ldots,\e_n)} \angled{W, L, A} & = \angled{W'', L'', A''}, \text{where} \\
W'' & = W \cup \set{ \e_i \mid \exists j\neq i.\, \e_i \prefix \e_j} \cup
   W'[\e_1/\argm_1,\ldots,\e_n/\argm_n], \\
L'' & = \addlock{L}{L'}, \\
A'' & = A \cup\set{\angled{c, \addlock{L}{L_c}}~|~\angled{c, L_c}
   \in A'}[\e_1/\argm_1,\ldots,\e_n/\argm_n]^{\dagger}
\\
\angled{W', L', A'} & = \asem{\body{\m}} \angled{\emptyset, 0,
                                               \emptyset}
\end{align*}
\caption{Definition of the abstract analysis semantics
  $\asem{\cdot} : \Abs \to \Abs$. We define
  $\addlock{L}{L'} = \min\{L+L', \hat{L}\}$,
  $\incrlock{L} = \addlock{L}{1}$, $\decrlock{L} = \max\{L-1, 0\}$ and
  $\hat{L}\in\Nat^+$. $^{\dagger}$We define a substitution
  $\cdot[\theta]$ (applied recursively to all syntactic elements) as
  applying only to the path component of $c$ (eg,
  $(\pi:=x)[\theta]=(\pi[\theta]:=x)$).  The function $\fml$ acts as a
  filter that only selects expressions rooted at formals, \ie,
  variables of the form $\argm_i$, and is extended straightforwardly
  to sets of elements of $A$, depending on the path component of the
  command.}
\label{fig:analysis}
\end{figure}

%% file: completeness.tex
\section{Towards TP Theorem, Part I: Analysis Completeness}
\label{sec:nfp-comp}



In this section, we deliver on the first part of the agenda towards
True Positive Theorem: the definition of \racerdx \emph{abstraction}
(\wrt the concrete semantics), and the proof of \emph{completeness} of
the analysis with respect to it.
That is, informally, if for a program $C$, the analysis $\asem{C}$
reports a certain path~$\pi$ accessed under a locking context~$L$, then
\emph{there exists} an initial configuration $\angled{S, h, 0}$ and an
execution trace $\tr \in \sem{C} \angled{S, h, 0}$, such that $\tr$
contains a command accessing $\pi$ with a locking context $L$.
To establish this, we follow the general approach of the Abstract
Interpretation framework~\cite{Cousot:PhD78,Cousot-Cousot:POPL77},
aiming for completeness, \ie, no information loss \wrt chosen
abstraction~\cite{Ranzato:VMCAI13}, stated in terms of a
traditional abstract transition function and an abstraction from a
concrete domain of sets of traces.

We first formulate the abstraction and prove its desirable properties.
The remainder of this formal development in this arc of our story takes
a ``spiral'' pattern, coming in two turns: establish the completeness
of the analysis for straight-line programs (turn one), and then lift
this proof for the inter-procedural case (turn two). In most of the
cases, we provide only statements of the theorems, referring the
reader to \ifext{Appendix~\ref{app:comp}}{the accompanying extended
  version of the paper} for auxiliary definitions and proofs.

\subsection{\racerdx Abstraction for Trace-Collecting Semantics}
\label{sec:trace-coll-semant}

The abstraction connects the results of the analyzer to the elements
of the concrete semantics, \ie, traces. It is natural to restrict the
considered traces and states to such that could indeed be produced by
executions. We will refer to them as \emph{well-formed} (WF) and
\emph{well-behaved} (WB).

\begin{definition}
\label{def:wft}
  An execution trace is well-formed if for each two subsequent states
  $\astate_i$ and $\astate_{i+1}$, the stack and heap of
  $\astate_{i+1}$ can be obtained by executing the simple command
  in $\astate_{i+1}$'s first component with respect to the stack/heap
  of $\astate_i$.
\end{definition}
\noindent
It is easy to show that for any program $C$, $S$, $h$, $L$ trace
$\tr \in \sem{C}~\angled{S, h, L}$, $\tr$ is well-formed.

\begin{definition}
\label{def:wbstate}
  A program state $\angled{\_,(S_1,S_2),h,\_}$ is well-behaved \Iff
(a) for any $s$ appearing in $S_1,S_2$ and any variable $x$,
$s(x)\in\locn{h}$, and,
(b) for any address $\alpha\in\dom{h}$, $h(\alpha)\in\locn{h}$.
\end{definition}

We remark that dangling pointers do not occur in Java, and we reproduce a similar result
for our language below.

\begin{lemma}[Preservation of well-behavedness]
\label{lem:well-behavedness}
Let $\tau$ be a non-empty, well-formed trace, whose starting state is well-behaved. Then, every
state in $\tau$ is well-behaved.
\end{lemma}

%
%

\paragraph{Abstract domain and abstraction function}
The analysis's domain (Definition~\ref{def:astate}) is rather
coarse-grained: it does not feature any information about runtime
heaps or stacks. To bring the concrete traces closer to abstract
summaries, for an execution trace $\tr$, we define \emph{syntactic
  trace} $\ab{\tr}$ as a list, obtained by taking only the first
components of each of $\tr$'s elements. That is,
$\ab{\tr} \eqdef \mathsf{map}~(\lambda~\angled{C, -, -, -}. C)~\tr$.
%
%
We now define the abstraction function $\alpha : \TTraces \to \Abs$,
from the lattice of sets of execution traces $\TTraces$ to $\Abs$,
which corresponds to ``folding'' a syntactic trace, encoding a run of
straight-line program (equivalent to the original program for this
particular execution), left-to-right, to an abstract state
recursively.

\input{fig-exec}

\begin{definition}[Abstraction Function]
\label{def:alpha}
For a set of well-formed traces $T \subseteq \Traces$,
\[
  \alpha(T) = \Lub_{\tr \in
    T}(\fst~(\exec~\ab{\tr}~(\bota, \epsilon))),~\text{where}~ \bota
    = \angled{\bot_{\Wob}, \bot_{\Lock},
      \bot_{\Acc}}
\]
and $\exec : \Traces \to \Abs \to \Abs \times (\Var \finpartialfn \Exp)$ is
defined in Figure~\ref{fig:exec}.\footnote{We overload the list notation $\ab{\tr} \cons c$
  to denote appending a list $[c]$ to $\ab{\tr}$, \ie,
  $\ab{\tr} \cons c \eqdef \ab{\tr}~+\!\!+~[c]$.}
\end{definition}

The $\bar{\theta}$ component in the state carried forward by $\exec$ is a stack of substitutions
(from formals, $\argm_i$, to paths) which mirrors the call stack in a concrete execution.
Whenever an access occurs, the substitutions~$\bar{\theta}$ are immediately applied, and the result
is a path that is rooted at a top-level variable $\argm_i$.  Accesses rooted at local variables
(any variable $x$ which is not in the syntactic form $\argm_i$ or not in the domain of the top-most
substitution) are discarded (the substitution function returns an empty set).
Similarly, the set $W$ is populated with expressions that have the substitutions $\bar{\theta}$
already applied, and similarly discarded if rooted at a local variable.

The commands $\cmdpush{}$ and $\cmdpop{}$ accordingly manipulate the substitution stack,
while $\cmdpush$ also adds certain extra paths to $W$; this is to avoid the effects of aliasing
of paths rooted at different formals inside the method body. That is, paths can become wobbly
because parts of the same path have been provided as parameters in a method call, such as in
$\m(x.\f, x.\f.\g)$. The full reasoning behind this definition will become clearer in
Section~\ref{sec:nfp-reconstr}, where it is used in the construction of a memory state
where there is exactly one parameter pointing to the heap-image of the racy path.

We next establish a number of facts about $\alpha$ necessary for the
proof of our analysis completeness.

\begin{lemma}[Additivity of $\alpha$]
\label{lm:add}
$\alpha$ is additive (\ie, preserves lubs) with respect to
$\cup_{\TTraces}$ and $\lub_{\Abs}$.
\end{lemma}

\noindent
Thanks to Lemma~\ref{lm:add}, we can define the (monotone) Galois
connection
$\angled{\TTraces, \subseteq} \galois{\alpha}{\gamma} \angled{\Abs,
  \pre}$ between the two complete lattices, where
$\gamma \eqdef \lambda a. \bigcup\set{T \in \TTraces ~|~ \alpha(T)
  \pre a}$~\cite{Cousot-Cousot:JLC92}.
Having a Galois connection between $\TTraces$ and $\Abs$ in
conjunction with completeness of the analysis
(Theorem~\ref{thm:completeness}) allows us to argue for the presence
of a certain accesses in some concrete trace of a program $C$, if the
analysis reports~it for $C$. This is due to the following fact
establishing that if $\alpha$ records an access in a certain trace,
then such a path was indeed present in its argument:

\begin{lemma}[Path access existence]
  \label{lem:path-access-existence}
  Let $T$ be a set of traces, $\alpha(T) =\angled{\_,\_,A}$, and
  $q = \angled{c, L}$ (where $c=(x:=\pi)$ or $c=(\pi:=x)$)
  is a query about the access path $\pi$ in the
  locking context $L$. If $q \in A$ then there exist
  a trace $\tr \in T$ and a non-empty, shortest prefix $\tr'\prefix\tr$
  such that
    \begin{itemize}
    \item the last state of $\tr'$ is $\angled{c',\_, \_,L}$ and $c'$, $c$ are both stores or loads;
    \item  $\exec~\ab{\tr'}~(\bota,\epsilon)=(\angled{\_,\_,A},\bar{\theta})$ where $\pi\in A$;
    \item a path $\pi'$ such that $c' = (x:=\pi')$ or $c'=(\pi':=x)$, and $\set{\pi}=\subst{\set{\pi'}}{\bar{\theta}}$.
  \end{itemize}
\end{lemma}
\begin{proof}
  By the definition of $\alpha(\cdot)$ and the properties of $\Lub$ follows that there must exist a
  trace $\tr\in T$ such that $q\in\alpha(\set{\tr})_A$. By the definition of $\exec$
  the other elements follow directly.
\end{proof}

\subsection{Proving that the Analysis Loses No Information}
\label{sec:completeness}

We structure the proof in two stages: first considering only
straight-line programs with no method calls
(Section~\ref{sec:intra-procedural}) and then lifting it to programs
with finite method call hierarchies
(Section~\ref{sec:inter-procedural}).
We do not consider the cases with recursive calls (\cf
Assumption~\ref{asm:recur}).

\subsubsection{Intra-procedural case}
\label{sec:intra-procedural}

The abstract transfer function of the analysis for simple commands
is defined in Figure~\ref{fig:analysis} as
$\asem{c}{\angled{W, L, A}}$.
We first prove the completeness for simple commands for a
singleton-trace concrete domain.

\begin{lemma}[Analysis is complete for simple commands (per-trace)]
\label{lm:comcom}
For any non-empty WF trace $\tr \in \Traces$, sets $W$, $A$, number
$L$, and a simple command~$c$, which is not $\cmdpop{}$, such
that
(a) $\angled{W, L, A} = \alpha(\set{\tr})$,
%
%
(b)
$\sem{c}~\llast{\tr} = \set{[\astate]~|~\astate \text{ is an execution
    state}}$, where $\llast{\tr}$ is the configuration
$\angled{S, h, L}$ of the last element of $\tr$, the following
\textit{{holds}}:
\[
\small
\asem{c}{\angled{W, L, A}} =
\alpha \left(\bigcup\set{\tr \cons \astate \mid [\astate] \in \sem{c}~\llast{\tr}} \right).
\]
\end{lemma}

\noindent
The following result lifts the reasoning of Lemma~\ref{lm:comcom} from
singletons to sets of arbitrary traces.

\begin{lemma}[Analysis of simple commands is complete for sets of
  traces]
\label{lm:comcom2}
For any set $T$ of non-empty well-formed traces, $W$, $L$, $A$, and
simple command~$c$, which is not $\cmdpop{}$, such that
(a) $\angled{W, L, A} = \alpha(T)$,
%
%
(b) for any $\tr \in T$, $\sem{c}~\llast{\tr} = \set{[\astate]~|~\astate \text{ is an
        execution state}}$,
then
\[
\small
{\asem{c}{\angled{W, L, A}}} =
{\alpha \left(\bigcup\set{\tr \cons \astate \mid \tr \in T, [\astate] \in \sem{c}~\llast{\tr}} \right)}.
\]
\end{lemma}

\noindent
The fact of preserving the equality of abstract results in
Lemmas~\ref{lm:comcom} and~\ref{lm:comcom2} is quite noteworthy: for
straight-line programs (and, in fact, for any program in our IL) the
analysis is \emph{precise}, \ie, we do not lose information \wrt
locking context, wobbliness, or access paths, and hence can include
elements $\angled{c, L}$ into the $A$-component as a part of
$\accessessymb$-machinery without the loss of precision.

We now lift these facts to the analysis for compound programs without
method calls.
%
Recall that we only  consider programs with \emph{balanced locking}
(Assumption~\ref{asm:lock}), \ie, such that within them
\begin{itemize}
\item the commands $\cmdlock{}$ and $\cmdunlock{}$ are only allowed to
  appear in balanced pairs (including their appearances in nested
  method calls) within conditional branches and looping statements,
  and,
\item every $\cmdunlock{}$ command has a matching $\cmdlock{}$.
\end{itemize}
\noindent
The following two lemmas are going exploit this fact for proving the
analysis completeness.

\begin{lemma}[\racerdx analysis and balanced locking]
\label{lm:balanced}
If\ \ $C$ is a compound program with balanced locking, and $\angled{W', L',
  A'} = \asem{C} \angled{W, L, A}$. Then $L' = L$.
\end{lemma}

\begin{lemma}[Balanced locking and syntactic traces]
  \label{lem:abstraction-state-insensitivity}
  For any program $C$ with balanced locking
  and well-behaved states $S_1,h_1$ and $S_2,h_2$ and $L \ge 0$,
  \[
  \{\hat \tau \mid \tau\in\sem{C}\angled{S_1,h_1,L})\} =
  \{\hat \tau \mid \tau\in\sem{C}\angled{S_2,h_2,L})\} \neq \emptyset.
  \]
  As a corollary,
  $\alpha(\sem{C}\angled{S_1, h_1, 0}) = \alpha(\sem{C}\angled{S_2,
    h_2, 0})$.
\end{lemma}

The proof of Lemma~\ref{lem:abstraction-state-insensitivity} hinges on
the the following observations:
\begin{itemize}
\item Balanced locking ensures that
  $\sem{\cmdunlock}~\angled{S,h,L}\neq\emptyset$, which the case for
  any intermediate state of the traces of the programs in
  consideration.
\item Well-behavedness ensures that the set of traces for loads and stores is non-empty.
\item Well-behavedness is preserved along traces (\cf{} the semantics
  of $\cmdalloc$).
  \end{itemize}
\noindent
That is, for compound programs with balanced lockings the
$L$-component of the abstraction remains immutable (by the end of
execution) for both $\asem{\cdot}$ and the abstraction $\alpha$ over
the concrete semantics.
The following lemma delivers the completeness in the intra-procedural
case.

\begin{lemma}
\label{lm:com-compound}
For any compound program $C$ with balanced locking and no method
calls, the starting components $W_0$, $L_0$, $A_0$, and a set of
well-formed non-empty traces $T$, such that
$\angled{W_0, L_0, A_0} = \alpha(T)$,
  \[
     \asem{C} \angled{W_0, L_0, A_0}
    =
    \begin{cases}
        \alpha\left(\bigcup_{\tr \in T} \set{\tr \append \tr' ~|~ \tr' \in
        \sem{C}~\llast{\tr}}\right)
        & \text{ if}~~T \neq \emptyset,
        \\[4pt]
        \alpha\left(\sem{C} \angled{S, h, 0} \right)
        & \text{ otherwise, for any well-behaved $S,h$,}
      \end{cases}
  \]
  where $\llast{\tr}$ is well-defined, as $T$ consists of non-empty traces.
\end{lemma}
\begin{proof}
  By induction on $C$ and Lemma~\ref{lm:comcom2}. \ifext{Details are
    in~\refp{proof:com-compound} in Appendix~\ref{app:comp}.}{\yada}
\end{proof}

\subsubsection{Inter-procedural case}
\label{sec:inter-procedural}
The main technical hurdle on the way for extending the statement of
Lemma~\ref{lm:comcom2} to an inter-procedural case (\ie, allowing for
method calls) is the gap between the semantic implementation of method
calls via stack-management discipline (Figure~\ref{fig:traces}) and
treatment of method summaries by the analysis, via explicit
substitutions (Figure~\ref{fig:analysis}).
To address this, we introduce the following convention, enforced by
\racerdx's intermediate language representation:
\begin{definition}[ANF]
\label{def:anf}
We say that a method $\m(\argm_1, \ldots, \argm_n)$ is in
Argument-Normal Form (ANF), if for every simple command $c$ in its
body, $c$ is \emph{not} of the form $\pi := \argm_j$ for some
$\argm_j \in \many{\argm_i}$.
\end{definition}
Intuitively, this requirement enforces a ``sanitisation'' of used
arguments from the set $\many{\argm_i}$, so one could
\emph{substitute} them with some (non-variable) paths without
disrupting the syntactic structure of a compound program.
Let us denote as $\wb$ the set of paths
$\set{ \e_i \mid \exists j\neq i.\, \e_i \prefix \e_j \lor \e_j \prefix \e_i}$ for
$i \in \set{1 \dots n}$. We will also denote via $\cup_1$ the
following operation on $\Abs$:
\[
\angled{W, L, A} \wcup W' = \angled{W \cup W', L, A}.
\]
The following lemma accounts for the mentioned concrete/abstract
discrepancy in treatment of methods, allowing us to reformulate the
definition of abstract method summaries returned by the analyzer
($\asem{\m(\e_1,\ldots,\e_n)} \angled{W, L, A}$) via substitutions of
their bodies.

\begin{lemma}[ANF and method summaries]
\label{lm:anf-sum}
For a method $\m$, a vector of expressions $e_1, \ldots, e_n$, a
configuration $\angled{s :: S, h, L} = \llast{\tr}$ such that
  \begin{itemize}
  \item $\forall i, 1 \leq i \leq n, \sem{e_i}_{s, h}$ are defined, and
  \item $\body{\m}$ has no nested calls and features well-balanced
    locking, the following holds:
  \end{itemize}
{\small{
\[
\begin{array}{r@{\ }c@{\ }l}
\asem{\m(\e_1,\ldots,\e_n)} \angled{W, L, A} &{\;=\;}&
\angled{W', L', A'} \wcup \wb, \text{where} \\[3pt]
&&
\angled{W', L', A'} = \asem{\body{\m}[\many{\e_i/\argm_i}]}~\angled{W, L, A}
\end{array}
\]
}}
\end{lemma}
%
%
%
The first Good Thing afforded by Lemma~\ref{lm:anf-sum} is the removal
of substitutions from the $W$ and $A$ components and moving them to
the method's body, which will give the uniformity necessary for
conducting the forthcoming proofs.
The second Good Thing is the given possibility to ``shift'' the
computation of the method summary from the
$\bota = \angled{\emptyset, 0, \emptyset}$ (as given by
Figure~\ref{fig:analysis}) to the actual abstract context
$\angled{W, L, A}$.
Finally, notice that in the absence of recursion, the statement of
Lemma~\ref{lm:anf-sum} can be generalised to methods with nested
calls, with the proof by induction on the size of the call tree.

Armed by Lemma~\ref{lm:anf-sum} (generalised to nested calls), we can
prove analogues of Lemmas~\ref{lm:comcom}--\ref{lm:com-compound}. In
the presence of method calls. The formal development first establishes
the result for a singleton trace, similar to Lemma~\ref{lm:comcom} for
a program with no \emph{nested} method calls, then lifting it, by
induction on the size of the call graph, to arbitrary non-recursive
call hierarchies. It is then further generalised for sets of traces,
in a way similar to the proof of Lemma~\ref{lm:comcom}.
For the sake of saving space, we do not present the statements of this
lemmas here (as they are not particularly remarkable or pretty) and
refer the reader \ifext{to Appendix~\ref{app:comp}}{the accompanying
  extended version of the paper} for auxiliary definitions and proofs.
This arc of formal results concludes with the following statement,
similar to Lemma~\ref{lm:com-compound}:

\begin{lemma}
\label{lm:com-compound-calls}
For any compound program $C$ with balanced locking and all methods in
ANF, the analysis domain components $W_0$, $L_0$, $A_0$, and a set of
well-formed non-empty traces $T$, such that
$\angled{W_0, L_0, A_0} = \alpha(T)$,
  \[
     \asem{C} \angled{W_0, L_0, A_0}
    =
    \begin{cases}
        \alpha\left(\bigcup_{\tr \in T} \set{\tr \append \tr' ~|~ \tr' \in
        \sem{C}~\llast{\tr}}\right)
        & \text{ if}~~T \neq \emptyset,
        \\[4pt]
        \alpha\left(\sem{C} \angled{S, h, 0} \right)
        & \text{ otherwise, for any well-behaved $S,h$,}
      \end{cases}
  \]
  where $\llast{\tr}$ is well-defined, as $T$ consists of non-empty traces.
\end{lemma}

\subsubsection{Main completeness result}
\label{sec:main-completeness}

We are now ready to establish the completeness of the abstract
semantics $\asem{\cdot}$ of our analysis \wrt trace-collecting
semantics $\sem{\cdot}$ and abstraction $\alpha$.

\begin{theorem}[Completeness of \racerdx analysis.]
  \label{thm:completeness}
  For any compound program $C$ with all methods in ANF, balanced
  locking, and well-behaved $S, h$,
  \[
    \asem{C} \bota
    =
    \alpha\left(\sem{C} \angled{S, h, 0} \right).
  \]
\end{theorem}
\begin{proof}
Follows immediately from Lemma~\ref{lm:com-compound-calls} by taking
$\angled{W_0, L_0, A_0} = \bota$, $T = \emptyset$.
\end{proof}

%% file: fig-exec.tex
\begin{figure}[t]
\[\def\arraystretch{1.7}
\begin{array}{@{\exec\ \ }l@{~(\angled{W, L, A}, \bar{\theta})\ \eqdef\ \ }l@{}}
\epsilon & (\angled{W, L, A}, \bar{\theta'})
\\
\ab{\tr} \snoc \angled{x:=y} &
\left(
\angled{W' \cup \subst{\set{x, y}}{\bar{\theta'}}, L', A'}, \bar{\theta'}
\right)
\\
\ab{\tr} \snoc \angled{x:=\pi} &
\left(
\angled{
W' \cup \subst{\set{x, \pi}}{\bar{\theta'}}, L', A' \cup \subst{\set{\angled{x:=\pi,L}}}{\bar{\theta'}}
}, \bar{\theta'}
\right)
\\
\ab{\tr} \snoc \angled{\pi:=x} &
\left(
\angled{
W' \cup \subst{\set{x, \pi}}{\bar{\theta'}}, L', A' \cup \subst{\set{\angled{\pi:=x, L}}}{\bar{\theta'}}
}, \bar{\theta'}
\right)
\\
\ab{\tr} \snoc \angled{x:=\cmdalloc} &
\left(
\angled{
W' \cup \subst{\set{x}}{\bar{\theta'}}, L', A'
}, \bar{\theta'}
\right)
\\
\ab{\tr} \snoc \angled{\cmdlock} &
\left(
\angled{W', \addlock{L'}{1}, A'}, \bar{\theta'}
\right)
\\
\ab{\tr} \snoc \angled{\cmdunlock} &
(\angled{W', \max(L'-1, 0), A'}, \bar{\theta'})
\\
\ab{\tr} \snoc \angled{\cmdpush{\e_1,\ldots,\e_n}} &
\left(
\begin{array}{l}
\angled{W' \cup \subst{\set{\e_i \mid \exists j\neq i.\, \e_i \prefix \e_j}}{\bar{\theta'}}, L', A'}, \\
\lbrack \argm_i \mapsto \e_i \rbrack^n_{i=1}
\cons \bar{\theta'}
\end{array}
\right)
\\
\ab{\tr} \snoc \angled{\cmdpop{}} &
(\angled{W', L', A'}, \mathsf{tail}~\bar{\theta'})
\end{array}
\]

where $(\angled{W', L', A'}, \bar{\theta'}) \eqdef \exec~\ab{\tr}~(\angled{W, L, A},\bar{\theta})$ and

\begin{align*}
  \subst{E}{\bar{\theta}} & \eqdef \bigcup_{\e\in E} \subst{\e}{\bar{\theta}}
  \text{ for a set of expressions $E$}
  &
  \subst{x.\fs}{\epsilon} & \eqdef
  \begin{cases}
    \emptyset & \text{if $x\neq\argm_i$ for any $i$}  \\
    \set{x.\fs} & \text{otherwise}
  \end{cases}
  \\
  \subst{x.\fs}{\theta\cons\bar{\theta}} & \eqdef
  \begin{cases}
    \emptyset & \text{if $x\notin\dom{\theta}$} \\
    \subst{((\argm_i[\theta]).\fs)}{\bar{\theta}} & \text{if $x=\argm_i\in\dom{\theta}$}
  \end{cases}
  &
  \subst{\set{x:=\pi}}{\bar{\theta}} & \eqdef
  \begin{cases}
    \emptyset & \text{if $\subst{\pi}{\bar{\theta}}=\emptyset$} \\
    \set{x:=\pi'} & \text{if $\subst{\pi}{\bar{\theta}}=\set{\pi'}$}
  \end{cases}
\end{align*}
\caption{An auxiliary function $\exec$ for executing syntactic traces to compute \racerdx abstraction.}
\label{fig:exec}
\end{figure}

%% file: reconstruction.tex
\section{Towards TP Theorem, Part II: Reconstructing Racy Traces}
\label{sec:nfp-reconstr}

Here we show how, given a pair of programs $C_1,C_2$ for which \racerdx reports a race, we \emph{construct}
a concurrent trace (as defined in Figure~\ref{fig:ctraces}) that \emph{exhibits the race}, in the sense of
Definition~\ref{def:race}.

We crucially rely on the notion of \emph{stability} (\emph{non-wobblyness}): an access path $\pi$ is stable if
no proper prefix of $\pi$ is ever read or written.  We track the negation of this property through the
$W$ component of the abstract state, which is slightly complicated by our need to track accesses
to proper prefixes of $\pi$ even inside procedure calls, and thus employ substitutions of actuals over
formals (cf.~the definition of $\exec$ for load/stores and calls).
Note that stability does not preclude accesses to the path itself.
Stability allows us to prove that the backbone of a path (the domain of its heap-image) is
preserved during execution, and thus when the path is accessed at the end of a trace for $C_1$/$C_2$,
the same address in both threads is accessed.

However, not all traces of $C_1$/$C_2$ will allow us to do this. For this reason we introduce the
notions of \emph{path disconnectedness} and \emph{acyclicity}, two semantic restrictions that ensure
that accesses through paths syntactically distinct to $\pi$, and suffixes of $\pi$ respectively,
do not affect the heap-image of $\pi$.

\begin{definition}[Prefix]
An access path $\pi = x.\fs$ is a prefix of $\pi' = y.\gs$ if $x = y$
and the sequence $\mathsf{fs}$ is a prefix of $\mathsf{gs}$. We denote this fact
by $\pi \prefix \pi'$. The notion of proper prefix ($\properprefix$) is straightforward.
We lift this to expressions by fixing $x \properprefix x.\fs$ and $x\prefix x$.
\end{definition}
The root of an access path $x.\mathsf{fs}$ is $x$, \ie, $\rootvar{x.\fs} = x$.
If $\pi = x.\fs$, we may write $\pi.\g$ or $x.\fs.\g$ to mean $x.(\fs.\g)$.
We set $|x.\fs|=|\fs|$, \ie, the length of the field list.

\begin{definition}[Path Footprint]
Let $s,h$ be a state where $\lval{\pi}{s,h}$ is defined. Define the \emph{footprint of $\pi$},
denoted $\pathimg{h}{s}{\pi}$, as
\[
\pathimg{h}{s}{\pi} = \bigcup_{\pi'\prefix\pi } \{ \lval{\pi'}{s,h} \mapsto h(\lval{\pi'}{s,h}) \}\ .
\]
\end{definition}
Intuitively, the footprint of an access path is a heap fragment containing all addresses necessary to resolve the address
of the path, inclusive. It's easy to see that
\[ \dom{\pathimg{h}{s}{\pi}} = \{ \lval{\pi'}{s,h} \mid \pi' \prefix\pi\}. \]

\begin{definition}[Disconnectedness]
A path $\pi$ is \emph{disconnected in a state $s,h$} if $\lval{\pi}{s,h}$ is defined, and
\begin{itemize}
\item for any $x\neq\rootvar{\pi}$, $s(x)\notin\locn{\pathimg{h}{s}{\pi}}$;
\item for all $\alpha\in\dom{h}$, if
$h(\alpha)\in\locn{\pathimg{h}{s}{\pi}}$ then $\alpha\in\dom{\pathimg{h}{s}{\pi}}$.
\end{itemize}
\end{definition}
%

\begin{definition}[Acyclicity]
A path $\pi$ is \emph{acyclic} in state $s,h$ if $\lval{\pi}{s,h}$ is defined,
and for any two paths $\pi''\prefix\pi'\prefix\pi$,
$h(\lval{\pi'}{s,h})\neq\locn{\lval{\pi''}{s,h}}$.
\end{definition}

\begin{definition}[Path preservation]
Let $\pi$ be a path, and $s,h$ and $s',h'$ two states. We say that $\pi$ is preserved from $s,h$ to $s',h'$
if
\begin{enumerate}
  \item $\pi$ is disconnected and acyclic in $s,h$ and $s',h'$;
  \item $s(\rootvar{\pi})=s'(\rootvar{\pi})$;
  \item $\dom{\pathimg{h}{s}{\pi}} = \dom{\pathimg{h'}{s'}{\pi}}$;
  \item $\forall\alpha\in\dom{\pathimg{h}{s}{\pi}}\setminus\{\lval{\pi}{s,h}\}.\,h(\alpha)=h'(\alpha)$.
\end{enumerate}
\end{definition}
It is easy to see that if $\pi$ is preserved from $s,h$ to $s',h'$, then $\lval{\pi}{s,h}=\lval{\pi}{s',h'}$.
Also, is is straightforward to see that path preservation is an equivalence relation on memory states.

We begin by showing that a trace without calls, for which $\pi$ is stable, preserves $\pi$.
\begin{lemma}
\label{lem:stability-means-preservation-nocalls}
Let $\pi$ be an access path and $\tr$ a WF, non-empty trace without any
$\cmdpush{}$/$\cmdpop{}$ commands,
whose starting state
$\angled{\cmdskip, s_0\cons S, h_0, \_}$ is well-behaved,
and where $\pi$ is acyclic and disconnected in $s_0,h_0$.
Let the last state of $\tr$ be $\angled{\_, s'\cons S, h', \_}$.

If it is the case that $\forall \e\in \alpha(\set{\tr})_W. \e\not\properprefix\pi$, then
$\pi$ is preserved from $s_0,h_0$ to $s',h'$.
\end{lemma}
\begin{proof}
  By induction on the length of $\tr$ and results on stability.
  See \ifext{Appendix~\ref{app:reconstr}}{the accompanying extended version} for details of the proof.
\end{proof}

We can now lift this up to the trace of a single, non-nested procedure call.

\begin{lemma}
  \label{lem:nocalls-preservation}
  Suppose that $\mbody(\m)$ has no procedure calls, that $\pi$ is an access path which is
  disconnected and acyclic in the well-behaved state $s,h$,
  that $\tr\in\sem{\cmdskip;\m(\bar{\e})}\angled{s,h,L}$ is a non-empty trace,
  that $\llast{\tr}=\angled{\cmdpop{},s,h',\_}$, that $\alpha(\set{\tr})=\angled{W,\_,\_}$, and that
  $\forall\e\in W.\, \e\not\properprefix\pi$.

  Then, $\pi$ is preserved from $s,h$ to $s,h'$.
\end{lemma}
\begin{proof}
  First, notice that $\forall\e\in W.\,\e\not\properprefix\pi$ and the definition of $\exec$ implies
  there is at most one $\e_i$ such that $\e_i= z$, or $\e_i=\pi_i\properprefix\pi$. We then analyse
  the three cases separately: no $\e_i\properprefix\pi$, $\e_i=\rootvar{\pi}$ and $\e_i\properprefix\pi$.
  The most interesting is the last one, where we cut out the subheap of $h$ reachable from the
  arguments, apply Lemma~\ref{lem:stability-means-preservation-nocalls} to it, and then restore
  the facts we prove back to the top-level stack. Notably, here we use the \emph{frame property}
  at the semantics level \cite{Hongseok02}.
\end{proof}

We can further lift the preservation result to traces that have arbitrarily nested procedure calls.
\begin{lemma}
  \label{lem:stability-on-traces}
  Let $\tr$ be a WF, non-empty trace with matched $\cmdpush{}/\cmdpop{}$ pairs,
  whose first state is $\angled{\cmdskip, s,h,\_}$ and
  $\llast{\tr}=\angled{\_,s',h',\_}$.  Let $\pi$ be an access path which is
  disconnected and acyclic in the well-behaved state $s,h$,
  that $\alpha(\set{\tr})=\angled{W,\_,\_}$, and that
  $\forall\e\in W.\, \e\not\properprefix\pi$.

  Then, $\pi$ is preserved from $s,h$ to $s,h'$.
\end{lemma}
\begin{proof}
  By induction over subtraces corresponding to method calls, using Lemma~\ref{lem:nocalls-preservation}.
\end{proof}

While we have shown that a stable path is preserved along a trace, this only applies to balanced
pairs of $\cmdpush{}$/$\cmdpop{}$ commands in the trace.  However, the racy access may occur
inside a chain of procedure calls, thus we need to further show that a stable path propagated through
the call stack is also preserved, up to the point of access.

\begin{lemma}
  \label{lem:call-stack}
  Let $\tr$ be a (prefix) trace produced by Lemma~\ref{lem:path-access-existence}.
  If additionally it is the case that $\forall\e\in\alpha(\set{\tr})_W.\,\e\not\properprefix\pi$
  then $\lval{\pi}{s,h}=\lval{\pi'}{s',h'}$.
\end{lemma}
\begin{proof}
  We split $\tr=\tr_1\append\tr_{\text{push}}\append\tr_0$ on the last unmatched $\cmdpush{}$ command,
  and apply Lemma~\ref{lem:stability-on-traces} to $\tr_0$, showing that $\pi'$ is preserved.
  Then, we lift $\lval{\pi'}{s',h'}$ to the parent call stack and identify the argument $\e_j$
  such that $\rootvar{\pi'}=\argm_j$.  Then, we repeat the process for $\e_j$ as $\pi'$ up to the
  top-level stack.
  See \ifext{Appendix~\ref{app:reconstr}}{the accompanying extended version} for details of the proof.
\end{proof}

We are ready to state the main result of this section and the whole paper.

\begin{theorem}[True Positives Theorem]
\label{thm:no-fp}
  Let $C_1,C_2$ be two programs such that $\asem{C_i} \bota= \angled{W_i,\_,A_i}$.
   Let $\pi_i$ be two paths and $c_i$ two commands such that $\angled{c_i,L_i}\in A_i$, and,
  \begin{itemize}
    \item $\pi_1=v_1.\fs$ and $\pi_2=v_2.\fs$ (\ie, the field sequences are the same);
    \item $\forall \e_i\in W_i.\, \e_i\not\properprefix \pi_i$ (for $i\in\set{1,2}$);
    \item $c_1=(\pi_1:=x)$ and $c_2=(\pi_2:=y)$, or, $c_1=(\pi_1:=x)$ and $c_2=(y:=\pi_2)$;
    \item $L_1+L_2 \le 1$.
  \end{itemize}
  Then there exists a well-behaved state $\astate=\angled{(s_1,s_2),h,(0,0)}$  and a
  concurrent trace $\tr^{\parallel} \in \sem{C_1\parallel C_2}~\astate$ that races.
\end{theorem}
\begin{proof}
  We sketch the important intuitions behind the proof. For more details, see
  \ifext{Appendix~\ref{app:reconstr}}{the accompanying extended version}.

  By assumption, one of $L_1=0\lor L_2=0$, so \Wlog, we set $L_1=0$.
  This means that when $C_1$ accesses the path $\pi$ it is \emph{not holding the lock}.

  We first construct a special initial state, where both programs ``see'' the same path on the
  shared heap, and where that path is acyclic and disconnected.

  We then obtain a partial trace $\tr_1$ of $C_1$
  as ordained by Lemma~\ref{lem:path-access-existence},
  that ends just before committing an access to $\pi$.

  We further show that the final state of that trace satisfies the preconditions of
  Lemma~\ref{lem:path-access-existence} for $C_2$ thus obtaining a trace $\tr_2$ that picks up
  from where $\tr_1$ finished.

  We weave $\tr_1$ and $\tr_2$ into a concurrent trace
  by letting $C_1$ progress until it reaches the end of $\tr_1$, and then $C_2$ progresses until
  the end of $\tr_2$.  Crucially, this schedule is realisable because of the above observation
  that $C_1$ isn't holding the lock.

  We finally use
  Lemmas~\ref{lem:stability-on-traces}--\ref{lem:call-stack} to show that the path is preserved
  to the concurrent trace' final state, and that the path each thread sees through its call
  stack resolves to the same address.
\end{proof}

%% file: evaluation.tex
\section{Implementation and Evaluation}
\label{sec:evaluation}

In the previous sections we laid out the design of a theoretical analyser that is similar to
\racerd, but which enjoys the True Positives property: under certain assumptions, reports from
the theoretical analyser proposed here, are true positives. But that still does not address
the effectiveness of the new analyser: we would like to know would an implementation of \racerdx be
\begin{enumerate}
  \item effective, in that it reports a sufficient proportion of reports which \racerd generates;
  \item efficient, in that it is not significantly slower than \racerd.
\end{enumerate}

We implemented the proposed analyser as a modification of \racerd.
The implementation was relatively straightforward: it required less than 1kLOC of OCaml code, and used
data structures and algorithms that are either standard or come with the Infer analyser.

We next discuss the differences between the two analysers.

\subsection{Differences between \racerd and \racerdx}
\label{sec:diff-racerd}

The main difference between \racerd and \racerdx is in substituting the ownership domain of the former
with that of stability.

Newly allocated objects in Java are known only to the caller of \code{new}.  Thus, it is not
possible for two threads to race on a newly allocated object, unless the address of that object
flows into some shared location (basically, the heap).  In addition, accesses to fields of the
newly allocated object cannot possibly participate in a data race, for much the same reasons.
\racerd uses a notion of \emph{ownership} typing to track access paths which have been obtained
through allocation, in an effort to avoid the associated false positives.

\racerdx swaps the whole abstract domain of ownership for that of
stability. This substitution is logically warranted by the fact that
allocation immediately destabilises all paths extending that of the
receiver. For instance, in the code \code{x.b = new Bloop(); x.b.f =
  1;} the access at \code{x.b.f} will not be reported by \racerdx
because \code{x.b} $\properprefix$ \code{x.b.f}, and \code{x.b} will
have been recorded in the $W$ component of the analysis state.

\subsection{Differences between Theoretical Analyser and its Implementation}
\label{sec:diff-theory}

Mathematical formulations of static analysers typically deal with a simplified language and a
simplified semantics for reasons of simplicity and clarity.  Their implementations, in so far as
the target is a real programming language on a real concrete semantics, usually are not one-to-one
renditions of the math in code.  Here we examine some of the differences between the formal model
of \racerdx and its implementation. Some of these differences are directly inherited from \racerd;
these include a treatment of pure procedures (aided by programmer annotations); support for
confining a set of accesses in a specific thread (for instance, a UI thread) such that no lock is
required to avoid interference; \code{static} and \code{volatile} fields; inheritance,
and others.

The parts of \racerd we had to specifically modify for stability, and which constitute a complex
feature not present in the formal model of \racerdx include:
\begin{description}

\item[Global variables:] These clearly do not fit the model of \racerdx, as they
  do not behave as formals.  We model them in the implementation as paths that
  never undergo substitution, and are never destabilised.  This means that any
  access to them not under synchronisation may be reported in the context of a race.

\item[Constructors:] A constructor cannot normally race on the freshly
  created object. To avoid the obvious false positives, we treat
  \code{this} as always unstable when inside a constructor. Thus,
  \code{this}$~\in W$ and therefore any statement that initialises
  a member field will not be reported because
  \code{this}$\properprefix$\code{this.f} for any field
  \code{f}. Note that this allows potential false negatives in
  the case where the address of \code{this} (or some path
  starting at \code{this}) escapes the constructor.


  \item[Containers:] Such objects are treated specially by \racerd, because they are usually opaque
  (live in libraries) and because the model they expose is much simpler than their implementation.
  For instance, an addition of an element to an \code{ArrayList} at the path
  \code{this.arrayList} will be treated as a store, for the purposes of thread safety.
  However, \racerdx needs to distinguish between accesses to the container itself and accesses that
  manipulate the contents, for instance \code{this.arrayList = f()} and
  \code{this.arrayList.get(0)}.  The way this is achieved is through a dummy field that is
  read/written to when an element is read/added from/to the container.
\end{description}

\subsubsection{Analysing Open-Source Projects}
\label{sec:open}

We obtained six open-source Java projects and analysed them with \racerd and \racerdx, recording their
total CPU time cost and data race reports.
The test environment was Linux 4.11, running on a server with 56Gb of RAM, on an Intel CPU at 2.5GHz.
Results are shown in Table~\ref{tbl:evaluation}.
The total CPU time of each tool (in seconds) and the proportion of their difference are given in the columns
``D CPU'', ``DX CPU'', ``CPU $\pm$\%''.

An access path may be accessed at a set of locations.  In the worst case where are these accesses
are racy, we may end up with a number of reports that is quadratic in the number of locations.
\racerd and \racerdx both have de-duplication
capacilities to avoid spamming developers, thus selecting a subset of races.
Since, for the purposes of the evaluation, this selection is arbitrary, we
deactivated de-duplication in both tools to ease the comparison.
We checked that \racerdx never introduces reports, that is, if \racerdx makes a particular report
then \racerd always makes the same report too. We show the number of race reports for both tools
and the proportion of the difference in the columns ``D Reps'', ``DX Reps'' and ``Reps $\pm$\%''.

To further elucidate the fundamental differences, we extracted the access path $\pi$ a race report
identifies, and counted how many unique paths are reported by each tool. Again, we found that
\racerdx never reports a path not also reported by \racerd. We show the number of paths reported in
the columns ``D \#$\pi$'', ``DX \#$\pi$'' and the proportion of the difference in ``\#$\pi$ $\pm$\%''.

\begin{table}
  \footnotesize
  \begin{subtable}{\textwidth}
  \begin{tabular}{@{\:\:}l@{\:\:}l@{\:\:}l@{}}
    \toprule
    Project & Description & URL \\
    \midrule
    avrora & An AVR emulator & \url{https://github.com/ibr-cm/avrora} \\
    Chronicle-Map & A non-blocking key-value store & \url{https://github.com/OpenHFT/Chronicle-Map} \\
    jvm-tools & Tool for JVM troubleshooting and profiling & \url{https://github.com/aragozin/jvm-tools} \\
    RxJava & Library for asynchronous and event-based programs & \url{https://github.com/ReactiveX/RxJava} \\
    sunflow & Rendering system for photo-realistic image synthesis & \url{https://github.com/fpsunflower/sunflow} \\
    xalan-j & XSLT processor & \url{https://github.com/apache/xalan-j} \\
    \bottomrule
  \end{tabular}
  \caption{Evaluation targets.}
  \label{tbl:targets}
  \end{subtable}
\vspace{1ex}
\begin{subtable}{\textwidth}
  \begin{tabular}{@{}lrrrrrrrrrr@{}}
    \toprule
    Target & LOC & D CPU & DX CPU & CPU $\pm$\% & D Reps & DX Reps & Reps $\pm$\% & D \#$\pi$ & DX \#$\pi$ & \#$\pi$ $\pm$\% \\
    \midrule
    avrora & 76k & 103 & 102 & 0.4\%  & 143 & 92 & 36\%  & 78 & 38 & 51\% \\
    Chronicle-Map & 45k & 196 & 196 & 0.1\%  & 2 & 2 & 0\%  & 2 & 2 & 0\% \\
    jvm-tools & 33k & 106 & 109 & -3.6\%  & 30 & 26 & 13\%  & 14 & 11 & 21\% \\
    RxJava & 273k & 76 & 69 & 9.2\%  & 166 & 134 & 19\%  & 65 & 44 & 32\% \\
    sunflow & 25k & 44 & 44 & -1.4\%  & 97 & 42 & 57\%  & 116 & 38 & 67\% \\
    xalan-j & 175k & 144 & 137 & 5.0\%  & 326 & 295 & 10\%  & 135 & 94 & 30\% \\
    \bottomrule
  \end{tabular}
  \caption{Evaluation results.  CPU columns are in seconds; Reps are distinct reports; $\pi$ are distinct paths.}
  \label{tbl:evaluation}
\end{subtable}
\end{table}

We can make the following observations:
\begin{itemize}
  \item The difference between runtimes is largely within the noise margins, especially given
  that a large percentage of these runtimes is spent compiling Java source into bytecode, as Infer
  extracts an AST from the compiled artefact.
  \item The loss in terms of number of reports ranges between 10\% and 57\% (we exclude Chronicle-Map
  as there are too few reports to start with), and the loss in terms of number of distinct access
  paths ranges from 21\% to 67\%.
\end{itemize}

\subsubsection{The Causes for Deterioration of Reporting Rate}
\label{sec:caus-deter-report}

We triaged a sample of reports that \racerd made but \racerdx didn't. We discerned two main
classes of reports:
\begin{itemize}
  \item In a call \code{this.foo(this.f)}, the check whether one of the actuals is a proper
  prefix of another (Figure~\ref{fig:analysis}, method call case)
  fails because \code{this}$\properprefix$\code{this.f}.
  Thus, \code{this} is marked as unstable (is added to the $W$ component of the abstract state).
  But this means that all accesses in the caller method will not be considered in data race reports,
  leading to potential false negatives.

  A potential solution here may be to use the fact that the first
  argument of a non-\code{static} Java method cannot be reassigned,
  and thus may be left out of the check above, but we have not at
  present assessed how this might affect the status of our theorems.

  \item Inner (nested) classes are common in Java, and allow methods of an inner class object to
  reference fields and methods of the containing class.  To achieve this, the compiler inserts in the
  inner class a hidden reference to the outer class object, and initialises this appropriately at
  construction.  Unfortunately, this also means that the \code{this} reference of the outer
  class is marked as unstable whenever an inner class object is constructed, thus precluding
  accesses occurring in the enclosing method from being reported.
\end{itemize}

Remarkably, all the reports we triaged were true positives.
Both of these classes of missing reports may benefit from elaborating the stability abstract domain
to track escaping references, \ie, when a path is read it is not immediately marked as unstable,
but only when the address read ends up being stored somewhere.  This is something we will investigate in further work.

%% file: related.tex
\section{Discussion and Related Work}
\label{sec:related}



In this section we place \racerdx in the field of static race detectors that have been subjected to formal analysis.
Since the majority of the static race analyzers used in
industry---\tname{ThreadSafe}~\cite{Contemplate}, Coverity's analysis
suite~\cite{AndyChou14}, and
\racerd~\cite{Blackshear-al:OOPSLA18}---despite being
impactful in practice, do not come with any formal theorem about their
algorithm, we focus on the
only two race analyzers we know of, for which a formal statement has
been made: \chord~\cite{Naik-al:PLDI06,Raghothaman-al:PLDI18}
and~\racerdx.

\chord is a tool which strives to be \emph{sound}, or to favour
reduction of false negatives over false positives.  Note that our
discussion should not be interpreted as providing value judgements on
\racerdx versus \chord. We think of both as exemplars of points in the
design space which might be built upon or learned from in the
construction of new analyzers that are useful in practice.

The notion of soundness is a standard one: the results of an analysis
model \emph{all possible executions} of the program \wrt certain
behaviors of interest, \eg, uncaught exceptions, memory leaks,
\etc. This is the concept of over-approximation. A consequence of an
over-approximation result is that if an analyzer claims that there are
no bugs then there are none (at least in the idealized formal model regarded as the concrete semantics). 

\racerdx, on the other hand, is an example of a static analysis that favours reduction of false positives over false negatives.
It aims to be {\em complete}\/, to report only true races, under certain assumptions.
A series of diagrams, shown in Figure~\ref{fig:circles}, illustrate specific
aspects of the analysis design space and depict the
notions of \emph{over-of-under} and \emph{under-of-over} approximation from the introduction.
We discuss each of the diagrams.

%
%
%
%
%
%


\begin{figure}[t]
\centering
\includegraphics[width=0.98\textwidth]{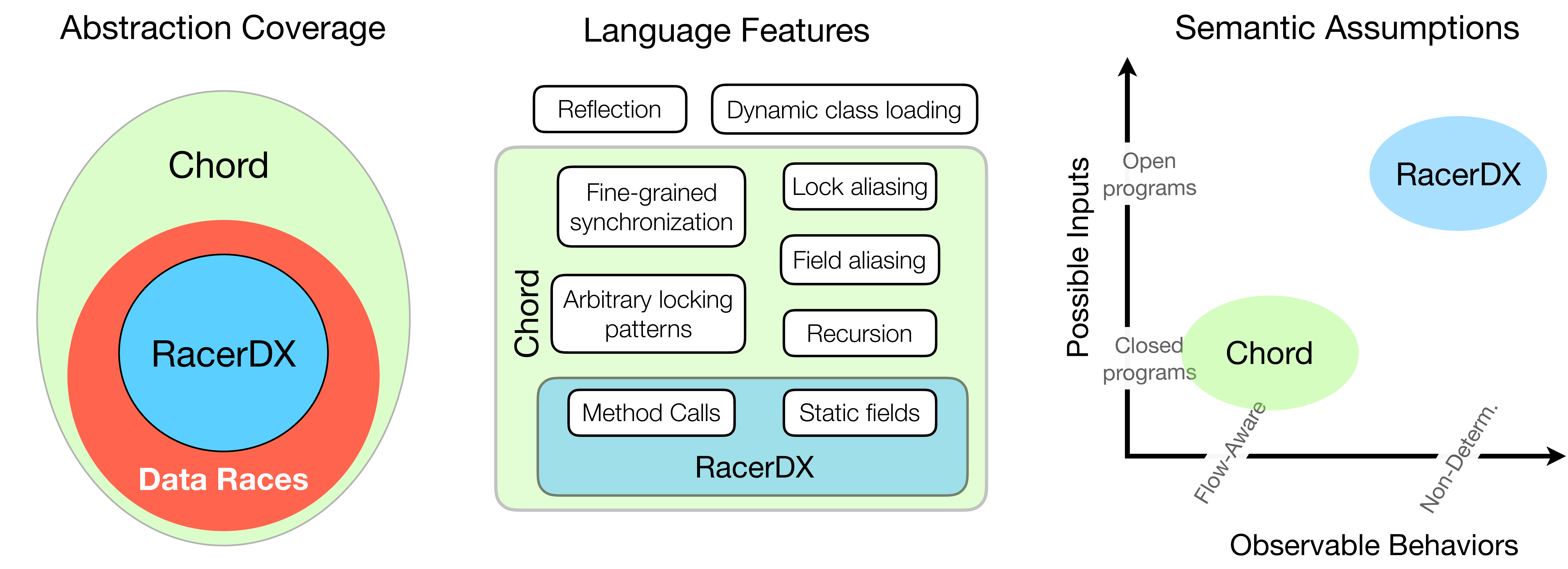}
\caption{The aspects of a race analyzer: abstraction coverage,
  language features, and semantic assumptions.}
\label{fig:circles}
\end{figure}

\begin{itemize}
\item \emph{Abstraction Coverage.}
  The left part of Figure~\ref{fig:circles} demonstrates the relative
  positioning of claims by \chord and \racerdx .
   \chord~ strives to  compute an over-approximation, so that its data race reports should be a superset of the races than can occur in program executions.
    In contrast, \racerdx targets \emph{under-approximation}, computing a subset of the real races, which is what our True Positives Theorem addresses.


\item \emph{Language Features.}
  Claims about \emph{unconditional} absence of false negatives or
  false positives made with regard to an industrial-scale project and
  a realistic programming language (\eg, Java), are often false for particular analyzers,
  even if their designs are guided by considerations of   soundness or completeness.
%
  %
  The set of  choices of which language features to model
  faithfully or to ignore, which we refer to as \emph{context}, affects
  both the soundness and completeness claims, and frame the analysis' claims, serving
  as signs of warning to the analysis' users.
  The \emph{``Soundiness''} movement~\cite{Livshits-al:CACM15}
  mandates the designers of new static analyses to make the context
  explicit and offers a list of likely causes of unsoundness in
  popular languages and
  runtimes; that is, soundness (over-approximation) can be wrt a restriction (under-approximation) of the chosen concrete semantics.
  Similarly, for analyzers aiming at completeness the context should be made explicit.\footnote{However, we avoid inventing a new term, as both soundiness and completeness-relative-to-assumptions can, as far as we are aware, be formulated as standard soundness or completeness results wrt models capturing the assumptions (which themselves can often be under-approximate or over-approximate models).}
   As an example, the central part of Figure~\ref{fig:circles} shows a
  selection of features, framing the claims made for \chord
  (with extensions to its points-to analysis
  component by~\citet{Naik-Aiken:POPL07}) and \racerdx, correspondingly.
  In the latter case, it corresponds to the
  assumptions~\ref{asm:class}--\ref{asm:recur} we made in
  Section~\ref{sec:language}.

\item \emph{Semantic Assumptions.}
  The main reasons for having the claims framed properly with
  regard to the context, is to make sure that the class of the
  programs satisfying the context assumptions, can be faithfully
  represented by a semantics, with respect to which the soundness or completeness
  claims were established.
  But the choice of the semantics itself, as well as
  non-language-specific assumptions, is rarely questioned, or exercised
  to the advantage of the analysis designer.
  For instance, a soundness result for a data race analysis,
  proven modulo a rigorously stated context, can be invalidated if,
  for instance, one assumes a version of a relaxed-memory
  semantics~\cite{Kang-al:POPL17} instead of more traditional
  sequential consistency~\cite{Lamport:TC79}.
  At the same time, a completness result would be
  \emph{simpler} to establish in this setting, as it would allow one
  to \emph{choose from more program behaviors}!
  The right part of Figure~\ref{fig:circles} depicts this remarkable
  variance of the soundness/completeness results with regard to the chosen
  foundations, determining the set of allowed behaviors and inputs.

  An analysis aiming for soundness will be easier to prove in a more
  restrictive semantics and in closed-world assumptions, which is the
  case for \chord.
  In contrast, a formal result for an analysis aiming for completeness
  (\eg, \racerdx) will be easier (or even possible) to state in a less
  restrictive, possibly non-deterministic semantics (such as what we
  assumed via~\ref{asm:control}), in the presence of an open world
  assumption (which we exploited in proofs in
  Sections~\ref{sec:nfp-comp} and~\ref{sec:nfp-reconstr}).

  As far as we are aware, this observation is novel, although in
  retrospect not very surprising.
\end{itemize}

\noindent
In the framework of Figure~\ref{fig:circles}, our True Positives
Theorem delivers \wrt claims on relating to the set of unwelcome
behaviors (No-False-Positives), framed within a particular context
\wrt supported language
features~(assumptions~\ref{asm:class}--\ref{asm:recur}) for carefully
chosen relaxed semantic assumptions~\ref{asm:control}.


\paragraph{Over-of-Under and Under-of-Over}
\label{sec:over-under}

In the light of the suggested classification the over-of-under and
under-of-over characterization of an analyzer is a matter of
positioning the \emph{claims} and choosing the \emph{foundations}.
Analyzers claiming a version of soundness to prevent bugs, such as
\chord, typically establish this in somewhat restricted
foundations (\eg, sequentially-consistent semantics with no
non-determinism, closed-world assumptions), but over-approximates the
set of unwelcome behaviors, therefore delivering the over-of-under
result.
Conversely, analyzers aiming for completeness to detect bugs, such as \racerdx,
choose to consider a more relaxed semantics, while under-approximating
the set of the bugs that they can report, going only for
high-confidence ones, which corresponds to the under-of-over result.

Regardless of whether the over-of-under or under-of-over approach is
chosen, how can one rigorously validate the set of choices made with
regard to the supported features and semantic assumptions?
In the case of soundness claims, one can instrument code with checks that
reflect the sources of deliberate
unsoundness~\cite{Christakis-al:VMCAI15} and compute the fraction of
the code that is analyzed soundly.
Another way, which, in our experience, worked well for completeness/TP claims,
is to \emph{instrument the analysis itself}, so that it would not
report bugs that it is not absolutely adamant about---precisely what
is achieved by our notion of wobbliness/stability from
Section~\ref{sec:analysis}.\footnote{That said, wobbliness does not
  account for all of our language features, such as, \eg, well-scoped
  locking (\ref{asm:lock}).}

\paragraph{Relevance of Theory}
\label{sec:wages-soundness}

How can we establish that the formal analysis of an analyzer is relevant in practice, especially when the contextual assumptions for its theorems might be violated?

For analyzers that come with a proof of soundness, the criterion often used is that of  \emph{precision}, \ie, that the analyzer does not report too many false
positives. The rate of true positives can be improved by chaining the
main analysis with a suitably tailored
\emph{pre-analysis}~\cite{Oh-al:PLDI14}, or by involving a user into the loop
and taking her feedback into account~\cite{Raghothaman-al:PLDI18}.

Our approach, which addresses this challenge for an completeness-style
analyzer, is a complementary one and assesses whether the analysis's
\emph{generality} does not deteriorate (\ie, whether it does not
ignore too many plausible bugs).
To do this, we conducted an empirical comparison to an already
deployed tool, which has been considered highly
impactful by industry standards~\cite{Blackshear-al:OOPSLA18}, and
demonstrate that a rigorously proven version does not do too much
worse (Section~\ref{sec:evaluation}).

\subsection{Theorems for Completeness-oriented Race Detectors}

Several results on  proven \emph{dynamic} concurrency analyses
have been published recently. The pioneering work
by~\citet{Sadowski-al:WMM08} described a partial verification of some
properties of the \tname{Velodrome}~\cite{Flanagan-al:PLDI08} dynamic
atomicity checker, but considered only a model of the algorithm, not
its actual implementation.
The work by~\citet{Mansky-al:CPP17} formalized the notion of run-time
data races on memory locations and proved the soundness and
completeness of the specification rules for a number of idealized
implementations of race detection analyses, including
\tname{FastTrack}~\cite{Flanagan-Freund:PLDI09}.
Finally, \citet{Wilcox-al:PPoPP18} verified a more
realistic implementation of \tname{FastTrack}, showing that it is as fast
as or faster than existing state-of-the-art non-verified
\tname{FastTrack} implementation.

The paper of \citet{Blackshear-al:OOPSLA18} compares \racerd to
several dynamic analyses.  We are not making claims here about the
relative usefulness of the dynamic and static techniques, except to
say that they have often been found to be complementary.

To the best of our knowledge, our result is the first formal proof of
a No-False-Positives property for a \emph{static}
concurrency analysis.

\subsection{Proving Static Analyzers Complete}
The subject of completeness is well-studied in the Abstract
Interpretation community in the context of over-approximating (soundness)
analyzers. For instance, the work by~\citet{Ranzato:VMCAI13}
demonstrates how completeness can be crucial for designing static
analyzers for a number of common intra-procedural properties (\eg,
signs, constant propagation, polyhedra domains, \etc), encouraging one
to reason about the completeness properties of their underlying
abstract domains.
\citet{Giacobazzi-al:POPL15} go even further, providing a \emph{proof
  system} for showing that the result of a certain analysis on a
particular given program is precise.
Those works address mostly numerical properties.
As far as we know, our work is the first to address
completeness of an abstraction and a single-threaded abstract
semantics for detection of concurrent races.

%% file: conclusion.tex
\section{Conclusions}
\label{sec:conclusion}


In this paper we have formulated and proven a True Positives Theorem for an idealized static race detector, motivated by one that has been proven to be effective in production at Facebook.
We have also provided an empirical evaluation of the distance between the idealized analyzer and the in-production version.

It is important to emphasize that we don't view the primary role of the True Positives Theorem as providing guarantees to programmers. Rather, it clarifies to the analysis designer the nature of the analysis algorithm: its purpose is to provide understanding and guide design of current and future analyzers.  This kind of role of theory is complementary to the goal of confirming properties of analyses after they have been finished, and is a relative of the corresponding role of semantics in programming language design advocated by ~\citet{Tennent77}.  Indeed, the theorem was not formulated before~\citet{Blackshear-al:OOPSLA18} implemented \racerd, as a specification for them to meet. Rather, the theorem arose in response to the fact that the tool did not satisfy a standard over-approximation (soundness for bug prevention) or under-approximation (completeness, or soundness for testing) result, and yet was effective in practice. And, as we mentioned, the statement of the theorem then impacted the further development of \racerd.

We have not formulated a general theory for the TP theorem, into which many analyzers would fit.  The reason is not that we can't see how to make {\em some}\/ such theory, but rather that we are being careful not to engage in quick generalisation from few examples from practice. Beyond the {\em shape}\/ of the TP theorem (over-of-under), the more significant issue of the nature of the assumptions supporting a meaningful theorem is one that we would like to see validated by other in-production analyses, before generalization.
Formal theorems about static analyses for data race detection have concentrated on soundness results, where here we considered completeness. In both cases contextual assumptions often need to be stated which concern, \eg, language features treated and ignored or other assumptions reflected in the abstractions.

The way we compose over- and under-approximations is  consistent with the spirit of abstract interpretation.  Although abstract interpretation can be used to formulate and prove vanilla soundness and completeness results (over and under on their own), perhaps its greater value comes from composing and comparing different abstractions, to give finer insight on the nature of analyses than the vanilla results (which often do not literally hold) would do.  It is obvious that one {\em might} consider any sequence of over- and under-approximations.
A modest suggestion of this work is that perhaps under-of-over is deserving of more attention, as a way to study static analyses that are designed for bug catching rather than prevention.\footnote{This suggestion is consistent with bounded model checking and symbolic execution for testing, which can often be seen as computing {\em either}\/ over-of-under or under-of-over.}

\begin{acks}  
 We are grateful to Don Stewart for his encouragement and support
 for this research project.
 We thank Francesco Logozzo for his comments on the formalisation of
 the analysis completeness.
 Finally, we wish to thank the POPL'19 PC and AEC reviewers for the
 careful reading of the paper, and for many insightful comments and
 suggestions.
\end{acks}

%% file: appendix-comp.tex
\section{Proofs for Section~\ref{sec:nfp-comp} (Analysis Completeness)}
\label{app:comp}

\begin{alemma}{\ref{lm:add}}
$\alpha$ is additive (\ie, preserves lub's) with respect to
$\cup_{\TTraces}$ and $\lub_{\Abs}$.
\end{alemma}
\begin{proof}
  Both $\TTraces$ and $\Abs$ are complete lattices. For any
  $T_1, T_2 \in \TTraces$, we have: {\small{
\begin{align*}
  \alpha(T_2 \cup T_2) = \Lub_{\tr \in T_1 \cup T_2}(\fst(\exec~\ab{\tr}~\bota)) =
  \Lub_{\tr \in T_1}(\fst(\exec~\ab{\tr}~\bota)) \lub \Lub_{\tr \in T_2}(\fst(\exec~\ab{\tr}~\bota)) =
  \alpha(T_1) \lub \alpha(T_2)
\end{align*}
}}
\end{proof}

\begin{alemma}{\ref{lm:comcom}}
For any non-empty WF trace $\tr \in \Traces$, sets $W$, $A$, number
$L$, and a simple command~$c$, which is not~$\cmdpop{}$, such
that
(a) $\angled{W, L, A} = \alpha(\set{\tr})$,
%
%
(b)
$\sem{c}~\llast{\tr} = \set{[\astate]~|~\astate \text{ is an execution
    state}}$, where $\llast{\tr}$ is the configuration
$\angled{S, h, L}$ of the last element of $\tr$, the following
\textit{{holds}}:
\[
\small
{\asem{c}{\angled{W, L, A}}} =
{\alpha \left(\bigcup\set{\tr \cons \astate \mid [\astate] \in \sem{c}~\llast{\tr}} \right)}.
\]
\end{alemma}
\begin{refproof}{\emph{Proof}}
\label{proof:comcom}
  If $c = \cmdskip$ the proof is trivial, as
  $\sem{c}~\llast{\tr} = \set{\epsilon}$, so the side condition
  $\sem{c}~\llast{\tr} = \set{[\astate]}$ does not hold.
  For other simple commands, let us notice that the
  condition $\sem{c}~\llast{\tr} = \set{[\astate ~|~ \ldots]}$
  literally means that for any $\astate$, the trace $t \cons \astate$
  (since $c$ is a simple command, the result of its successful
  execution is \emph{always} a set of a single-state traces, an
  infinite one in the case of $\cmdalloc$ --- \cf
  Figure~\ref{fig:traces}) is well-formed.

  We therefore proceed by case analysis of the clauses of the \racerd
  analysis definition in Figure~\ref{fig:analysis} and
  trace-collecting semantics in Figure~\ref{fig:traces}, assuming
  $\angled{W', L', A'} = \fst(\exec~\ab{\tr}~ \bota)$, so by the assumptions
  $\angled{W, L, A} = \angled{W', L', A'}$, so $L = L'$ \etc.
 Since there are no $\cmdpush{}/\cmdpop{}$ in $\ab{\tr}$, we do not
 have to account for the effect of a substitution $\theta$,
 ``accumulated'' by $\exec$ while proceccing $\ab{\tr}$.

  \vspace{3pt}
  \noindent\textbf{Case: $c= (x := y)$}.
  In this case, there is exactly one resulting state $\astate$, and
  $\ab{\tr \cons \astate} = \ab{\tr} \cons \angled{x := y}$.
  Therefore,
  $\alpha \left(\set{\tr \cons \astate} \right) =
  \fst(\exec~\ab{\tr}\cons \angled{x := y}~\bota) = \angled{W' \cup
    \set{x, y}, L, A'}$. At the same time, by the definition in
  Figure~\ref{fig:analysis},
  $\asem{(x := y)}{\angled{W, L, A}} = \angled{W \cup \set{x, y}, L,
    A}$ $= \angled{W' \cup \set{x, y}, L, A'}$.

  \vspace{3pt}
  \noindent\textbf{Case: $c= (x := \pi)$}.
  There is exactly one resulting state $\astate$,
  $\alpha \left(\set{\tr \cons \astate} \right) = \fst(\exec~\ab{\tr}
  \cons \angled{x := \pi}~\bota) = \angled{W' \cup \set{x, \pi}, L, A'
    \cup \angled{x := \pi, L}}$. Since
  $\asem{(x := \pi)}{\angled{W, L, A}} = \angled{W \cup \set{x, \pi},
    L, A \cup \angled{x := \pi, L}}$, the desired equality
  holds.

  \vspace{3pt}
  \noindent\textbf{Case: $c= (\pi := x)$}.
  There is exactly one resulting state $\astate$. We have
  $\alpha \left(\set{\tr \cons \astate} \right) = \fst(\exec~\ab{\tr}
  \cons \angled{\pi := x}~\bota) = \angled{W' \cup \set{x, \pi}, L, A'
    \cup \angled{\pi := x, L}}$.
  By the definition,\\
  $\asem{(\pi := x)}{\angled{W, L, A}} = \angled{W \cup \set{x,
      \pi}, L, A \cup \angled{\pi := x, L}}$, so the desired
  equality holds.

  \vspace{3pt}
  \noindent\textbf{Case: $c= (x:=\cmdalloc)$}.
  Due to the non-determinism of object allocation, there is an
  infinite set of resulting states $\astate$, yet the abstraction
  $\alpha$ collapses all of them into the same syntactic statement
  $\angled{x:=\cmdalloc}$. Therefore,
  \[
\alpha \left(\bigcup_{\astate, [\astate] \in
    \sem{x:=\cmdalloc}~\llast{\tr}} \set{\tr \cons \astate} \right) =
\fst(\exec~\ab{\tr} \cons \angled{x:=\cmdalloc}~\angled{W', L, A'}) = \angled{W' \cup \set{x}, L, A'}.
  \]
  By the definition~\ref{fig:analysis},
  $\asem{(x:=\cmdalloc)}{\angled{W, L, A}} = \angled{W \cup
    \set{x}, L, A}$, so the equality holds.

  \vspace{3pt}
  \noindent\textbf{Case: $c= \cmdlock{}$}.
  There is exactly one resulting state $\astate$. We have
  $\alpha \left(\set{\tr \cons \astate} \right) = \fst(\exec~\ab{\tr}
  \cons \angled{\cmdlock{}}~\bota) = \angled{W', \addlock{L}{1}, A'}$.
  By the definition,
  $\asem{\cmdlock{}}{\angled{W, L, A}} = \angled{W, \addlock{L}{1},
    A}$, so the desired equality holds.

  \vspace{3pt}
  \noindent\textbf{Case: $c= \cmdunlock$}.
  For the only resulting state $\astate$, we have
  $\alpha \left(\set{\tr \cons \astate} \right) = \fst(\exec~\ab{\tr} \cons
  \angled{\cmdunlock}~\bota) = \angled{W', \max(L-1, 0), A'}$.
  Since
  $\asem{\cmdunlock}{\angled{W, L, A}} = \angled{W, \max(L-1, 0),
    A}$, the desired equality holds.
\QED\end{refproof}


\begin{lemma}[$\asem{\cdot}$ distributes over $\lub$]
\label{lm:anal-lub}
Let $\angled{W, A, L} = \Lub_{i \in I}\angled{W_i, A_i, L_i}$ for
some index set $I$. Then for any simple command $c$,
\[
\asem{c}{\angled{W, A, L}} = \Lub_{i \in I}
\left(\asem{c}{\angled{W_i, A_i, L_i}}\right).
\]
\end{lemma}
\begin{proof}
  The proof is by case analysis on the command $c$. In general, let us
  notice that, as defined, $\asem{\cdot}$ only \emph{adds} elements to
  the $W$ and $A$.

  In the case of $A$-component, the locking set $L'$ to be paired with
  $c$ (as in $\angled{c, L'}$ added by the analysis) is taken from the
  initial $L$ component. Since the $\lub$ on $\Lock$ is defined as
  $\max$ there should exist $j \in I$, such that $L_i = L$. Therefore,
  no element $\angled{L', A}$ will be ``missed'' on the left hand side
  of the equality in question for
  $L' \in \set{L, \max(L+1, L), \min(L-1, L)}$, depending on $c$,
  matching the corresponding entry on the right.

  Finally, for the same reason, the $L$-component on the left side of
  the equality being proven is no large than the one on the right,
  as per definition of the $\lockingsymb$ subroutine (which does not
  depend on $W$ and $A$) and taking $\lub$ to be the $\max$.
\end{proof}


\begin{alemma}{\ref{lm:comcom2}}
For any set $T$ of non-empty well-formed traces, $W$, $L$, $A$, and
simple command~$c$, which is not $\cmdpop{}$, such that
(a) $\angled{W, L, A} = \alpha(T)$,
%
%
(b) for any $\tr \in T$, $\sem{c}~\llast{\tr} = \set{[\astate]~|~\astate \text{ is an
        execution state}}$,
then the following holds:
\[
\small
{\asem{c}{\angled{W, L, A}}} =
{\alpha \left(\bigcup\set{\tr \cons \astate \mid \tr \in T, [\astate] \in \sem{c}~\llast{\tr}} \right)}.
\]
\end{alemma}
\begin{refproof}{\emph{Proof}}
\label{proof:comcom2}
  The proof is by case analysis of the clauses of the
  \racerdx analysis definition in Figure~\ref{fig:analysis} and
  trace-collecting semantics in Figure~\ref{fig:traces}, assuming
  $\angled{W', L', A'} = \alpha(T) = \Lub_{\tr \in
    T}(\fst(\exec~\ab{\tr}~\bota))$, so by the assumptions
  $\angled{W, L, A} = \angled{W', L', A'}$. Also, any
  trace $\tr \cons \astate$ (where $\astate$ is and state from the
  lemma assumptions) is well-formed.

  From Lemma~\ref{lm:add} and Definition~\ref{def:alpha}, we have
  \[\alpha \left(\bigcup_{\tr \in T, [\astate] \in
      \sem{c}~\llast{\tr}} \set {\tr \cons \astate} \right) =
  \Lub_{\tr \in T, [\astate]\in
      \sem{c}~\llast{\tr}} \fst(\exec \left(\ab{\tr \cons \astate} \right))
  \]
  Let us notice now that for any $p \in W'$ there exists $\tr \in T$,
  such that $p \in (\fst(\exec~{\ab{\tr}}~\bota))_W$ (where $X_W$ is the
  projection to the $W$-component of an abstract domain element $X$).
  Similarly, for any $\angled{c, L} \in A'$ there exists $\tr \in T$,
  such that $\angled{c, L} \in (\fst(\exec~{\ab{\tr}}~\bota))_A$. Finally,
  since $\lub$ is defined as $\max$ on $\Lock$, there exists
  $\tr \in T$, such that $L = (\fst(\exec~{\ab{\tr}}~\bota))_L$.
  Therefore, we can represent, $\angled{W, L, A}$ as follows:
  \[
    \angled{W, L, A} = \Lub_{\tr \in T} \angled{W \cap W_{\tr},
      \min(L, L_{\tr}), A \cap A_{\tr}}, ~\text{where}~
    \angled{W_{\tr}, L_{\tr}, A_{\tr}} = \fst(\exec~\ab{\tr}~\bota)
  \]
  Now, let us denote via $\angled{W^{\tr}, L^{\tr}, A^{\tr}}$ the
  individual ``$\tr$-projected fragments'' of $\angled{W, L, A}$ so
  that for any $\tr \in T$,
  $\angled{W^{\tr}, L^{\tr}, A^{\tr}} = \angled{W \cap W_{\tr},
    \min(L, L_{\tr}), A \cap A_{\tr}}$. Altogether
  $\angled{W^{\tr}, L^{\tr}, A^{\tr}}$s deliver the overall value
  $\angled{W, L, A}$ in the $\Lub$-union above.
  This way, for each $\tr \in T$,
  $\angled{W^{\tr}, L^{\tr}, A^{\tr}} \pre \fst(\exec~\ab{\tr}~\bota) =
  \alpha(\set{\tr})$ and $L^{\tr} = \alpha(\set{\tr})_L$ (since
  $L^{\tr} \leq L$). Via Lemma~\ref{lm:comcom}, we obtain, for each
  $\tr \in T$:
  \begin{equation}
    \label{eq:lm1}
     \asem{c}{\angled{W^{\tr}, L^{\tr}, A^{\tr}}} \pre \alpha
    \left(\bigcup_{\astate, [\astate] \in \sem{c}~\llast{\tr}} \set{\tr \cons \astate} \right)
  \end{equation}
  Therefore, via Lemma~\ref{lm:anal-lub}, we have
  \begin{align*}
   \asem{c}{\angled{W, L, A}} &=
   \Lub_{\tr \in T}\asem{c}{\angled{W^{\tr}, L^{\tr}, A^{\tr}}} &
    \text{ via~Lemma~\ref{lm:anal-lub}}
    \\
   &= \Lub_{\tr \in T} \alpha \left(\bigcup_{\astate, [\astate] \in \sem{c}~\llast{\tr}} \set{\tr \cons \astate} \right)
  & \text{ via~\eqref{eq:lm1} and properties of $\lub$}
    \\
   &=  \alpha \left(\bigcup_{\tr \in T, [\astate] \in \sem{c}~\llast{\tr}} \set{\tr \cons \astate} \right) &
  \text{ via~Lemma~\ref{lm:add}.}
  \end{align*}

\QED\end{refproof}


\begin{alemma}{\ref{lm:balanced}}
If $C$ is a compound program with balanced locking, and $\angled{W', L',
  A'} = \asem{C} \angled{W, L, A}$. Then $L' = L$.
\end{alemma}
\begin{refproof}{\emph{Proof}}
  By induction on the shape of $C$ and the definition in
  Figure~\ref{fig:analysis}, counting the number of ``pending''
  $\cmdunlock{}/\cmdlock{}$ statements, which eventually balance each
  other, so all additions/subtractions to the $L$ component of the
  abstract state cancel each other.
\end{refproof}


\begin{alemma}{\ref{lm:com-compound}}
For any compound program $C$ with balanced locking and no method
calls, the initial components $W_0$, $L_0$, $A_0$, and a set of
well-formed non-empty traces $T$, such that
$\angled{W_0, L_0, A_0} = \alpha(T)$,
  \[
     \asem{C} \angled{W_0, L_0, A_0}
    =
    \begin{cases}
        \alpha\left(\bigcup_{\tr \in T} \set{\tr \append \tr' ~|~ \tr' \in
        \sem{C}~\llast{\tr}}\right)
        & \text{ if}~~T \neq \emptyset,
        \\[4pt]
        \alpha\left(\sem{C} \angled{S, h, 0} \right)
        & \text{ otherwise, for any well-behaved $S,h$,}
      \end{cases}
  \]
  where $\llast{\tr}$ is well-defined, as $T$ consists of non-empty traces.
\end{alemma}
\begin{refproof}{\emph{Proof}}
\label{proof:com-compound}
  The proof is by the backwards induction on the shape of the program
  $C$, with the use of Lemma~\ref{lm:comcom2} for simple commands.

\vspace{3pt}
\noindent\textbf{Case: $C = c$ (simple command)}.
The proof follows Lemma~\ref{lm:comcom2} for the subcase
$T \neq \emptyset$. Otherwise, $\alpha(T)=\alpha(\emptyset)=\bota$ and
as $\angled{W_0, L_0, A_0} = \alpha(T)$, $\angled{W_0, L_0, A_0}=\bota$.
Then, the proof is similar, but simpler, as
the statement of interest reduces to
$\asem{c}{\bota} = \alpha (\sem{c}~\angled{S,
  h, 0})$, which is established by case analysis on $c$ and the observation
  that due to well-behavedness, $\sem{c}~\angled{S,h, 0}\neq\emptyset$.

\vspace{3pt}
\noindent\textbf{Case: $C = C'; c$ ($c$ is not a method call)}.
For the first sub-case, follows immediately from
Lemma~\ref{lm:comcom2} with
$T = \bigcup_{\tr \in T}\sem{C'}~\llast{\tr}$ and
$\angled{W, L, A} = \asem{C'} \angled{W_0, L_0, A_0}$. For the second
sub-case, it is similar, via Lemma~\ref{lm:comcom2} with
$T = \sem{C'}~\angled{S, h, 0}$ and
$\angled{W, L, A} = \asem{C'} \angled{W_0, L_0, A_0}$.
Notice that while $C$ and, by consequently, $C'$ enjoy balanced
locking (which only matters for control-flow constructs), both $C$ and
$C'$ may have non-matched $\cmdlock{}$-statements.

\vspace{3pt}
\noindent\textbf{Case: $C = C'; (\cmdif{*}{C_1}{C_2})$}.
Let $\angled{W, L, A} = \asem{C'} \angled{W_0, L_0, A_0}$,
$\angled{W_i, L, A_i} = \asem{C_i} \angled{W, L, A}$ for
$i \in \set{1, 2}$ (the equality of $L$-components is thanks to the
balancing locking assumptions and Lemma~\ref{lm:balanced}). Let us
focus on the resulting abstract state
$\angled{W_1, L, A_1} \lub \angled{W_2, L, A_2}$. By induction
hypothesis, every element of $W_1$ is also in
$\alpha{\sem{C_1} \ldots}$; the inclusion also holds for elements of
$A_1$, and similarly for the corresponding result of the second branch
$C_2$. By Definition~\ref{def:lub} of $\lub$ in the desired equality
indeed holds (\ie, the analysis correctly under-approximates the
abstraction of the concrete semantics). Remarkably, even though it
takes the minimum of $L$-components, the analysis \emph{does not lose
  soundness} in the presence of conditionals because of the balanced
locking assumption.

\vspace{3pt}
\noindent\textbf{Case: $C = C'; (\cmdwhile{*}{C''})$}.
For this case, we will employ the enforced requirement that
$\cmdlock{}$ and $\cmdunlock{}$ only appear in balanced pairs.
First, let us pose
$\angled{W, L, A} = \asem{C'} \angled{W_0, L_0, A_0}$ and Let
$\angled{W', L, A'} = \asem{C''} \angled{W, L, A}$ (thanks to
well-scoped locking). Notice that, while our analysis functional
$\asem{\cdot}$ is not monotone in general, it \emph{is} monotone on
programs with balanced locks, because it is monotone \wrt the first
and the last components of the abstract state, and preserves the
second component, as established by Lemma~\ref{lm:balanced}.
Furthermore, we can see that after taking just one execution of
$\asem{C''} \angled{W, L, A}$ we already reach a fixed point: no
further accesses are going to be added (ang by inductive hypothesis,
\emph{all} concretely occurring accesses are recorded) and the
$L$-component will not change. Therefore, the analysis represents
precisely computation of the fixpoint of
$\ab{F} = \lambda D, \asem{C''}~D$, so the desired equality
holds.
\QED\end{refproof}


The following lemma ensures that an abstract
execution of ANF methods can be modelled via direct
\emph{substitutions} rather with \emph{stack manipulation}:

\begin{lemma}[ANF, substitution, and abstraction]
\label{lm:anf-subst}
For any non-empty well-formed trace $\tr \in \Traces$, a method $\m$,
a vector of expressions $e_1, \ldots, e_n$, a configuration
$\angled{s \cons S, h, L} = \llast{\tr}$, such that
(a)~$\forall i, 1 \leq i \leq n, \sem{e_i}_{s, h}$ are defined,
(b)~$\hat s = s_{\nil}[\argm_1 \mapsto
    \eval{\e_1}{s,h}]\cdots[\argm_n\mapsto\eval{\e_n}{s,h}]$,
(c)~$ \tr_{\text{push}} = [\angled{\cmdpush{\e_1,\ldots,\e_n}, \hat s
      \cons s \cons S, h, L}]$, and
(d)~$\body{\m}$ has no nested calls and features well-balanced locking,
the following holds:
{{
\begin{align*}
  \alpha \left(\bigcup\set{\tr \append \tr_{\text{push}} \append \tr'
    \mid {\tr' \in \sem{\body{\m}; \cmdpop{}}~\angled{\hat
      s \cons s \cons S, h, L}}} \right)
&=
\\
\alpha \left(\bigcup \set{\tr \append \tr' \mid {\tr' \in \sem{\body{\m}[\many{e_i/\argm_i}]}~\llast{\tr}}} \right)
\wcup
\wb
\end{align*}
}}
\end{lemma}

In other words, one can replace computation of the \emph{method $\m$'s call}
abstraction via computing the abstraction of the \emph{$\m$'s body}
with substituted expressions.

\begin{refproof}{\emph{Proof}}
\label{proof:anf-subst}
  Starting from the expression on the left hand side of the equality,
  let us notice that
{\small{
\begin{align*}
\alpha \left(\bigcup\set{\tr \append \tr_{\text{push}} \append \tr' \mid
\tr' \in \sem{\body{\m}; \cmdpop{}}~\angled{\hat s \cons s \cons S, h, L}} \right) &= \text{(via Def.~\ref{def:alpha})}
\\
\alpha \left(\bigcup\set{\tr \append \tr_{\text{push}} \append \tr' \mid
\tr' \in \sem{\body{\m}}~\angled{\hat s \cons s \cons S, h, L}}
  \right)
%
&= \text{(via Def.~\ref{def:alpha})}
\\
\alpha \left(\bigcup\set{\tr \append \tr' \mid
\tr' \in \sem{\body{\m}[\many{e_i/\argm_i}]}~\angled{s \cons S, h, L}}
  \right) \wcup \wb &= \text{(by assumption)}
\\
\alpha \left(\bigcup\set{\tr \append \tr' \mid
\tr' \in \sem{\body{\m}[\many{e_i/\argm_i}]}~\llast{\tr}}\right) \wcup \wb
\end{align*}
}}
\QED\end{refproof}

\begin{alemma}{\ref{lm:anf-sum}}
For a method $\m$, a vector of expressions $e_1, \ldots, e_n$, a
configuration $\angled{s :: S, h, L} = \llast{\tr}$ such that
  \begin{itemize}
  \item $\forall i, 1 \leq i \leq n, \sem{e_i}_{s, h}$ are defined, and
  \item $\body{\m}$ has no nested calls and features well-balanced
    locking, the following holds:
  \end{itemize}
{\small{
\[
\begin{array}{r@{\ }c@{\ }l}
\asem{\m(\e_1,\ldots,\e_n)} \angled{W, L, A} &{\;=\;}&
\angled{W', L', A'} \wcup \wb, \text{where} \\[3pt]
&&
\angled{W', L', A'} = \asem{\body{\m}[\many{\e_i/\argm_i}]}~\angled{W, L, A}
\end{array}
\]
}}
\end{alemma}

\begin{proof}
  By induction on the shape of $\body{m}$ as a compound command
  (Figure~\ref{fig:analysis}).
\end{proof}
%

\begin{lemma}[Analysis is complete for method calls (per-trace)]
\label{lm:mcom1}
For any non-empty well-formed trace $\tr \in \Traces$, sets $W$, $A$,
number $L$, a method $\m$ in ANF, a vector of expressions
$e_1, \ldots, e_n$, a configuration
$\angled{s :: S, h, L} = \llast{\tr}$ such that
  \begin{itemize}
  \item $\forall i, 1 \leq i \leq n, \sem{e_i}_{s, h}$ are defined,
  \item $\angled{W, L, A} = \alpha(\set{\tr})$,
  \item $ \tr_{\text{push}} = [\angled{\cmdpush{\e_1,\ldots,\e_n}, \hat s
      \cons s \cons S, h, L}]$
  \item $\hat s = s_{\nil}[\argm_1 \mapsto \eval{\e_1}{s,h}]\cdots[\argm_n\mapsto\eval{\e_n}{s,h}]$,
  \item $\body{\m}$ has no nested calls and enjoys the well-scoped locking.
  \end{itemize}
\noindent
Then the following holds:
\[
{\asem{\m(\e_1,\ldots,\e_n)}{\angled{W, L, A}}} =
{\alpha \left(\bigcup
\set{\tr \append \tr_{\text{push}} \append \tr'
\mid
\tr' \in \sem{\body{\m}; \cmdpop{}}~\angled{\hat s \cons s \cons S, h, L}} \right)}.
\]
\end{lemma}
\begin{proof}
By Lemma~\ref{lm:anf-subst}, it is sufficient to prove that
{\small{
\[
{\asem{\m(\e_1,\ldots,\e_n)}{\angled{W, L, A}}} =
{\alpha \left(\bigcup \set{\tr \append \tr' \mid
      \tr' \in \sem{\body{\m}[\many{e_i/\argm_i}]}~\llast{\tr}}
  \right) \wcup \wb},
\]
}}
\noindent 
which \emph{almost} matches the statement
of~Lemma~\ref{lm:com-compound} if we take
$C := \body{\m}[\many{e_i/\argm_i}]$.
By Lemma~\ref{lm:anf-sum}, we need to prove that
\[
\left(\asem{C}~\angled{W, L, A}\right) \wcup \wb = \alpha \left(\bigcup \set{\tr
    \append \tr' \mid \tr' \in
    \sem{C}~\llast{\tr}} \right) \wcup \wb,
\]
for $C = \body{\m}[\many{\e_i/\argm_i}]$.
%
%
%
That is, now it is sufficient to prove the following goal
\[
\asem{C}~\angled{W, L, A} = \alpha \left(\bigcup \set{\tr
    \append \tr' \mid \tr' \in
    \sem{C}~\llast{\tr}} \right),
\]
This is precisely the conclusion of Lemma~\ref{lm:com-compound}, for
$T = \set{\tr}$, so applying it completes the proof.
\end{proof}


\begin{lemma}[Analysis is complete for \emph{nested} method calls (per-trace)]
\label{lm:mcom2}
For any non-empty well-formed trace $\tr \in \Traces$, sets $W$, $A$,
number $L$, a method $\m$, a vector of expressions
$e_1, \ldots, e_n$, a configuration
$\angled{s :: S, h, L} = \llast{\tr}$ such that
  \begin{itemize}
  \item $\forall i, 1 \leq i \leq n, \sem{e_i}_{s, h}$ are defined,
  \item $\angled{W, L, A} = \alpha(\set{\tr})$,
  \item $ \tr_{\text{push}} = [\angled{\cmdpush{\e_1,\ldots,\e_n}, \hat s
      \cons s \cons S, h, L}]$
  \item $\hat s = s_{\nil}[\argm_1 \mapsto \eval{\e_1}{s,h}]\cdots[\argm_n\mapsto\eval{\e_n}{s,h}]$,
  \item $\body{\m}$ and all its callees are in ANF and enjoy the
    well-scoped locking.
  \end{itemize}
\noindent
Then the following holds:
\[
{\asem{\m(\e_1,\ldots,\e_n)}{\angled{W, L, A}}} =
{\alpha \left(\bigcup
\set{\tr \append \tr_{\text{push}} \append \tr'
\mid
\tr' \in \sem{\body{\m}; \cmdpop{}}~\angled{\hat s \cons s \cons S, h, L}} \right)}.
\]
\end{lemma}
\begin{proof}
By top-level induction on the size of the call tree and per-call
induction on the structure of the body, with the base case delivered
by Lemma~\ref{lm:comcom}.
\end{proof}


\begin{lemma}[Analysis is complete for \emph{nested} method calls (for sets of traces)]
\label{lm:mcom3}
  For any set $T$ of non-empty well-formed traces, sets $W$, $A$,
  number $L$, a method $\m$, a vector of expressions
  $e_1, \ldots, e_n$, such that
  \begin{itemize}
  \item $\forall i, 1 \leq i \leq n, \sem{e_i}_{s, h}$ are defined,
  \item $\angled{W, L, A} = \alpha(T)$,
  \item for any $\tr \in T, \llast{\tr} = \angled{s :: S, h, L}$ for
    some $s$, $S$,
  \item $\body{\m}$ and all its callees are in ANF and enjoy the
    well-scoped locking.
  \end{itemize}
\noindent
Then the following holds:
{\footnotesize{
\[
  {\asem{\m(\e_1,\ldots,\e_n)}{\angled{W, L, A}}} =
  {
\alpha \left(\bigcup_{
\begin{array}{c}
\tr \in T
\end{array}}
\set{\tr \append \tr_{\text{push}} \append \tr'
\left|
\begin{array}{l}
\llast{\tr} = \angled{s :: S, h, L}
\\
\hat s = s_{\nil}[\argm_1 \mapsto
    \eval{\e_1}{s,h}]\cdots[\argm_n\mapsto\eval{\e_n}{s,h}]
\\
\tr_{\text{push}} = [\angled{\cmdpush{\e_1,\ldots,\e_n}, \hat s
      \cons s \cons S, h, L}] \\
\tr' \in \sem{\body{\m}; \cmdpop{}}~\angled{\hat s \cons s \cons S, h,
  L}
\end{array}
\right.
} \right)
%
}.
\]
}}
\end{lemma}
\begin{proof}
Similar to Lemma~\ref{lm:comcom2}.
\end{proof}



\begin{alemma}{\ref{lm:com-compound-calls}}
For any compound program $C$ with balanced locking and all methods in
ANF, the analysis domain components $W_0$, $L_0$, $A_0$, and a set of
well-formed non-empty traces $T$, such that
$\angled{W_0, L_0, A_0} = \alpha(T)$,
  \[
     \asem{C} \angled{W_0, L_0, A_0}
    =
    \begin{cases}
        \alpha\left(\bigcup_{\tr \in T} \set{\tr \append \tr' ~|~ \tr' \in
        \sem{C}~\llast{\tr}}\right)
        & \text{ if}~~T \neq \emptyset,
        \\[4pt]
        \alpha\left(\sem{C} \angled{S, h, 0} \right)
        & \text{ otherwise, for any well-behaved $S,h$,}
      \end{cases}
  \]
  where $\llast{\tr}$ is well-defined, as $T$ consists of non-empty traces.
\end{alemma}
\begin{refproof}{\emph{Proof}}
\label{proof:com-compound-calls}
  The proof is by the backwards induction on the shape of the program
  $C$ similarly to the proof of Lemma~\ref{lm:com-compound}, with the
  use of Lemma~\ref{lm:comcom2} for simple commands and
  Lemma~\ref{lm:mcom3} for method calls. Let us consider the only new
  case of method calls:

\vspace{3pt}
\noindent\textbf{Case: $C = C'; \m(\e_1,\ldots,\e_n)$}
By the definition from Figure~\ref{fig:traces}, the call
$\m(\e_1,\ldots,\e_n)$ is ``unfolded'' into
$\cmdpush{\e_1,\ldots,\e_n}; \body{\m}; \cmdpop{}$, so the statement
of Lemma~\ref{lm:mcom3} applies.
For the first sub-case, follows immediately from Lemma~\ref{lm:mcom3}
with $T = \bigcup_{\tr \in T}\sem{C'}~\llast{\tr}$ and
$\angled{W, L, A} = \asem{C'} \angled{W_0, L_0, A_0}$. For the second
sub-case, it is similar, via Lemma~\ref{lm:mcom3} with
$T = \sem{C'}~\angled{S, h, 0}$ and
$\angled{W, L, A} = \asem{C'} \angled{W_0, L_0, A_0}$.
\QED\end{refproof}

%% file: appendix-reconstr.tex
\section{Proofs for Section~\ref{sec:nfp-reconstr} (Racy Trace Reconstruction)}
\label{app:reconstr}

For a path $\pi$ such that $|\pi|>1$ we fix $\trunc{\pi}$ to be the proper prefix of $\pi$ such
that is $|\trunc{\pi}|=|\pi|-1$.

\begin{lemma}
\label{lem:preserve-image-stack-update}
Let $\fs$ be a non-empty sequence of fields, and $x,y$ two variables.
Let $\pi= x.\fs$ and $\rho= y.\fs$.
If $s(x)=s'(y)$ and $h\subseteq h'$ then
$\pathimg{h}{s}{x.\gs} = \pathimg{h'}{s'}{y.\gs}$
and $\lval{x.\gs}{s,h}=\lval{y.\gs}{s',h'}$ for any $\gs\prefix\fs$
(and therefore $\lval{\pi}{s,h}=\lval{\rho}{s',h'}$).
\end{lemma}
\begin{proof}
By induction on the structure of $\gs\prefix\fs$.

If $\gs=\g$ then $\lval{x.\gs}{s,h} = (s(x),\g) = (s'(y),\g) = \lval{y.\gs}{s',h'}$.
Therefore, $\pathimg{h}{s}{x.\gs} = \{ \lval{x.\g}{s,h} \mapsto h(\lval{x.\g}{s,h}) \} =
\{ \lval{y.\g}{s',h'} \mapsto h(\lval{y.\g}{s',h'}) \}$. By $h\subseteq h'$ this
means $\pathimg{h}{s}{x.\gs} = \{ \lval{y.\g}{s',h'} \mapsto h'(\lval{y.\g}{s',h'}) \} = \pathimg{h'}{s'}{y.\gs}$.

Let $\gs= \gs'.\g$. Then
$\lval{x.\gs}{s,h} = (h(\lval{x.\gs'}{s,h}),\g)$. By the inductive hypothesis, this is equal to
$(h(\lval{y.\gs'}{s',h'}),\g)$, and by \mbox{$h\subseteq h'$},
equal to $(h'(\lval{y.\gs'}{s',h'}),\g)$ which is simply $\lval{y.\gs}{s',h'}$. At the same time,
$\pathimg{h}{s}{x.\gs} = \pathimg{h}{s}{x.\gs'} \cup \{ \lval{x.\gs}{s,h} \mapsto h(\lval{x.\gs}{s,h}) \}$.
By the inductive hypothesis, $\pathimg{h}{s}{x.\gs'} = \pathimg{h'}{s'}{y.\gs'}$.
Thus, given that $\lval{x.\gs}{s,h} = \lval{y.\gs}{s',h'}$,  we have
$\pathimg{h}{s}{x.\gs} = \pathimg{h'}{s'}{y.\gs'} \cup \{ \lval{y.\gs}{s',h'} \mapsto h'(\lval{y.\gs}{s',h'}) \} = \pathimg{h'}{s'}{y.\gs}$.
\end{proof}

\begin{lemma}
\label{lem:preserve-acyclicity}
If $\pi$ is acyclic in $s,h$, $s(\rootvar{\pi})=s'(\rootvar{\pi})$ and
$\pathimg{h}{s}{\pi} = \pathimg{h'}{s'}{\pi}$ then $\pi$ is acyclic in $s',h'$.
\end{lemma}
\begin{proof}
First we show that if $s(\rootvar{\pi})=s'(\rootvar{\pi})$ and
$\pathimg{h}{s}{\pi} = \pathimg{h'}{s'}{\pi}$, then for any $\pi'\prefix\pi$,
$\lval{\pi'}{s,h}=\lval{\pi'}{s',h'}$. We do this by induction on $\pi'$.
The base case is immediate,
given $s(\rootvar{\pi})=s'(\rootvar{\pi})$.
Let $\pi'=\pi''.\f\prefix\pi$. We need to show that
 $(h(\lval{\pi''}{s,h}),\f)=(h'(\lval{\pi''}{s',h'}),\f)$, or more simply that
 $h(\lval{\pi''}{s,h})=h'(\lval{\pi''}{s',h'})$. By the inductive hypothesis,
 $\lval{\pi''}{s,h}=\lval{\pi''}{s',h'}$. By the definition of $\pathimg{h}{s}{\pi}$, we have
that $\lval{\pi''}{s,h}\in\dom{\pathimg{h}{s}{\pi}}$, and the required result follows from
the assumption that $\pathimg{h}{s}{\pi} = \pathimg{h'}{s'}{\pi}$.

To show acyclicity in $s',h'$ we need to show that for any $\pi''\prefix\pi'\prefix\pi$,
$h'(\lval{\pi'}{s',h'})\neq\locn{\lval{\pi''}{s',h'}}$. Given the above result,
$\lval{\pi'}{s',h'} = \lval{\pi'}{s,h}$ and $\lval{\pi''}{s',h'}=\lval{\pi''}{s,h}$.
In addition, since $\lval{\pi'}{s,h}\in\dom{\pathimg{h}{s}{\pi}}$ and
$\pathimg{h}{s}{\pi} = \pathimg{h'}{s'}{\pi}$, it follows that
$h'(\lval{\pi'}{s',h'})=h(\lval{\pi'}{s,h})$. But from the acyclicity in $s,h$ follows that
$h(\lval{\pi'}{s,h})\neq\locn{\lval{\pi''}{s,h}}=\locn{\lval{\pi''}{s',h'}}$ as required.
\end{proof}

\begin{lemma}
Let $\pi,\pi'$ be two access paths such that $\pi'\not\prefix\pi$.
In addition, $\pi$ is acyclic and disconnected in $s,h$.
Then, $\lval{\pi'}{s,h}\notin\dom{\pathimg{h}{s}{\pi}}$.
\label{lem:indomain-prefix}
\end{lemma}
\begin{proof}
By induction on the structure of $\pi$.

Base case: suppose $\pi= x.\f$. Suppose, for a contradiction, that
$\lval{\pi'}{s,h}\in\dom{\pathimg{h}{s}{\pi}}$, meaning
$\lval{\pi'}{s,h}=\lval{\pi}{s,h}=(s(x),\f)$.

If $|\pi'|=1$ then it is easy to see that $s(\rootvar{\pi'})=s(\rootvar{\pi})$. If
$\rootvar{\pi'}\neq\rootvar{\pi}$,
then this contradicts the first clause of the disconnectedness of $\pi$.  Otherwise,
$\pi'=\pi$, contradicting the assumption that $\pi'\not\prefix\pi$.

If $|\pi'|>1$ then $\trunc{\pi'}$ is defined, and in particular,
$\lval{\pi'}{s,h}=(h(\lval{\trunc{\pi'}}{s,h}), \f)$, meaning that
$h(\lval{\trunc{\pi'}}{s,h}) = s(x)$, and furthermore,
$h(\lval{\trunc{\pi'}}{s,h})\in\locn{\pathimg{h}{s}{\pi}}$. By the second clause of
disconnectedness, it follows that $\lval{\trunc{\pi'}}{s,h}\in\dom{\pathimg{h}{s}{\pi}}$.
Therefore, $\lval{\trunc{\pi'}}{s,h} = (s(x),\f)$ again. As such,
$h(s(x),\f)=s(x)$, contradicting the acyclicity of $\pi$.

Inductive case: let $\pi=\rho.\f$. Since $\pi'\not\prefix\pi$ it will also be the case that
$\pi'\not\prefix\rho$.  It is not hard to see that $\rho$ is disconnected and acyclic in $s,h$,
because $\pi$ is, therefore the inductive hypothesis applies, completing the proof.
\end{proof}

\begin{lemma}
\label{lem:preserve-image-same-heap-update}
Suppose $\pi$ is acyclic in $s,h$, that $s(\rootvar{\pi})=s'(\rootvar{\pi})$,
that $\dom{\pathimg{h}{s}{\pi}}\subseteq\dom{h'}$ and
that for all $\alpha\in\dom{\pathimg{h}{s}{\pi}}\setminus\{\lval{\pi}{s,h}\}$,
$h(\alpha)=h'(\alpha)$.
Then, for any $\pi'\prefix\pi$, $\lval{\pi'}{s,h}=\lval{\pi'}{s',h'}$, and
$\dom{\pathimg{h}{s}{\pi}}=\dom{\pathimg{h'}{s'}{\pi}}$.
\end{lemma}
\begin{proof}
If $|\pi|=1$ then the result follows trivially.

Let $|\pi|>1$ and fix some path $\pi'\properprefix\pi$.
First, observe that $\lval{\pi'}{s,h}\neq\lval{\pi}{s,h}$, because
otherwise $\locn{\lval{\pi'}{s,h}}=\locn{\lval{\pi}{s,h}}=h(\lval{\trunc{\pi}}{s,h})$, contradicting
acyclicity.
This means that for all $\pi'\properprefix\pi$, $h(\lval{\pi'}{s,h})=h'(\lval{\pi'}{s,h})$.
Lemma~\ref{lem:preserve-image-distinct-heap-update} can now be used on $\trunc{\pi}$,
giving us that $\pathimg{h}{s}{\trunc{\pi}}=\pathimg{h'}{s'}{\trunc{\pi}}$.
Applying Lemma~\ref{lem:preserve-image-stack-update} to the restricted heap
$\pathimg{h}{s}{\trunc{\pi}}$
lets us conclude that $\lval{\pi'}{s,h}=\lval{\pi'}{s',h'}$.
Finally, this gives $\lval{\pi}{s,h}=\lval{\pi}{s',h'}$ and
$\dom{\pathimg{h}{s}{\pi}}=\dom{\pathimg{h'}{s'}{\pi}}$, completing the proof.
\end{proof}

\begin{lemma}
\label{lem:preserve-image-distinct-heap-update}
Suppose that $s(\rootvar{\pi})=s'(\rootvar{\pi})$, that
$\dom{\pathimg{h}{s}{\pi}}\subseteq\dom{h'}$,
and that for all addresses $\alpha\in\dom{\pathimg{h}{s}{\pi}}$, $h(\alpha)=h'(\alpha)$.
Then $\pathimg{h}{s}{\pi}=\pathimg{h'}{s'}{\pi}$
and $\lval{\pi}{s,h}=\lval{\pi}{s',h'}$.
\end{lemma}
\begin{proof}
By induction on $\pi$. If $\pi= x.\f$ then
$\pathimg{h}{s}{\pi} = \{ (s(x),\f)\mapsto h(s(x),\f)\} = \{ (s'(x),\f)\mapsto h(s'(x),\f)\}$.
By assumption, $h(s(x),\f)=h'(s(x),\f)=h'(s'(x),\f)$, thus
$\pathimg{h}{s}{\pi} = \pathimg{h'}{s'}{\pi}$, and $\lval{\pi}{s,h}=\lval{\pi}{s',h'}$.

Let $\pi=\pi'.\f$. Then,
$\pathimg{h}{s}{\pi} = \pathimg{h}{s}{\pi'} \cup \{ \lval{\pi}{s,h} \mapsto h(\lval{\pi}{s,h})\}$,
and $\dom{\pathimg{h}{s}{\pi'}}\subseteq\dom{h'}$.
By the inductive hypothesis, $\pathimg{h}{s}{\pi'}=\pathimg{h'}{s'}{\pi'}$, and
$\lval{\pi'}{s,h}=\lval{\pi'}{s',h'}$. Also, $\lval{\pi'}{s,h}\in\dom{\pathimg{h}{s}{\pi'}}$, thus
$h(\lval{\pi'}{s,h})=h'(\lval{\pi'}{s,h})=h'(\lval{\pi'}{s',h'})$, meaning
$\pathimg{h}{s}{\pi}=\pathimg{h'}{s'}{\pi}$, and $\lval{\pi}{s,h}=\lval{\pi}{s',h'}$.
\end{proof}

\begin{alemma}{\ref{lem:stability-means-preservation-nocalls}}
Let $\pi$ be an access path and $\tr$ a WF, non-empty trace without any
$\cmdpush{}$/$\cmdpop{}$ commands,
whose starting state
$\angled{\cmdskip, s_0\cons S, h_0, \_}$ is well-behaved,
and where $\pi$ is acyclic and disconnected in $s_0,h_0$.
Let the last state of $\tr$ be $\angled{\_, s'\cons S, h', \_}$.

If it is the case that $\forall \e\in \alpha(\set{\tr})_W. \e\not\properprefix\pi$, then
$\pi$ is preserved from $s_0,h_0$ to $s',h'$.
\end{alemma}
\begin{proof}
By induction on the length of $\tr$. The base case is trivial (traces always start with $\cmdskip$).
For the inductive case, let $\angled{\_, s\cons S, h, \_}$ be the penultimate state in $\tr$.
By the inductive hypothesis, $\pi$ is preserved from $s_0,h_0$ to $s,h$.  Thus it suffices to show
that $\pi$ is preserved from $s,h$ to $s',h'$. The cases for $\cmdlock{}/\cmdunlock{}$ are trivial.

We call the assumption that $\forall \e\in \alpha(\set{\tr})_W. \e\not\properprefix\pi$, the \emph{stability assumption}.

\noindent\textbf{Case: $c=(x:=y)$ or $c=(x:=\pi')$ or $c=(x := \cmdalloc)$}.

Clearly, $\lval{\pi}{s',h'}$ is defined.

By the stability assumption and
Definition~\ref{def:alpha} ($\exec$), it follows that $x\neq \rootvar{\pi}$.
Therefore, by the concrete semantics of $c$, $s(\rootvar{\pi})=s'(\rootvar{\pi})$,
and in all cases, $h\subseteq h'$.
Thus, by Lemma~\ref{lem:preserve-image-stack-update}, it follows that
$\pathimg{h'}{s'}{\pi} = \pathimg{h}{s}{\pi}$  and
$\forall\alpha\in\dom{\pathimg{h}{s}{\pi}}.\,h'(\alpha)=h(\alpha)$.

The acyclicity of $\pi$ in $s',h'$ follows from Lemma~\ref{lem:preserve-acyclicity} and the fact
that $\pathimg{h'}{s'}{\pi} = \pathimg{h}{s}{\pi}$.

We show that $\pi$ is disconnected in $s',h'$.  Given
$\pathimg{h}{s}{\pi} = \pathimg{h'}{s'}{\pi}$, the second clause
of the disconnectedness definition carries over directly from
the disconnectedness of $\pi$ in $s,h$ whenever $c$ is
$x :=y $ or $x:=\pi'$ (because $h=h'$). The case where $c$ is $x:=\cmdalloc$
follows by the same reasoning, plus the fact that all newly-allocated
addresses point to the newly allocated location, and therefore by construction, cannot alias
addresses in $h$.

Given the disconnectedness of $\pi$ in $s,h$ and the fact that for all variables
$w\neq x$, $s(w)=s'(w)$, we only need to show
$s'(x)\notin\locn{\pathimg{h'}{s'}{\pi}}$ for the first clause.

If $c=(x:=y)$ then $s'(x)=s(y)$.
Given that $y\neq \rootvar{\pi}$, due to the stability assumption,
and the first clause of the definition of disconnectedness for $\pi$ in $s,h$, we have that
$s(y)\notin\locn{\pathimg{h}{s}{\pi}}=\locn{\pathimg{h'}{s'}{\pi}}$
(since $\pathimg{h}{s}{\pi} = \pathimg{h'}{s'}{\pi}$), as required.

If $c=(x:=\pi')$ then $s'(x)=h(\lval{\pi'}{s,h})$.
Given that $h=h'$ and that $\pathimg{h}{s}{\pi} = \pathimg{h'}{s'}{\pi}$, it
suffices to show $h(\lval{\pi'}{s,h})\notin\locn{\pathimg{h}{s}{\pi}}$.
By the stability assumption,  $\pi'\not\prefix\pi$, thus applying Lemma~\ref{lem:indomain-prefix},
we obtain $\lval{\pi'}{s,h}\notin\dom{\pathimg{h}{s}{\pi}}$.
By the contrapositive of the second clause of disconnectedness of $\pi$ in $s,h$ follows that
$h(\lval{\pi'}{s,h})\notin\dom{\pathimg{h}{s}{\pi}}$, as required.

If $c=(x:=\cmdalloc)$ then $s'(x)=\ell$ for some $\ell\notin\dom{h}$.
Clearly, $\ell\notin\locn{\pathimg{h}{s}{\pi}}$ and since
$\pathimg{h}{s}{\pi} = \pathimg{h'}{s'}{\pi}$, it follows that
$s'(x)=\ell\notin\locn{\pathimg{h'}{s'}{\pi}}$, as required.

\noindent\textbf{Case: $c=(\pi':=x)$, where $\pi'\neq\pi$}.

By the preservation of well-behavedness $c$ is safe to execute, thus, $\lval{\pi'}{s,h}\in\dom{h}$.
Also, $s=s'$.

The stability assumption forces $\pi'\not\prefix\pi$, thus by Lemma~\ref{lem:indomain-prefix} we have that
$\lval{\pi'}{s,h}\notin\dom{\pathimg{h}{s}{\pi}}$.
This, in addition with the semantics of the assignment to $\pi'$, means that
for all addresses $\alpha\in\dom{\pathimg{h}{s}{\pi}}$, $h(\alpha)=h'(\alpha)$.
Thus Lemma~\ref{lem:preserve-image-distinct-heap-update} applies, giving
$\lval{\pi}{s,h}=\lval{\pi}{s',h'}$,
$\pathimg{h}{s}{\pi}=\pathimg{h'}{s'}{\pi}$ and
$\forall\alpha\in\dom{\pathimg{h}{s}{\pi}}.\,h(\alpha)=h'(\alpha)$,
thus,
$\dom{\pathimg{h}{s}{\pi}} = \dom{\pathimg{h'}{s'}{\pi}}$.

For disconnectedness, first, $\lval{\pi}{s',h'}$ is clearly defined;
the first clause holds due to the fact that $s=s'$, $\pathimg{h}{s}{\pi}=\pathimg{h'}{s'}{\pi}$,
and the assumption of disconnectedness of $\pi$ in $s,h$.

By the semantics of the command, $\locn{h}=\locn{h'}$.
Therefore, the second clause of disconnectedness in $s',h'$ holds for all
addresses $\alpha\neq\lval{\pi'}{s,h}$,
by the second clause of the disconnectedness in $s,h$.
By the first clause and stability ($x\neq\rootvar{\pi}$) follows that
$h'(\lval{\pi'}{s,h})=s(x)\notin\locn{\pathimg{h}{s}{\pi}}=\locn{\pathimg{h'}{s'}{\pi}}$,
covering the case where $\alpha=\lval{\pi'}{s,h}$.

As before, acyclicity follows from the fact that $\pathimg{h}{s}{\pi}=\pathimg{h'}{s'}{\pi}$,
the acyclicity in $s,h$ and Lemma~\ref{lem:preserve-acyclicity}.

\noindent\textbf{Case: $c=(\pi:=x)$}.

By the acyclicity of $\pi$ in $s,h$ and the semantics of the store it follows that
$\forall\alpha\in\dom{\pathimg{h}{s}{\pi}}\setminus\{\lval{\pi}{s,h}\}.\,h(\alpha)=h'(\alpha)$.
Thus, Lemma~\ref{lem:preserve-image-same-heap-update} applies and we have
$\lval{\pi}{s,h}=\lval{\pi}{s',h'}$ and
$\dom{\pathimg{h}{s}{\pi}} = \dom{\pathimg{h'}{s'}{\pi}}$.

Disconnectedness: The first clause follows from the disconnectedness in $s,h$, $s=s'$
and $\locn{\pathimg{h}{s}{\pi}} = \locn{\pathimg{h'}{s'}{\pi}}$
(because $\dom{\pathimg{h}{s}{\pi}} = \dom{\pathimg{h'}{s'}{\pi}}$).

For the second clause, we need to show that for all $\alpha\in\dom{h'}$,
if $h'(\alpha)\in\locn{\pathimg{h'}{s'}{\pi}}$, then $\alpha\in\dom{\pathimg{h'}{s'}{\pi}}$.
As $\dom{\pathimg{h}{s}{\pi}} = \dom{\pathimg{h'}{s'}{\pi}}$ and $\dom{h'}=\dom{h}$,
this is equivalent to showing that for all $\alpha\in\dom{h}$,
if $h'(\alpha)\in\locn{\pathimg{h}{s}{\pi}}$, then $\alpha\in\dom{\pathimg{h}{s}{\pi}}$.

If $\alpha\neq\lval{\pi}{s,h}=\lval{\pi}{s',h'}$, then $h(\alpha)=h'(\alpha)$. Thus,
if $h'(\alpha)\in\locn{\pathimg{h}{s}{\pi}}$
then $h(\alpha)\in\locn{\pathimg{h}{s}{\pi}}$
thus $\alpha\in\dom{\pathimg{h}{s}{\pi}}$ (by disconnectedness in $s,h$).

If $\alpha=\lval{\pi}{s,h}=\lval{\pi}{s',h'}$, then $h'(\alpha)=s(x)$.
By stability, $x\neq\rootvar{\pi}$, thus by the first clause of disconnectedness,
$s(x)\notin\locn{\pathimg{h}{s}{\pi}}$, as required.

Acyclicity: we wish to prove that for every $\pi''\prefix\pi'\prefix\pi$,
$h'(\lval{\pi'}{s',h'})\neq\locn{\lval{\pi''}{s',h'}}$.
Using Lemma~\ref{lem:preserve-image-same-heap-update}, this reduces to showing
$h'(\lval{\pi'}{s,h})\neq\locn{\lval{\pi''}{s,h}}$
(because for any path $\pi'\prefix\pi$, $\lval{\pi'}{s,h}=\lval{\pi'}{s',h'}$).
We distinguish two cases.

Let $\lval{\pi''}{s,h}\neq\lval{\pi}{s,h}$.
Then $h'(\lval{\pi''}{s,h}) = h(\lval{\pi''}{s,h})\neq\locn{\lval{\pi''}{s,h}}$,
by the acyclicity in of $\pi$ in $s,h$.

Let $\lval{\pi''}{s',h'}=\lval{\pi}{s',h'}$.
This means $h'(\lval{\pi''}{s',h'})= s(x) \neq \locn{\lval{\pi'}{s,h}}$.
by stability and the disconnectedness in $s,h$.
\end{proof}

\begin{lemma}[Frame Property]
  \label{lem:frame}
  Let $h=h_0 \cup h_F$ (where $\cup$ is disjoint union), and let $\tr\in\sem{C}\angled{s,h,L}$ such that $\llast{\tr}=\angled{s',h',L'}$.
  If $\sem{C}\angled{s,h_F,L}\neq\emptyset$ then there exists $\tr'\in\sem{C}\angled{s,h_F,L}$
  such that $\llast{\tr'}=\angled{s',h'_F,L'}$, and $h'=h_0\cup h'_F$.
\end{lemma}
\begin{proof}
  By the semantics of the language and standard separation-logic style reasoning.
\end{proof}

\begin{alemma}{\ref{lem:nocalls-preservation}}
  Suppose that $\mbody(\m)$ has no procedure calls, that $\pi$ is an access path which is
  disconnected and acyclic in the well-behaved state $s,h$,
  that $\tr\in\sem{\cmdskip;\m(\bar{\e})}\angled{s,h,L}$ is a non-empty trace,
  that $\llast{\tr}=\angled{\cmdpop{},s,h',\_}$, that $\alpha(\set{\tr})=\angled{W,\_,\_}$, and that
  $\forall\e\in W.\, \e\not\properprefix\pi$.

  Then, $\pi$ is preserved from $s,h$ to $s,h'$.
\end{alemma}
\begin{proof}
Let $\e=(\e_1,\ldots,\e_n)$ and let $\pi=z.\fs$.
Since $\tr$ is non-empty, for all $i\le n$, $\eval{\e_i}{s,h}=\ell_i$ (and is defined).
By the semantics for method calls and those for simple commands ($\cmdpop$) we have
$\tr = \tr_{\text{push}} \append \tr' \append \tr_{\text{pop}}$ where
\begin{itemize}
  \item $\tr_{\text{push}} = [\angled{\cmdpush{\e_1,\ldots,\e_n}, \hat s \cons s, h, L}]$,
  \item $\tr' \in \sem{\body{\m}}\angled{\hat s \cons s, h, L}$, where $\llast{\tr'}=\angled{\_,\hat s' \cons s,h',L'}$, and
  \item $\tr_{\text{pop}}=[\angled{\cmdpop{}, s, h', L'}]$.
\end{itemize}
where
$\hat s \defeq s_{\nil}[\argm_1 \mapsto \ell_1]\cdots[\argm_n\mapsto\ell_n]$.

By the assumption that $\forall\e\in W.\,\e\not\properprefix\pi$ (we call this the stability assumption)
and the definition of
how $\exec$ modifies the $W$ component for
the method call case, we can assume there is at most one $\e_i$ such that
$\e_i= z$, or $\e_i=\pi_i\properprefix\pi$.

\noindent \textbf{Case: there is no $i\le n$ such that $\e_i= z$, or $\e_i=\pi_i\properprefix\pi$}.
Fix the stacks $\tilde s \defeq \hat s[\argm_{n+1}\mapsto s(z)]$
and $\tilde s' \defeq \hat s'[\argm_{n+1}\mapsto s(z)]$.
Recall that, by convention, $\body{\m}$ does not use $\argm_{n+1}$;
thus, since $\tr' \in \sem{\body{\m}}\angled{\hat s \cons s, h, L}$ and where $\llast{\tr'}=\angled{\_,\hat s' \cons s,h',L'}$,
it will also be the case that there is another ``isomorphic'' trace $\tr''$ such that
$\tr'' \in \sem{\body{\m}}\angled{\tilde s \cons s, h, L}$, where $\llast{\tr''}=\angled{\_,\tilde s' \cons s,h',L'}$.

We now show that $\rho=\argm_{n+1}.\fs$ is disconnected and acyclic in
$\tilde s, h$. Note that by Lemma~\ref{lem:preserve-image-stack-update},
$\pathimg{h}{s}{\pi} = \pathimg{h}{\tilde s}{\rho}$ and $\lval{\pi}{s,h}=\lval{\rho}{\tilde s, h}$,
therefore also $\dom{\pathimg{h}{s}{\pi}} = \dom{\pathimg{h}{\tilde s}{\rho}}$.
\begin{enumerate}
\item First clause of disconnectedness: We need to show that for all variables
$x\neq \argm_{n+1}$, $\tilde s(x)\notin\locn{\pathimg{h}{\tilde s}{\rho}}$.
By the definition of $\tilde s$, this means showing that for all $i\le n$,
$\tilde s(\argm_i) = \hat s(\argm_i)\notin\locn{\pathimg{h}{\tilde s}{\rho}}$.
But $\hat s(\argm_i) = \eval{\e_i}{s,h}$.
We have two cases:
\begin{enumerate}
  \item If $\e_i= x$ for
some variable $x$, then $\eval{\e_i}{s,h} = s(x)$. By the stability assumption,
we have that $x\neq z$ and by the disconnectedness of $\pi$ in $s,h$
we have that $s(x)\notin\locn{\pathimg{h}{s}{\pi}}$. This completes the proof
through the fact that $\pathimg{h}{s}{\pi} = \pathimg{h}{\tilde s}{\rho}$.
  \item if $\e_i=\pi'$ then $\eval{\e_i}{s,h}=h(\lval{\pi'}{s,h})$.
By the stability assumption, we have that $\pi'\not\properprefix\pi$. If $\pi'\neq\pi$
then by Lemma~\ref{lem:indomain-prefix}, we have that
$\lval{\pi'}{s,h}\notin\dom{\pathimg{h}{s}{\pi}}$. This directly means
(through the fact that $\pathimg{h}{s}{\pi} = \pathimg{h}{\tilde s}{\rho}$) that
$\hat s(\argm_i)\notin\locn{\pathimg{h}{\tilde s}{\rho}}$, completing the proof.
If $\pi'=\pi$, then $\eval{\e_i}{s,h}=h(\lval{\pi}{s,h})$.
By acyclicity of $\pi$ in $s,h$ we have that
$\hat s(\argm_i)=h(\lval{\pi}{s,h})\notin\locn{\pathimg{h}{s}{\pi}}=\locn{\pathimg{h}{\tilde s}{\rho}}$.
\end{enumerate}

\item Second clause of disconnectedness: We need to show that for all $\alpha\in\dom{h}$, if
$h(\alpha)\in\locn{\pathimg{h}{\tilde s}{\rho}}$, then
$\alpha\in\dom{\pathimg{h}{\tilde s}{\rho}}$, or, that if
$h(\alpha)\in\locn{\pathimg{h}{s}{\pi}}$, then
$\alpha\in\dom{\pathimg{h}{s}{\pi}}$. This is immediate by the disconnectedness
of $\pi$ in $s,h$.
\end{enumerate}

Acyclicity of $\rho$ in $\tilde s, h$ follows from Lemma~\ref{lem:preserve-acyclicity},
noting that $\pathimg{h}{s}{\pi} = \pathimg{h}{\tilde s}{\rho}$ as shown above.

We now can apply Lemma~\ref{lem:stability-means-preservation-nocalls} to
obtain that $\rho$ is preserved from $\tilde s,h$ to $\tilde s',h'$ (note that since there are no
accesses to $\argm_{n+1}$ in $\body{\m}$, and no procedure calls, trivially
$\forall\e\in\alpha(\set{\tr''})_W.\, \e\not\properprefix\rho$).

Lemma~\ref{lem:preserve-image-stack-update} yields
$\lval{\rho}{\tilde s',h'} = \lval{\pi}{s,h'}$,
and $\pathimg{h'}{\tilde s'}{\rho} = \pathimg{h'}{s}{\pi}$,
therefore also $\dom{\pathimg{h'}{\tilde s'}{\rho}} = \dom{\pathimg{h'}{s}{\pi}}$.
Lemma~\ref{lem:preserve-acyclicity} yields the acyclicity of $\pi$ in $s,h'$.

Finally, we also need to finally show that $\pi$ is disconnected in $s, h'$.
This follows from the fact that $\dom{\pathimg{h'}{s}{\pi}}=\dom{\pathimg{h'}{\tilde s'}{\rho}}$
and the disconnectedness of $\rho$ in $\tilde s',h'$.

\noindent \textbf{Case: for exactly one $j\le n$, $\e_j\properprefix\pi$, and $\e_j= z$}.
This case proceeds almost identically, but for the following differences:
we don't use $\argm_{n+1}$, nor ``jump'' to another trace $\tr''$ as above, but operate directly on
$\tr'$; the definitions $\tilde s \eqdef \hat s$ and $\tilde s' \eqdef \hat s'$ are trivial;
and $\rho \eqdef \argm_j.\fs$.

Application of Lemma~\ref{lem:stability-means-preservation-nocalls}: Note that $\alpha(\set{\cdot})$
is monotonic: $\tr_{\text{push}}\append\tr'\prefix\tr$, thus $\alpha(\set{\tr_{\text{push}}\append\tr'})_W \subseteq \alpha(\set{\tr})_W$.
Thus, from the top-level assumption about the $W$ component of the abstraction of $\tr$ follows that, also,
$\forall\e\in  \alpha(\set{\tr_{\text{push}}\append\tr'})_W.\,\e\not\properprefix\pi$.
By the definition of $\exec$ for $\cmdpush{}$, this means
$\forall\e\in \subst{\alpha(\set{\tr'})_W}{[\argm_i\mapsto\e_i]_{i=1}^n}.\,\e\not\properprefix\pi$,
thus
$\forall\e\in \alpha(\set{\tr'})_W.\,\e\not\properprefix\rho$, by the properties of substitutions, the
definition of $\exec$ for load/stores, the fact that there are no procedure calls in $\m$, and
the fact that the root of $\rho$ is $\argm_j$ and thus not a local.

\newcommand{\hs}{\bar{\mathsf{h}}}
\noindent \textbf{Case: for exactly one $j\le n$, $\e_j\properprefix\pi$, and $\e_j=z.\gs$ where $\gs\properprefix\fs$ and $\gs\neq\epsilon$}.
Again, we don't use $\argm_{n+1}$; the definitions of $\tilde s \eqdef \hat s$
and $\tilde s' \eqdef \hat s'$ are again trivial.
Let $\hs$ be such that $\pi=z.\gs.\hs$; it is easy to see that $\hs\neq\epsilon$.
We set $\rho\defeq\argm_j.\hs$.

We define a heap $\tilde h$ as the subheap of $h$, reachable from any $\e_i$: that is, $\tilde h$ is
the smallest subheap of $h$ whose domain includes any address of the form $(\eval{\e_i}{s,h},\f)$
and which is closed under the image of $h$.  This is well-defined to the well-behavedness of $s,h$.
We fix as $h_0 \eqdef h\setminus\tilde h$, i.e., the remainder heap after removing $\tilde h$.
It should be easy to see that $\lval{\rho}{\tilde s, \tilde h}$ is defined.
By Lemma~\ref{lem:frame}, there exists a trace $\tr''\in\sem{\body{\m}}\angled{\tilde s\cons s, \tilde h, l}$
such that $\llast{\tr''}=\angled{\tilde s'\cons s, \tilde h', L'}$ and $h'=h_0\cup\tilde h'$.

We now show that $\rho=\argm_j.\hs$ is disconnected and acyclic in $\tilde s, \tilde h$.

It is easily seen by the definition of $\tilde h$ that
$\pathimg{h}{s}{\pi} \supset \pathimg{h}{\tilde s}{\rho}=\pathimg{\tilde h}{\tilde s}{\rho}$
and $\lval{\pi}{s,h}=\lval{\rho}{\tilde s, h}=\lval{\rho}{\tilde s, \tilde h}$,
therefore also $\dom{\pathimg{h}{s}{\pi}} \supset \dom{\pathimg{h}{\tilde s}{\rho}}= \dom{\pathimg{\tilde h}{\tilde s}{\rho}}$.

\begin{enumerate}
\item First clause of disconnectedness: We need to show that for all variables
$x\neq \argm_j$, $\tilde s(x)\notin\locn{\pathimg{\tilde h}{\tilde s}{\rho}}$,
or, that for all $i\le n$ and $i\neq j$,
$\tilde s(\argm_i) = \eval{\e_i}{s,h} \notin\locn{\pathimg{\tilde h}{\hat s}{\rho}}=\locn{\pathimg{h}{\hat s}{\rho}}\subset\locn{\pathimg{h}{s}{\pi}}$.

By the assumption of this case, for all $i\neq j$, $i\le n$ is it the case that $\e_i\not\properprefix\pi$.
If $\e_i$ is a variable, then by the disconnectedness of $\pi$ we have $\eval{\e_i}{s,h}\notin\locn{\pathimg{h}{s}{\pi}}$.
If $\e_i$ is a path, by Lemma~\ref{lem:indomain-prefix}, it follows that $\lval{\e_i}{s,h}\notin\dom{\pathimg{h}{s}{\pi}}$.
Thus, by the contrapositive of the 2nd clause of disconnectedness of $\pi$, $\eval{\e_i}{s,h}=h(\lval{\e_i}{s,h})\notin\locn{\pathimg{h}{s}{\pi}}$.
This suffices because $\dom{\pathimg{h}{s}{\pi}} \supset \dom{\pathimg{\tilde h}{\tilde s}{\rho}}$.

\item Second clause of disconnectedness: We need to show that for all $\alpha\in\dom{\tilde h}$, if
$\tilde h(\alpha)\in\locn{\pathimg{\tilde h}{\tilde s}{\rho}}$, then
$\alpha\in\dom{\pathimg{\tilde h}{\tilde s}{\rho}}$. This is true by the construction of $\tilde h$
and the disconnectedness of $\pi$ in $s,h$.
\end{enumerate}

Acyclicity of $\rho$ in $\tilde s, \tilde h$ follows the acyclicity of $\pi$ in $s,h$ and the fact
that $\tilde h \subseteq h$.

Application of Lemma~\ref{lem:stability-means-preservation-nocalls}: Note that $\alpha(\set{\cdot})$
is monotonic: $\tr_{\text{push}}\append\tr'\prefix\tr$, thus $\alpha(\set{\tr_{\text{push}}\append\tr'})_W \subseteq \alpha(\set{\tr})_W$.
Thus, from the top-level assumption about the $W$ component of the abstraction of $\tr$ follows that, also,
$\forall\e\in  \alpha(\set{\tr_{\text{push}}\append\tr'})_W.\,\e\not\properprefix\pi$.
By the definition of $\exec$ for $\cmdpush{}$, this means
$\forall\e\in \subst{\alpha(\set{\tr'})_W}{[\argm_i\mapsto\e_i]_{i=1}^n}.\,\e\not\properprefix\pi$,
but
$\alpha(\set{\tr'})_W=\alpha(\set{\tr''})_W$ (because $\ab{\tr'}=\ab{\tr''}$),
thus
$\forall\e\in \subst{\alpha(\set{\tr''})_W}{[\argm_i\mapsto\e_i]_{i=1}^n}.\,\e\not\properprefix\pi$.
This means
$\forall\e\in \alpha(\set{\tr''})_W.\,\e\not\properprefix\rho$, by the properties of substitutions, the
definition of $\exec$ for load/stores, the fact that there are no procedure calls in $\m$, and
the fact that the root of $\rho$ is $\argm_j$ and thus not a local.

We now can apply Lemma~\ref{lem:stability-means-preservation-nocalls} to
obtain that $\rho$ is preserved from $\tilde s,\tilde h$ to $\tilde s',\tilde h'$.
Also, due to the stability assumption, it is clear that $\tilde{s}(\argm_j)=\tilde{s'}(\argm_j)$,
thus $\lval{\rho}{\tilde s',\tilde h'} = \lval{\rho}{\tilde s, \tilde h'}$.
Because $h'=h_0\cup \tilde h'$, it is easy to check that
$\lval{\rho}{\tilde s,\tilde h'} = \lval{\rho}{\tilde s, h'} = \lval{\pi}{s,h'}$.
By the same token, it can be shown that $\dom{\pathimg{h'}{s}{\pi}}=\dom{\pathimg{h}{s}{\pi}}$
and that $\pi$ is acyclic and disconnected in $s,h'$.
\end{proof}

\begin{alemma}{\ref{lem:call-stack}}
  Let $\tr$ be a (prefix) trace produced by Lemma~\ref{lem:path-access-existence}.
  If additionally it is the case that $\forall\e\in\alpha(\set{\tr})_W.\,\e\not\properprefix\pi$
  then $\lval{\pi}{s,h}=\lval{\pi'}{s',h'}$.
\end{alemma}
\begin{proof}
  We prove by backwards induction on the prefixes $\tr'\prefix\tr$ we show the following:
  Let $\tr'=\tr_1\append\tr_{\text{push}}\append\tr_0$ such that $\tr_{\text{push}}$ is the last
  $\cmdpush{}$ command in $\tr'$ that is unmatched by a $\cmdpop{}$.  Let
  $\tr_{\text{push}} = [\angled{\cmdpush{\e_1,\ldots,\e_n}, s_0\cons s_1\cons S, h_0}]$ and
  $\llast{\tr'}=\angled{\_,s'\cons s_1\cons S, h'}$.
  Let $\pi'$ be some path such that $\pi'[\bar{\theta}]\prefix\pi$, where $\bar{\theta}=\theta_0\cons\_$ is the
  sequence of substitutions induced by the active calls at the end of $\tr'$.
  Then, $\lval{\pi'}{s',h'}=\lval{\e_j.\fs}{s_0,h_0}$, where $\fs$ are such that $\pi'[\theta_0]=\e_j$.

  By the definition of $\exec$, the assumption $\forall\e\in\alpha(\set{\tr})_W.\,\e\not\properprefix\pi$
  and the fact that $\pi'[\bar{\theta}]\prefix\pi$, we can conclude that
  $\forall\e\in\alpha(\set{\tr_0})_W.\,\e\not\properprefix\pi'$.
  Thus, the trace $\tr_0$ satisfies the conditions of Lemma~\ref{lem:stability-on-traces}, thus
  $\pi'$ is preserved from $s_0,h_0$ to $s_0,h'$.  It is easy to see that
  $s_0(\rootvar{\pi'})=s'(\rootvar{\pi'})$, thus by Lemma~\ref{lem:preserve-image-distinct-heap-update}
  $\pi'$ is preserved from $s_0,h_0$ to $s',h'$. By the observation that $s_0(\argm_j)=\eval{\e_j}{s,h_0}$,
  and setting $\fs$ to be such that $\pi'[\theta_0]=\e_j$, it can be seen that
  $\lval{\pi'}{s',h'}=\lval{\e_j.\fs}{s_0,h_0}$.

  At the last step, $\pi'=\pi$ and we are done.
\end{proof}

\begin{lemma}
  \label{lem:starting-state} Given paths $\pi_1=x.\fs$ and $\pi_2=y.\fs$, there exists a well-behaved
  state $(s_1,s_2,h)$ such that $\pathimg{h}{s_1}{\pi_1}=\pathimg{h}{s_2}{\pi_2}$,
  $\lval{\pi_1}{s_1,h}=\lval{\pi_2}{s_2,h}$ and $\pi_i$ is acyclic and
  disconnected in $s_i,h$, for $i\in\{1,2\}$.
\end{lemma}
\begin{proof}
  Let $\fs = \f_1.\cdots.\f_n$ for some $n\ge 1$. Let $\ell_0,\ell_1,\ldots,\ell_n$ be distinct locations,
  and set $s_1(x)=s_2(y)=\ell_1$.  Also, set $s_1(x')=s_2(y')=\ell_0$ for all $x'\in\Var\setminus\{x\}$
  and $y'\in\Var\setminus\{y\}$. This defines $s_1,s_2$.

  Next, set $h(\ell_0,\f)=\ell_0$ for all $\f\in\FieldName$.
  For each $1\le j < n$, set $h(\ell_j, f_j)  = \ell_{j+1}$ and
  $h(\ell_j, \f) = \ell_0$ for all $\f\in\FieldName\setminus\{\f_j\}$. Set $h(\ell_n,\f) = \ell_0$
  for all $\f\in\FieldName$. This determines $h$.

  It is easy to check that all required properties are true of $(s_1,s_2,h)$.
\end{proof}

\begin{lemma}
  \label{lem:heap-monotonicity}
  For any trace $\tr\in\sem{C} \astate$ and any two positions $i\le j \le |\tr|$, such that
  $\tr_i=\angled{\_,\_,h,\_}$ and $\tr_j=\angled{\_,\_,h',\_}$, it is the case that
  $\dom{h}\subseteq\dom{h'}$.
\end{lemma}
\begin{proof}
  By the concrete semantics and the fact that there is no deallocation.
\end{proof}

\begin{lemma}
  \label{lem:other-thread-preserved}
  Let $\pi_1=x.\fs$ and $\pi_2=y.\fs$ (ie., paths with the same field sequence).
  Let $s_1,s_2$ be two stacks such that $s_1(x)=s_2(y)$.
  Let $h,h'$ be two heaps such that $\dom{h}\subseteq\dom{h'}$ and, in addition,
  $\forall\alpha\in\dom{\pathimg{h}{s_1}{\pi_1}}\setminus\{\lval{\pi_1}{s_1,h}\}.\, h(\alpha)=h'(\alpha)$.
  Assume $\pi_1$ is acyclic in $s_1,h$.
  Then, $\lval{\pi_2}{s_2,h} = \lval{\pi_2}{s_2,h'}$.
\end{lemma}
\begin{proof}
  We apply Lemma~\ref{lem:preserve-image-stack-update} letting $s=s_1$, $s'=s_2$ and $h=h'$,
  obtaining $\pathimg{h}{s_1}{\pi_1}=\pathimg{h}{s_2}{\pi_2}$.

  As $\pi_1$ is acyclic in $s_1,h$ and $s_1(x)=s_2(y)$, we can easily
  conclude that $\pi_2$ is acyclic in $s_2,h$.

  From
  $\pathimg{h}{s_1}{\pi_1}=\pathimg{h}{s_2}{\pi_2}$
  we can conclude
  $\dom{\pathimg{h}{s_2}{\pi_2}} \subseteq \dom{h'}$,
  since $\pathimg{h}{s_1}{\pi_1}\subseteq\dom{h}\subseteq\dom{h'}$.

  From the facts
  $\forall\alpha\in\dom{\pathimg{h}{s_1}{\pi_1}}\setminus\{\lval{\pi_1}{s_1,h}\}.\, h(\alpha)=h'(\alpha)$
  and
  $\pathimg{h}{s_1}{\pi_1}=\pathimg{h}{s_2}{\pi_2}$
  and
  $\lval{\pi_1}{s_1,h}=\lval{\pi_2}{s_2,h}$
  we obtain by simple substitution
  $\forall\alpha\in\dom{\pathimg{h}{s_2}{\pi_2}}\setminus\{\lval{\pi_2}{s_2,h}\}.\, h(\alpha)=h'(\alpha)$.

  We apply Lemma~\ref{lem:preserve-image-same-heap-update} letting $\pi=\pi_2$, $s_2=s=s'$
  and $h,h'$ as here to obtain $\lval{\pi_2}{s_2,h} = \lval{\pi_2}{s_2,h'}$.
\end{proof}

\begin{lemma}
  \label{lem:parallel-paths}
  Let $\pi_1=x.\fs$, $\pi_2=y.\fs$ and $s_1(x)=s_2(y)$ such that $\lval{\pi_1}{s_1,h}$ is defined.
  Then for any $\gs\prefix\fs$, $\lval{x.\gs}{s_1,h}=\lval{y.\gs}{s_2,h}$, and
  $\pathimg{h}{s_1}{x.\gs}=\pathimg{h}{s_2}{y.\gs}$.
\end{lemma}
\begin{proof}
  By induction on the length of $\gs$, the fact that $s_1(x)=s_2(y)$ and that $h$ is a function.
\end{proof}

\begin{lemma}
  \label{lem:disconnected-across-threads}
  Let $\pi_1=x.\fs$ and $\pi_2=y.\fs$ such that $s_1(x)=s_2(y)$. Let also $h,h'$ be such that
  $\dom{\pathimg{h}{s_2}{\pi_2}}=\dom{\pathimg{h'}{s_2}{\pi_2}}$.
  If $\pi_1$ is disconnected in $s_1,h'$ then so is $\pi_2$ in $s_2,h'$.
\end{lemma}
\begin{proof}
  We need to show that (a) for any $v\neq y$, $s_2(v)\notin\locn{\pathimg{h'}{s_2}{\pi_2}}$, and
  (b) that for all $\alpha\in\dom{h'}$, if $h'(\alpha)\in\locn{\pathimg{h'}{s_2}{\pi_2}}$ then
  $\alpha\in\dom{\pathimg{h'}{s_2}{\pi_2}}$.

  For (a), note that $\pi_2$ is disconnected in $s_2,h$ so for any
  $v\neq y$, $s_2(v)\notin\locn{\pathimg{h}{s_2}{\pi_2}}$.
  But $\locn{\pathimg{h}{s_2}{\pi_2}}=\locn{\pathimg{h'}{s_2}{\pi_2}}$,
  because $\dom{\pathimg{h}{s_2}{\pi_2}}=\dom{\pathimg{h'}{s_2}{\pi_2}}$, completing the proof.

  For (b) observe that by Lemma~\ref{lem:parallel-paths},
  $\pathimg{h'}{s_1}{\pi_1}=\pathimg{h'}{s_2}{\pi_2}$, thus we only need to show that
  for all $\alpha\in\dom{h'}$, if $h'(\alpha)\in\locn{\pathimg{h'}{s_1}{\pi_1}}$ then
  $\alpha\in\dom{\pathimg{h'}{s_1}{\pi_1}}$.  This is immediate by the disconnectedness of $\pi_1$
  in $s_1,h'$.
\end{proof}

\begin{lemma}
  \label{lem:acyclic-across-threads}
  Let $\pi_1=x.\fs$ and $\pi_2=y.\fs$ such that $s_1(x)=s_2(y)$.
  If $\pi_1$ is acyclic in $s_1,h$ then so is $\pi_2$ in $s_2,h$.
\end{lemma}
\begin{proof}
  Since $\pi_1$ is acyclic in $s_1,h$ we know that for any $\fs''\prefix\fs\prefix\fs$,
  $h(\lval{x.\fs'}{s_1,h})\neq\locn{\lval{x.\fs''}{s_1,h}}$.
  By Lemma~\ref{lem:parallel-paths}, we know that
  $\lval{x.\fs'}{s_1,h}=\lval{y.\fs'}{s_2,h}$ and $\lval{x.\fs''}{s_1,h}=\lval{y.\fs''}{s_2,h}$,
  thus $h(\lval{y.\fs'}{s_2,h})\neq\locn{\lval{y.\fs''}{s_2,h}}$, as required.
\end{proof}

\begin{atheorem}{\ref{thm:no-fp}}
  Let $C_1,C_2$ be two programs such that $\asem{C_i} \bota= \angled{W_i,\_,A_i}$.
   Let $\pi_i$ be two paths and $c_i$ two commands such that $\angled{c_i,L_i}\in A_i$, and,
  \begin{itemize}
    \item $\pi_1=v_1.\fs$ and $\pi_2=v_2.\fs$ (ie., the field sequences are the same);
    \item $\forall \e_i\in W_i.\, \e_i\not\properprefix \pi_i$ (for $i\in\set{1,2}$);
    \item $c_1=(\pi_1:=x)$ and $c_2=(\pi_2:=y)$, or, $c_1=(\pi_1:=x)$ and $c_2=(y:=\pi_2)$;
    \item $L_1+L_2 \le 1$.
  \end{itemize}
  Then there exists a well-behaved state $\astate=\angled{(s_1,s_2),h,(0,0)}$  and a
  concurrent trace $\tr^{\parallel} \in \sem{C_1\parallel C_2}~\astate$ that races.
\end{atheorem}
\begin{proof}
\Wlog~we take $L_1=0$, as at least one of $L_1,L_2$ must be zero.

Given the two paths share the same field accesses ($\fs$), Lemma~\ref{lem:starting-state} provides a
a well-behaved state $\astate=\angled{(s_1,s_2),h,(0,0)}$ such that
$\lval{\pi_1}{s_1,h}=\lval{\pi_2}{s_2,h}$, and where $\pi_i$ is disconnected and acyclic in
$s^i,h$.

By assumption, $\asem{C_1} \bota= \angled{W_1,\_,A_1}$, thus by Theorem~\ref{thm:completeness},
$\angled{W_1,\_,A_1} = \alpha(\sem{C_1} \angled{s_1, h, 0})$.
Thus Lemma~\ref{lem:path-access-existence} implies there is a trace in $\sem{C_1} \angled{s_1, h, 0}$
and a shortest prefix $\tr_1$ of that trace, whose last state is $\angled{c'_1,\_,\_,L_1}$, and:
\begin{itemize}
  \item $c'_1$, $c_1$ are both stores or loads;
  \item $\exec~\ab{\tr_1}~(\bota,\epsilon)=(\angled{\_,\_,A'_1},\bar{\theta_1})$ where $\pi_1\in A'_1$;
  \item a path $\pi'_1$ such that $c' = (x:=\pi'_1)$ or $c'=(\pi'_1:=x)$,
  and $\set{\pi_1}=\subst{\set{\pi'_1}}{\bar{\theta_1}}$.
\end{itemize}

Define the following functions. Each maps a single-thread state to a two-thread trace,
with the given parameter as the stack of the "other" thread, and zero as its lock context.
\[
\begin{array}{r@{\;}l}
\mathsf{left}(\angled{c, S, h, L}, S')  & = \angled{c\parallel\epsilon, (S, S'), (L,0)} \\
\mathsf{right}(\angled{c, S, h, L}, S')  & = \angled{\epsilon\parallel c, (S', S), (0,L)}
\end{array}
\]

We construct a concurrent trace $\tr^{\parallel}_1$ and its last state as follows:
\[\begin{array}{r@{\;}l}
\tr^{\parallel}_1 & \eqdef \left[ \mathsf{left}(\astate_i, s_2) \mid \astate_i \text{ is the $i$-th state in $\tr_1$}, 1\le i< |\tr_1| \right] \\[3pt]
\llast{\tr^{\parallel}_1} & = \angled{\_,(s'_1 \cons S'_1,s_2),h',(L'_1,0))}
\end{array}\]
that is, $C_1$ progresses along $\tr_1$, but \emph{without executing the instruction $c'_1$},
while $C_2$ never takes a step.
Observe that by construction, $L'_1=0$, as $c'_1$ is not an $\cmdunlock$ instruction (thus $L'_1=L_1$),
and by assumption $L_1=0$.

We show that the state $s_2,h'$ is well-behaved. We know that $s_1,h'$ is well-behaved, by
Lemma~\ref{lem:well-behavedness}. Thus, for any address
$\alpha\in\dom{h'}$, $h'(\alpha)\in\locn{h'}$, taking care of the second clause of well-behavedness.
For the first clause, observe that $s_2,h$ is well-behaved thus for any variable $x$,
$s_2(x)\in\locn{h}$.  But as $\dom{h}\subseteq\dom{h'}$ by Lemma~\ref{lem:heap-monotonicity},
it follows that $\locn{h}\subseteq \locn{h'}$, thus $s_2(x)\in\locn{h'}$ too.

Thus, we have $\asem{C_2} \bota= \angled{W_2,\_,A_2} = \alpha(\sem{C_2} \angled{s_2, h', 0})$.
By identical reasoning as with $C_1$ we obtain a trace in $\sem{C_2} \angled{s_2, h', 0}$
and a shortest prefix trace $\tr_2$ of that,
whose last state is $\angled{c'_2,\_,\_,L_2}$ where
\begin{itemize}
  \item $c'_2$, $c_2$ are both stores or loads;
  \item $\exec~\ab{\tr_2}~(\bota,\epsilon)=(\angled{\_,\_,A'_2},\bar{\theta_2})$ where $\pi_2\in A'_2$;
  \item a path $\pi'_2$ such that $c' = (x:=\pi'_2)$ or $c'=(\pi'_2:=x)$,
  and $\set{\pi_2}=\subst{\set{\pi'_2}}{\bar{\theta_2}}$.
\end{itemize}

We define $\tr^{\parallel}_2$ and its last state as follows:
\[\begin{array}{r@{\;}l}
\tr^{\parallel}_2 & \eqdef \left[ \mathsf{right}(\astate_i,s'_1\cons S'_1) \mid \text{where $\astate_i$ is the $i$-th state in $\tr_2$}, 1\le i < |\tr_2| \right] \\
\llast{\tr^{\parallel}_2} & = \angled{\_,(s'_1 \cons S'_1,s'_2\cons S'_2),h'',\_)}
\end{array}\]

Finally, we define
\[
\tr^{\parallel}\eqdef\tr^{\parallel}_1\append\tr^{\parallel}_2
\]
We argue that
$\tr^{\parallel}$ is in $\sem{C_1\parallel C_2}~\angled{(s_1,s_2),h,(0,0)}$.
This is straightforward by our concurrent trace semantics (Figure~\ref{fig:ctraces})
and the observation that as $\tr^{\parallel}_1$ ends at a state with lock context $(0,0)$,
the side-condition for lock commands in $\tr^{\parallel}_2$ is satisfied by construction.
\footnote{Figure~\ref{fig:ctraces} defines a prefix-closed semantics so we need not concern ourselves
with what happens after the race.}

We now show that $\lval{\pi'_1}{s'_1,h''}=\lval{\pi'_2}{s'_2,h''}$, establishing that the next
instruction following $\tr^{\parallel}$ will necessarily access the same address.
  By construction we have
  \begin{equation}
    \begin{array}{r@{\,}l}
      \lval{\pi_1}{s_1,h} & =\lval{\pi_2}{s_2,h}\\
      \pathimg{h}{s_1}{\pi_1} & = \pathimg{h}{s_2}{\pi_2}\\
      \multicolumn{2}{c}{\text{$\pi_i$ is acyclic and disconnected in $s_i,h$, for $i\in\set{1,2}$}}
    \end{array}
    \label{eq:pi1-pi2}
  \end{equation}

  By assumption, $\forall \e_i\in W_i.\, \e_i\not\properprefix \pi_i$.
  We apply Lemma~\ref{lem:stability-on-traces} on $\tr_1$ to obtain
  that
  \begin{equation}
    \pi_1 \text{ is preserved from } s_1,h \text{ to } s_1,h' \label{eq:pi-h-hprime}
  \end{equation}
  Lemma~\ref{lem:call-stack} gives us
  \begin{equation}
    \lval{\pi_1}{s_1,h'}=\lval{\pi'_1}{s'_1,h'} \label{eq:pi1-endpoint}
  \end{equation}

  From (\ref{eq:pi-h-hprime}) we have
  $\forall\alpha\in\dom{\pathimg{h}{s_1}{\pi_1}}\setminus\{\lval{\pi_1}{s_1,h}\}.\, h(\alpha)=h'(\alpha)$.
  By (\ref{eq:pi1-pi2}) we have $\pathimg{h}{s_1}{\pi_1}=\pathimg{h}{s_2}{\pi_2}$ and $\lval{\pi_1}{s_1,h}=\lval{\pi_2}{s_2,h}$,
  thus $\forall\alpha\in\dom{\pathimg{h}{s_2}{\pi_2}}\setminus\{\lval{\pi_2}{s_2,h}\}.\, h(\alpha)=h'(\alpha)$.
  In addition, by (\ref{eq:pi1-pi2}), $\pi_2$ is acyclic in $s_2,h$.
  Clearly, $\dom{h}\subseteq\dom{h'}$ by Lemma~\ref{lem:heap-monotonicity}.
  Thus Lemma~\ref{lem:other-thread-preserved} gives
  \begin{equation}
    \lval{\pi_2}{s_2,h}=\lval{\pi_2}{s_2,h'} \label{eq:pi2-equal}
  \end{equation}

  The state $s_2,h'$ is well-behaved as shown above.
  Since $\pi_1$ is acyclic in $s_1,h'$ (\ref{eq:pi-h-hprime}),
  by Lemma~\ref{lem:acyclic-across-threads}, $\pi_2$ is also acyclic in $s_2,h'$.
  Observe that $\dom{\pathimg{h}{s_2}{\pi_2}}=\dom{\pathimg{h}{s_1}{\pi_1}}$ (\ref{eq:pi1-pi2}),
  $\dom{\pathimg{h}{s_1}{\pi_1}}=\dom{\pathimg{h'}{s_1}{\pi_1}}$ (\ref{eq:pi-h-hprime}), and
  by Lemma~\ref{lem:parallel-paths} and the fact that $s_2(\rootvar{\pi_2})=s_1(\rootvar{\pi_1})$,
  \begin{equation}
    \pathimg{h'}{s_1}{\pi_1}=\pathimg{h'}{s_2}{\pi_2} \label{eq:pi1-pi2-hprime}
  \end{equation}
  thus $\dom{\pathimg{h}{s_2}{\pi_2}}=\dom{\pathimg{h'}{s_2}{\pi_2}}$.
  Lemma~\ref{lem:disconnected-across-threads} applies, therefore $\pi_2$ is disconnected in $s_2,h'$.
  Thus,
  \begin{equation}
    \pi_2 \text{ is preserved from } s_2,h \text{ to } s_2,h' \label{eq:pi2-h-hprime}
  \end{equation}

  This means we can apply Lemma~\ref{lem:stability-on-traces} to $\tr_2$ and obtain that
  \begin{equation}
    \pi_2 \text{ is preserved from } s_2,h' \text{ to } s_2,h'' \label{eq:pi2-hprime-hprimeprime}
  \end{equation}
  and Lemma~\ref{lem:call-stack} yields
  \begin{equation}
    \lval{\pi_2}{s_2,h'}=\lval{\pi'_2}{s'_2,h''} \label{eq:pi2-endpoint}
  \end{equation}


Chaining (\ref{eq:pi1-endpoint}), (\ref{eq:pi-h-hprime}), (\ref{eq:pi1-pi2}), (\ref{eq:pi2-equal})
and (\ref{eq:pi2-endpoint}),
we conclude $\lval{\pi'_1}{s'_1,h''}=\lval{\pi'_2}{s'_2,h''}$, thus establishing the
existence of a race as per Definition~\ref{def:race}.
\end{proof}

%% file: paper.bbl

\begin{thebibliography}{32}


\ifx \showCODEN    \undefined \def \showCODEN     #1{\unskip}     \fi
\ifx \showDOI      \undefined \def \showDOI       #1{#1}\fi
\ifx \showISBNx    \undefined \def \showISBNx     #1{\unskip}     \fi
\ifx \showISBNxiii \undefined \def \showISBNxiii  #1{\unskip}     \fi
\ifx \showISSN     \undefined \def \showISSN      #1{\unskip}     \fi
\ifx \showLCCN     \undefined \def \showLCCN      #1{\unskip}     \fi
\ifx \shownote     \undefined \def \shownote      #1{#1}          \fi
\ifx \showarticletitle \undefined \def \showarticletitle #1{#1}   \fi
\ifx \showURL      \undefined \def \showURL       {\relax}        \fi
\providecommand\bibfield[2]{#2}
\providecommand\bibinfo[2]{#2}
\providecommand\natexlab[1]{#1}
\providecommand\showeprint[2][]{arXiv:#2}

\bibitem[\protect\citeauthoryear{Atkey and Sannella}{Atkey and
  Sannella}{2015}]%
        {Contemplate}
\bibfield{author}{\bibinfo{person}{Robert Atkey} {and} \bibinfo{person}{Donald
  Sannella}.} \bibinfo{year}{2015}\natexlab{}.
\newblock \showarticletitle{ThreadSafe: Static Analysis for Java Concurrency}.
\newblock \bibinfo{journal}{\emph{{ECEASST}}}  \bibinfo{volume}{72}
  (\bibinfo{year}{2015}).
\newblock


\bibitem[\protect\citeauthoryear{Blackshear, Gorogiannis, O'Hearn, and
  Sergey}{Blackshear et~al\mbox{.}}{2018}]%
        {Blackshear-al:OOPSLA18}
\bibfield{author}{\bibinfo{person}{Sam Blackshear}, \bibinfo{person}{Nikos
  Gorogiannis}, \bibinfo{person}{Peter~W. O'Hearn}, {and} \bibinfo{person}{Ilya
  Sergey}.} \bibinfo{year}{2018}\natexlab{}.
\newblock \showarticletitle{RacerD: compositional static race detection}.
\newblock \bibinfo{journal}{\emph{{PACMPL}}} \bibinfo{volume}{2},
  \bibinfo{number}{{OOPSLA}} (\bibinfo{year}{2018}),
  \bibinfo{pages}{144:1--144:28}.
\newblock


\bibitem[\protect\citeauthoryear{Brookes}{Brookes}{2007}]%
        {Brookes:TCS07}
\bibfield{author}{\bibinfo{person}{Stephen Brookes}.}
  \bibinfo{year}{2007}\natexlab{}.
\newblock \showarticletitle{A semantics for concurrent separation logic}.
\newblock \bibinfo{journal}{\emph{Th. Comp. Sci.}} \bibinfo{volume}{375},
  \bibinfo{number}{1-3} (\bibinfo{year}{2007}).
\newblock


\bibitem[\protect\citeauthoryear{Cadar and Sen}{Cadar and Sen}{2013}]%
        {CadarS13}
\bibfield{author}{\bibinfo{person}{Cristian Cadar} {and}
  \bibinfo{person}{Koushik Sen}.} \bibinfo{year}{2013}\natexlab{}.
\newblock \showarticletitle{Symbolic execution for software testing: three
  decades later}.
\newblock \bibinfo{journal}{\emph{Commun. {ACM}}} \bibinfo{volume}{56},
  \bibinfo{number}{2} (\bibinfo{year}{2013}), \bibinfo{pages}{82--90}.
\newblock


\bibitem[\protect\citeauthoryear{Chou}{Chou}{2014}]%
        {AndyChou14}
\bibfield{author}{\bibinfo{person}{Andy Chou}.}
  \bibinfo{year}{2014}\natexlab{}.
\newblock \bibinfo{title}{From the Trenches: Static Analysis in Industry}.
\newblock   (\bibinfo{year}{2014}).
\newblock
\newblock
\shownote{Invited keynote talk at POPL'14. Available at
  \url{https://popl.mpi-sws.org/2014/andy.pdf}.}


\bibitem[\protect\citeauthoryear{Christakis, M{\"{u}}ller, and
  W{\"{u}}stholz}{Christakis et~al\mbox{.}}{2015}]%
        {Christakis-al:VMCAI15}
\bibfield{author}{\bibinfo{person}{Maria Christakis}, \bibinfo{person}{Peter
  M{\"{u}}ller}, {and} \bibinfo{person}{Valentin W{\"{u}}stholz}.}
  \bibinfo{year}{2015}\natexlab{}.
\newblock \showarticletitle{An Experimental Evaluation of Deliberate
  Unsoundness in a Static Program Analyzer}. In
  \bibinfo{booktitle}{\emph{VMCAI}} \emph{(\bibinfo{series}{LNCS})},
  Vol.~\bibinfo{volume}{8931}. \bibinfo{publisher}{Springer},
  \bibinfo{pages}{336--354}.
\newblock


\bibitem[\protect\citeauthoryear{Clarke and Drossopoulou}{Clarke and
  Drossopoulou}{2002}]%
        {Clarke-Drossopoulou:OOPSLA02}
\bibfield{author}{\bibinfo{person}{David~G. Clarke} {and}
  \bibinfo{person}{Sophia Drossopoulou}.} \bibinfo{year}{2002}\natexlab{}.
\newblock \showarticletitle{Ownership, encapsulation and the disjointness of
  type and effect}. In \bibinfo{booktitle}{\emph{OOPSLA}}.
  \bibinfo{publisher}{{ACM}}, \bibinfo{pages}{292--310}.
\newblock


\bibitem[\protect\citeauthoryear{Cousot}{Cousot}{1978}]%
        {Cousot:PhD78}
\bibfield{author}{\bibinfo{person}{Patrick Cousot}.}
  \bibinfo{year}{1978}\natexlab{}.
\newblock \emph{\bibinfo{title}{M{\'{e}}thodes it{\'{e}}ratives de construction
  et d'approximation de points fixes d'op{\'{e}}rateurs monotones sur un
  treillis, analyse s{\'{e}}mantique des programmes}}.
\newblock \bibinfo{thesistype}{Ph.D. Dissertation}.
  \bibinfo{school}{Universit\'{e} Scientifique et M\'{e}dicale de Grenoble}.
\newblock


\bibitem[\protect\citeauthoryear{Cousot and Cousot}{Cousot and Cousot}{1977}]%
        {Cousot-Cousot:POPL77}
\bibfield{author}{\bibinfo{person}{Patrick Cousot} {and}
  \bibinfo{person}{Radhia Cousot}.} \bibinfo{year}{1977}\natexlab{}.
\newblock \showarticletitle{Abstract Interpretation: {A} Unified Lattice Model
  for Static Analysis of Programs by Construction or Approximation of
  Fixpoints}. In \bibinfo{booktitle}{\emph{POPL}}. \bibinfo{publisher}{{ACM}},
  \bibinfo{pages}{238--252}.
\newblock


\bibitem[\protect\citeauthoryear{Cousot and Cousot}{Cousot and Cousot}{1979}]%
        {Cousot-Cousot:POPL79}
\bibfield{author}{\bibinfo{person}{Patrick Cousot} {and}
  \bibinfo{person}{Radhia Cousot}.} \bibinfo{year}{1979}\natexlab{}.
\newblock \showarticletitle{Systematic Design of Program Analysis Frameworks}.
  In \bibinfo{booktitle}{\emph{POPL}}. \bibinfo{publisher}{{ACM} Press},
  \bibinfo{pages}{269--282}.
\newblock


\bibitem[\protect\citeauthoryear{Cousot and Cousot}{Cousot and Cousot}{1992}]%
        {Cousot-Cousot:JLC92}
\bibfield{author}{\bibinfo{person}{Patrick Cousot} {and}
  \bibinfo{person}{Radhia Cousot}.} \bibinfo{year}{1992}\natexlab{}.
\newblock \showarticletitle{Abstract Interpretation Frameworks}.
\newblock \bibinfo{journal}{\emph{J. Log. Comput.}} \bibinfo{volume}{2},
  \bibinfo{number}{4} (\bibinfo{year}{1992}), \bibinfo{pages}{511--547}.
\newblock


\bibitem[\protect\citeauthoryear{Flanagan and Freund}{Flanagan and
  Freund}{2009}]%
        {Flanagan-Freund:PLDI09}
\bibfield{author}{\bibinfo{person}{Cormac Flanagan} {and}
  \bibinfo{person}{Stephen~N. Freund}.} \bibinfo{year}{2009}\natexlab{}.
\newblock \showarticletitle{{FastTrack: efficient and precise dynamic race
  detection}}. In \bibinfo{booktitle}{\emph{PLDI}}. \bibinfo{publisher}{{ACM}},
  \bibinfo{pages}{121--133}.
\newblock


\bibitem[\protect\citeauthoryear{Flanagan, Freund, and Yi}{Flanagan
  et~al\mbox{.}}{2008}]%
        {Flanagan-al:PLDI08}
\bibfield{author}{\bibinfo{person}{Cormac Flanagan},
  \bibinfo{person}{Stephen~N. Freund}, {and} \bibinfo{person}{Jaeheon Yi}.}
  \bibinfo{year}{2008}\natexlab{}.
\newblock \showarticletitle{Velodrome: a sound and complete dynamic atomicity
  checker for multithreaded programs}. In \bibinfo{booktitle}{\emph{PLDI}}.
  \bibinfo{publisher}{{ACM}}, \bibinfo{pages}{293--303}.
\newblock


\bibitem[\protect\citeauthoryear{Giacobazzi, Logozzo, and Ranzato}{Giacobazzi
  et~al\mbox{.}}{2015}]%
        {Giacobazzi-al:POPL15}
\bibfield{author}{\bibinfo{person}{Roberto Giacobazzi},
  \bibinfo{person}{Francesco Logozzo}, {and} \bibinfo{person}{Francesco
  Ranzato}.} \bibinfo{year}{2015}\natexlab{}.
\newblock \showarticletitle{Analyzing Program Analyses}. In
  \bibinfo{booktitle}{\emph{POPL}}. \bibinfo{publisher}{{ACM}},
  \bibinfo{pages}{261--273}.
\newblock


\bibitem[\protect\citeauthoryear{Goetz, Peierls, Bloch, Bowbeer, Holmes, and
  Lea}{Goetz et~al\mbox{.}}{2006}]%
        {jcip}
\bibfield{author}{\bibinfo{person}{Brian Goetz}, \bibinfo{person}{Tim Peierls},
  \bibinfo{person}{Joshua Bloch}, \bibinfo{person}{Joseph Bowbeer},
  \bibinfo{person}{David Holmes}, {and} \bibinfo{person}{Doug Lea}.}
  \bibinfo{year}{2006}\natexlab{}.
\newblock \bibinfo{booktitle}{\emph{Java Concurrency in Practice}}.
\newblock \bibinfo{publisher}{Addison-Wesley}.
\newblock
\showISBNx{978-0-321-34960-6}


\bibitem[\protect\citeauthoryear{Herlihy and Shavit}{Herlihy and
  Shavit}{2008}]%
        {Herlihy-Shavit:08}
\bibfield{author}{\bibinfo{person}{Maurice Herlihy} {and} \bibinfo{person}{Nir
  Shavit}.} \bibinfo{year}{2008}\natexlab{}.
\newblock \bibinfo{booktitle}{\emph{The art of multiprocessor programming}}.
\newblock \bibinfo{publisher}{M. Kaufmann}.
\newblock


\bibitem[\protect\citeauthoryear{Igarashi, Pierce, and Wadler}{Igarashi
  et~al\mbox{.}}{2001}]%
        {Igarashi-al:TOPLAS01}
\bibfield{author}{\bibinfo{person}{Atsushi Igarashi},
  \bibinfo{person}{Benjamin~C. Pierce}, {and} \bibinfo{person}{Philip Wadler}.}
  \bibinfo{year}{2001}\natexlab{}.
\newblock \showarticletitle{Featherweight Java: a minimal core calculus for
  Java and {GJ}}.
\newblock \bibinfo{journal}{\emph{{ACM} Trans. Program. Lang. Syst.}}
  \bibinfo{volume}{23}, \bibinfo{number}{3} (\bibinfo{year}{2001}),
  \bibinfo{pages}{396--450}.
\newblock


\bibitem[\protect\citeauthoryear{Kang, Hur, Lahav, Vafeiadis, and Dreyer}{Kang
  et~al\mbox{.}}{2017}]%
        {Kang-al:POPL17}
\bibfield{author}{\bibinfo{person}{Jeehoon Kang}, \bibinfo{person}{Chung{-}Kil
  Hur}, \bibinfo{person}{Ori Lahav}, \bibinfo{person}{Viktor Vafeiadis}, {and}
  \bibinfo{person}{Derek Dreyer}.} \bibinfo{year}{2017}\natexlab{}.
\newblock \showarticletitle{A promising semantics for relaxed-memory
  concurrency}. In \bibinfo{booktitle}{\emph{POPL}}.
  \bibinfo{publisher}{{ACM}}, \bibinfo{pages}{175--189}.
\newblock


\bibitem[\protect\citeauthoryear{Lamport}{Lamport}{1979}]%
        {Lamport:TC79}
\bibfield{author}{\bibinfo{person}{Leslie Lamport}.}
  \bibinfo{year}{1979}\natexlab{}.
\newblock \showarticletitle{How to Make a Multiprocessor Computer That
  Correctly Executes Multiprocess Programs}.
\newblock \bibinfo{journal}{\emph{{IEEE} Trans. Computers}}
  \bibinfo{volume}{28}, \bibinfo{number}{9} (\bibinfo{year}{1979}),
  \bibinfo{pages}{690--691}.
\newblock


\bibitem[\protect\citeauthoryear{Livshits, Sridharan, Smaragdakis,
  Lhot{\'{a}}k, Amaral, Chang, Guyer, Khedker, M{\o}ller, and
  Vardoulakis}{Livshits et~al\mbox{.}}{2015}]%
        {Livshits-al:CACM15}
\bibfield{author}{\bibinfo{person}{Benjamin Livshits}, \bibinfo{person}{Manu
  Sridharan}, \bibinfo{person}{Yannis Smaragdakis}, \bibinfo{person}{Ondrej
  Lhot{\'{a}}k}, \bibinfo{person}{Jos{\'{e}}~Nelson Amaral},
  \bibinfo{person}{Bor{-}Yuh~Evan Chang}, \bibinfo{person}{Samuel~Z. Guyer},
  \bibinfo{person}{Uday~P. Khedker}, \bibinfo{person}{Anders M{\o}ller}, {and}
  \bibinfo{person}{Dimitrios Vardoulakis}.} \bibinfo{year}{2015}\natexlab{}.
\newblock \showarticletitle{In defense of soundiness: a manifesto}.
\newblock \bibinfo{journal}{\emph{Commun. {ACM}}} \bibinfo{volume}{58},
  \bibinfo{number}{2} (\bibinfo{year}{2015}), \bibinfo{pages}{44--46}.
\newblock


\bibitem[\protect\citeauthoryear{Mansky, Peng, Zdancewic, and Devietti}{Mansky
  et~al\mbox{.}}{2017}]%
        {Mansky-al:CPP17}
\bibfield{author}{\bibinfo{person}{William Mansky}, \bibinfo{person}{Yuanfeng
  Peng}, \bibinfo{person}{Steve Zdancewic}, {and} \bibinfo{person}{Joseph
  Devietti}.} \bibinfo{year}{2017}\natexlab{}.
\newblock \showarticletitle{Verifying dynamic race detection}. In
  \bibinfo{booktitle}{\emph{CPP}}. \bibinfo{publisher}{{ACM}},
  \bibinfo{pages}{151--163}.
\newblock


\bibitem[\protect\citeauthoryear{Naik and Aiken}{Naik and Aiken}{2007}]%
        {Naik-Aiken:POPL07}
\bibfield{author}{\bibinfo{person}{Mayur Naik} {and} \bibinfo{person}{Alex
  Aiken}.} \bibinfo{year}{2007}\natexlab{}.
\newblock \showarticletitle{Conditional must not aliasing for static race
  detection}. In \bibinfo{booktitle}{\emph{POPL}}. \bibinfo{publisher}{{ACM}},
  \bibinfo{pages}{327--338}.
\newblock


\bibitem[\protect\citeauthoryear{Naik, Aiken, and Whaley}{Naik
  et~al\mbox{.}}{2006}]%
        {Naik-al:PLDI06}
\bibfield{author}{\bibinfo{person}{Mayur Naik}, \bibinfo{person}{Alex Aiken},
  {and} \bibinfo{person}{John Whaley}.} \bibinfo{year}{2006}\natexlab{}.
\newblock \showarticletitle{{Effective static race detection for Java}}. In
  \bibinfo{booktitle}{\emph{PLDI}}. \bibinfo{publisher}{{ACM}},
  \bibinfo{pages}{308--319}.
\newblock


\bibitem[\protect\citeauthoryear{Oh, Lee, Heo, Yang, and Yi}{Oh
  et~al\mbox{.}}{2014}]%
        {Oh-al:PLDI14}
\bibfield{author}{\bibinfo{person}{Hakjoo Oh}, \bibinfo{person}{Wonchan Lee},
  \bibinfo{person}{Kihong Heo}, \bibinfo{person}{Hongseok Yang}, {and}
  \bibinfo{person}{Kwangkeun Yi}.} \bibinfo{year}{2014}\natexlab{}.
\newblock \showarticletitle{Selective context-sensitivity guided by impact
  pre-analysis}. In \bibinfo{booktitle}{\emph{PLDI}}.
  \bibinfo{publisher}{{ACM}}, \bibinfo{pages}{475--484}.
\newblock


\bibitem[\protect\citeauthoryear{Raghothaman, Kulkarni, Heo, and
  Naik}{Raghothaman et~al\mbox{.}}{2018}]%
        {Raghothaman-al:PLDI18}
\bibfield{author}{\bibinfo{person}{Mukund Raghothaman},
  \bibinfo{person}{Sulekha Kulkarni}, \bibinfo{person}{Kihong Heo}, {and}
  \bibinfo{person}{Mayur Naik}.} \bibinfo{year}{2018}\natexlab{}.
\newblock \showarticletitle{{Interactive Program Reasoning using Bayesian
  Inference}}. In \bibinfo{booktitle}{\emph{PLDI}}. \bibinfo{publisher}{{ACM}},
  \bibinfo{pages}{722--735}.
\newblock


\bibitem[\protect\citeauthoryear{Ranzato}{Ranzato}{2013}]%
        {Ranzato:VMCAI13}
\bibfield{author}{\bibinfo{person}{Francesco Ranzato}.}
  \bibinfo{year}{2013}\natexlab{}.
\newblock \showarticletitle{Complete Abstractions Everywhere}. In
  \bibinfo{booktitle}{\emph{VMCAI}} \emph{(\bibinfo{series}{LNCS})},
  Vol.~\bibinfo{volume}{7737}. \bibinfo{publisher}{Springer},
  \bibinfo{pages}{15--26}.
\newblock


\bibitem[\protect\citeauthoryear{Sadowski, Yi, Knowles, and Flanagan}{Sadowski
  et~al\mbox{.}}{2008}]%
        {Sadowski-al:WMM08}
\bibfield{author}{\bibinfo{person}{Caitlin Sadowski}, \bibinfo{person}{Jaeheon
  Yi}, \bibinfo{person}{Kenneth Knowles}, {and} \bibinfo{person}{Cormac
  Flanagan}.} \bibinfo{year}{2008}\natexlab{}.
\newblock \showarticletitle{{Proving correctness of a dynamic atomicity
  analysis in Coq}}. In \bibinfo{booktitle}{\emph{Workshop on Mechanizing
  Metatheory}}.
\newblock


\bibitem[\protect\citeauthoryear{Serebryany and Iskhodzhanov}{Serebryany and
  Iskhodzhanov}{2009}]%
        {TSAN}
\bibfield{author}{\bibinfo{person}{Konstantin Serebryany} {and}
  \bibinfo{person}{Timur Iskhodzhanov}.} \bibinfo{year}{2009}\natexlab{}.
\newblock \showarticletitle{{ThreadSanitizer}: data race detection in
  practice}.
\newblock \bibinfo{journal}{\emph{Proceedings of the Workshop on Binary
  Instrumentation and Applications}}, \bibinfo{pages}{62--71}.
\newblock


\bibitem[\protect\citeauthoryear{Tennent}{Tennent}{1977}]%
        {Tennent77}
\bibfield{author}{\bibinfo{person}{Robert~D. Tennent}.}
  \bibinfo{year}{1977}\natexlab{}.
\newblock \showarticletitle{Language Design Methods Based on Semantic
  Principles}.
\newblock \bibinfo{journal}{\emph{Acta Inf.}}  \bibinfo{volume}{8}
  (\bibinfo{year}{1977}), \bibinfo{pages}{97--112}.
\newblock


\bibitem[\protect\citeauthoryear{Turon, Thamsborg, Ahmed, Birkedal, and
  Dreyer}{Turon et~al\mbox{.}}{2013}]%
        {Turon-al:POPL13}
\bibfield{author}{\bibinfo{person}{Aaron~Joseph Turon}, \bibinfo{person}{Jacob
  Thamsborg}, \bibinfo{person}{Amal Ahmed}, \bibinfo{person}{Lars Birkedal},
  {and} \bibinfo{person}{Derek Dreyer}.} \bibinfo{year}{2013}\natexlab{}.
\newblock \showarticletitle{Logical relations for fine-grained concurrency}. In
  \bibinfo{booktitle}{\emph{POPL}}. \bibinfo{publisher}{ACM},
  \bibinfo{pages}{343--356}.
\newblock


\bibitem[\protect\citeauthoryear{Wilcox, Flanagan, and Freund}{Wilcox
  et~al\mbox{.}}{2018}]%
        {Wilcox-al:PPoPP18}
\bibfield{author}{\bibinfo{person}{James~R. Wilcox}, \bibinfo{person}{Cormac
  Flanagan}, {and} \bibinfo{person}{Stephen~N. Freund}.}
  \bibinfo{year}{2018}\natexlab{}.
\newblock \showarticletitle{{VerifiedFT: a verified, high-performance precise
  dynamic race detector}}. \bibinfo{publisher}{{ACM}},
  \bibinfo{pages}{354--367}.
\newblock


\bibitem[\protect\citeauthoryear{Yang and O'Hearn}{Yang and O'Hearn}{2002}]%
        {Hongseok02}
\bibfield{author}{\bibinfo{person}{Hongseok Yang} {and} \bibinfo{person}{Peter
  O'Hearn}.} \bibinfo{year}{2002}\natexlab{}.
\newblock \showarticletitle{A Semantic Basis for Local Reasoning}. In
  \bibinfo{booktitle}{\emph{Foundations of Software Science and Computation
  Structures}}. \bibinfo{publisher}{Springer Berlin Heidelberg},
  \bibinfo{pages}{402--416}.
\newblock


\end{thebibliography}
